\newif\ifarxiv
\arxivtrue   
 
\ifarxiv
  \documentclass[11pt,a4paper]{article}
\else
  \documentclass[12pt,a4paper,oneside,onecolumn]{book}
\fi

\makeatletter
\@ifundefined{ifarxiv}{\newif\ifarxiv}{}
\makeatother

\usepackage{libertine}
\usepackage{courier}
\usepackage[utf8]{inputenc} 

\ifarxiv
\usepackage[T1]{fontenc}          
\usepackage[margin=1in]{geometry} 
\usepackage{microtype}            
\fi

\usepackage{float}
\usepackage{titlesec}
\usepackage[longnamesfirst,square,numbers,comma,sort&compress]{natbib}

\usepackage{algpseudocode}

\usepackage{xcolor}
\usepackage{amsmath}
\ifarxiv
\usepackage{amssymb}
\usepackage{amsfonts}
\usepackage{type1cm}
\else
\usepackage{amssymb}
\usepackage{amsfonts}
\fi
\usepackage{mathtools}
\usepackage{graphicx}
\usepackage{url}
\usepackage{hyperref} 

\ifarxiv
\hypersetup{
  colorlinks = true,
  linkcolor  = blue!55!black,
  citecolor  = blue!55!black,
  urlcolor   = blue!55!black,
}
\fi

\usepackage{amsthm}
\usepackage{nccmath}

\ifarxiv
\newtheorem{theorem}{Theorem}[section]
\else
\newtheorem{theorem}{Theorem}[chapter]
\fi
\newtheorem{lemma}[theorem]{Lemma}
\newtheorem{proposition}[theorem]{Proposition}
\newtheorem{definition}[theorem]{Definition}
\newtheorem{corollary}[theorem]{Corollary}

\newtheorem{claim}{Claim}

\newtheorem*{theoremStar}{Theorem}

\newtheorem{remark}{Remark}

\usepackage{wrapfig}
\usepackage{relsize}
\usepackage{todonotes}
\usepackage{circledsteps} 
\usepackage[linesnumbered,ruled,vlined]{algorithm2e}
\usepackage{refcount}
\usepackage{thm-restate}

\newboolean{showtodo}
\setboolean{showtodo}{false} 

\NewDocumentCommand{\todosimple}{mm}{%
  \ifthenelse{\boolean{showtodo}}%
   {{\scriptsize \textbf{{#1} says: \color{red} {#2}}}}%
    {}%
}

\NewDocumentCommand{\gytodo}{m}{%
  \ifthenelse{\boolean{showtodo}}%
    {\todo[color=yellow]{Gad: #1}}%
    {}%
}

\NewDocumentCommand{\dftodo}{m}{%
  \ifthenelse{\boolean{showtodo}}%
    {\todo[color=cyan]{Dror: #1}}%
    {}%
}

\NewDocumentCommand{\estodo}{m}{%
  \ifthenelse{\boolean{showtodo}}%
    {\todo[color=green]{Etay: #1}}%
    {}%
}

\usepackage{tikz}
\usetikzlibrary{automata,positioning,shapes,fit,calc}

\newcommand{\Sumnk}[2]{(\sum_{j=1}^{#1} X_j) \bmod{#2}}

\newcommand{\true}{\emph{true }}
\newcommand{\false}{\emph{false }}

\newcommand{\x}{\vec{x}}
\newcommand{\y}{\vec{y}}
\newcommand{\z}{\vec{z}}
\newcommand{\vsigma}{\vec{\sigma}}

\newcommand{\PH}{\mathrm{PH}}
\renewcommand{\P}{\mathrm{P}}
\newcommand{\odd}{\textit{odd }}
\newcommand{\even}{\textit{even }}

\newcommand{\parity}{\mathlarger{\mathlarger{\boldsymbol{\oplus}}}}

\newcommand{\TodaOne}{Toda's Reduction }

\let\boldoplus\oplus
\renewcommand{\oplus}{\mathbin{\mbox{\unboldmath$\boldoplus$}}}

\newcommand{\modq}{\mbox{\textcircled{q}}}
\newcommand{\modqj}[1]{\mbox{\textcircled{q}}_{#1}}

\ifarxiv
\else
\let\bookchapter\chapter
\let\booksection\section
\let\booksubsection\subsection
\let\booksubsubsection\subsubsection
\let\bookparagraph\paragraph

\let\section\bookchapter
\let\subsection\booksection
\let\subsubsection\booksubsection
\let\paragraph\booksubsubsection
\let\subparagraph\bookparagraph
\fi

\ifarxiv\else
\newenvironment{abstract}{\chapter*{Abstract}}{}
\fi

\ifarxiv
\newcommand{\mainbibstyle}{alpha}  
\else
\newcommand{\mainbibstyle}{ACM-Reference-Format}
\fi

\DeclareMathOperator*{\avg}{avg} 

\title{Algorithmic Aspects of Toda’s Theorem}
\ifarxiv
\author{
  Dror Fried\thanks{The Open University of Israel, Ra'anana, Israel. \texttt{dfried@openu.ac.il}}
  \and
  Gad E.\ Yaron\thanks{Corresponding author. The Open University of Israel, Ra'anana, Israel, and The Hebrew University of Jerusalem, Jerusalem, Israel. \texttt{gadelaz.yaron@mail.huji.ac.il}}
}
\else
\author{Your Name}
\fi

\begin{document}

\maketitle

\begin{abstract}
Toda's Theorem is a fundamental result in computational complexity theory, whose proof is based on a reduction from a QBF problem with a constant number of quantifiers to a model counting problem. The recent progress in model counting tools raises the question of whether this reduction, henceforth called Toda's reduction, can be utilized to construct a practical QBF solver. This question follows a line of research that revisits theoretical results from an algorithmic aspect, and thus brings new theoretical and engineering challenges. For Toda's reduction these challenges arise mainly because the reduction is purely theoretical and based on ideas that are entirely orthogonal to the search-space approach used by current QBF solvers.
In this work, we address this question by transforming Toda's reduction into a concrete probabilistic QBF solver that uses model counting as an oracle. A naive implementation is hopeless due to a massive formula blow-up. Therefore we next analyze and identify three main factors that drive the blow-up. While we present solutions that overcome some of the factors, we also discuss the limitations of some, and show that one of the factors, the union bound factor, largely overlooked in the literature, is in fact dominant and in some cases unavoidable. We then show how, for some cases, even this factor can be avoided, and report our preliminary results on a prototype implementation.
\end{abstract}

\section{Introduction}\label{sec:intro}

One of the most fundamental results in computational complexity, taught in many advanced classes, is Toda's Theorem~\cite{Toda91} which shows that the polynomial hierarchy is contained in $P^{\# SAT}$. Namely, for any constant $d>0$, the satisfiability of every quantified Boolean formula (QBF) $F$ with $d$ quantifier alternations can be decided by a polynomial-time algorithm (polynomial in the size of $F$), that makes oracle calls to a \emph{model counter}, a tool that counts the satisfying assignments of a given quantifier-free propositional formula.

Toda's theorem is constructive in nature and, since many QBF instances arising in practice have only a small number of alternations, it could conceptually serve as a basis for a QBF solver. Despite this potential, Toda's theorem has never been proposed or analyzed as a concrete algorithm for QBF solving. One reason may be its theoretical nature, which like many other complexity theoretic results, maintains an elegant clarity over an algorithmic description.
A more practical reason may be the severe blow-up in the size of the resulting formula given to the model counter. Even a simple QBF formula with only three quantifiers, fifteen variables, and twenty clauses, would result in a quantifier-free formula with over three hundred million clauses for the model counter to solve (see Section~\ref{sec:efficiency}).

Recently, however, the importance of bridging the gap between theory and practice has gained acknowledgment, with more and more works that explore algorithmic aspects of purely theoretical problems. While some of these investigations showed critical limitations (e.g.~\cite{NewmanVardi2020}), others turned successful. One notable success story is the transformation of the theoretical work on uniform generation of NP witnesses~\cite{BGPertrank1998} into efficient nearly uniform generators, and approximate model-counting tools~\cite{ChakrabortyMV13,ChakrabortyMV13-CP,MeelA20, DelannoyM22}. Another example is a recent work of making IP=PSPACE practical~\cite{CouillardCEM23}.

Moreover, recent years have shown a substantial progress in the quality of model counters, with exact model counters, such as~\cite{SharmaRSM19, SharpSAT_TD, LagniezM17}, nowadays handling formulas with hundreds of thousands of variables and clauses. When exploring Toda's theorem from algorithmic aspects, this progress allows us to shift the focus from improving model counters to refining the formula construction in Toda's reduction. Since the Toda-based approach is orthogonal to current search-based QBF solvers, a successful Toda-based algorithm for solving QBF could both provide new theoretical insights, and solve instances in which other solvers struggle.

Consequently, in this work we revisit Toda's Theorem from an algorithmic perspective, analyzing it not merely as a complexity-theoretic result but as a potential method for solving QBF formulas. We transform the reduction into a concrete algorithm, study its properties, and address the barriers that it presents. Our goal is to assess, from an algorithmic viewpoint, under what conditions a Toda-based approach can be considered an efficient algorithm for QBF solving.

We focus on the first part of Toda's theorem, which shows that $PH$ is contained in $BPP^{\oplus P}$. This part, which we call \emph{Toda's reduction}, describes how a given QBF formula $F$ with $d$ alternations can be transformed (in time exponential in $d$ but polynomial in the size of the formula when $d$ is fixed) into a quantifier-free Boolean formula $F_0$, such that with high probability, $F$ is \true if and only if $F_0$ has an \emph{odd} number of satisfying assignments. Once $F_0$ is established, a model counter is used to determine the parity of $F_0$.
The reduction proceeds in $d$ steps, each removing one quantifier starting from the innermost quantifier. Each step applies the Valiant-Vazirani hash function~\cite{ValiantVazirani}, which removes a quantifier and ensures, with low probability, that the resulting formula has an odd number of solutions if $F$ is true. To boost this probability, each step is repeated many times on fresh copies of the formula, which are then combined. Overall, an algorithm based on this reduction takes a QBF formula and an error bound $\epsilon > 0$, and returns a yes/no answer that is correct with probability at least $1 - \epsilon$.

We first convert Toda's reduction to a concrete algorithm, by following the proof description in~\cite{arora2006computational}. Then, we refine the analysis of the exact number of repetitions that the construction requires, and the size of the resulting formula. This is in contrast to the more asymptotic bounds that appear in the literature.
As expected, the huge number of repetitions in every step of the reduction ends in a quantifier-free Boolean formula that is enormously larger than the input formula. The Valiant-Vazirani hash function seems to be the obvious reason for these repetitions. Our analysis points, however, to not one, but three factors, that play an important role in this blow-up. The first is the allocation of the error probabilities along the steps, which is based on a geometric series. While this allocation provides a clean theoretical analysis, it leaves the deeper steps of the reduction with only a small share of the error budget, which increases their number of repetitions. The second factor is the aforementioned Valiant-Vazirani hash function. The third and final factor stems from union bound considerations. Properly addressing these three factors can reduce the number of repetitions, and the formula size.

To handle the first factor, we perform a finer analysis that better characterizes the error behaviors within the different steps. In addition, we balance the error allocation among these steps. We then show that for formulas with a small number of quantifiers, which are the ones within reach of a Toda-based algorithm, our technique indeed decreases the number of overall repetitions. To handle the second factor, we need to overcome the Valiant-Vazirani obstacle. For that we rely on a known observation, that for the construction to work, it suffices to capture an odd number of solutions, rather than a unique solution as in the Valiant-Vazirani lemma. We describe two methods to exploit this observation. The first one, \emph{modular addition}, allows us to amplify the success probability of Valiant-Vazirani arbitrarily close to $1/2$ without repeating the formula, at only a reasonable cost in the size of the hash function. Our second method is to replace the Valiant-Vazirani hash function with a hash function that Naik, Regan and Sivakumar~\cite{NRS95} designed for this task. We follow their construction, and our contribution is only in the choice of its parameters, and in its encoding as a Boolean formula. We also prove that these methods are essentially the best possible, as no hash function of this kind, however large, can succeed with probability noticeably better than $1/2$.

Finally, we analyze the union bound factor. This factor received little attention in the literature. Nevertheless, we show that once the first two factors are handled, it becomes the dominant one, and that for some formulas it cannot be reduced, so that the blow-up is unavoidable. On the other hand, we identify the parameter of the input formula that actually determines this factor, which can be much smaller than what the textbook analysis suggests, and we give a simple syntactic bound on it. For formulas in which this bound is small, the number of repetitions drops dramatically, suggesting a potential practical efficiency for our Toda-based algorithm.

We also obtain two further results. We examine the generalization of the reduction from counting modulo $2$ to counting modulo an arbitrary $q$, and show that no modulus improves on $2$. We also report preliminary results of a prototype implementation that uses a state-of-the-art model counter as an oracle. Our implementation is not yet competitive with mature QBF solvers, and we discuss the steps that are required in order to assess whether a QBF solver based on Toda's reduction is an effective direction to pursue.

\subsection{Thesis Contributions}\label{sec:cont}

All in all, the contributions of this thesis are the following.

\begin{enumerate}
    \item We convert Toda's reduction into an explicit algorithm, and derive exact values, rather than asymptotic bounds, for the error budget of every step, the number of repetitions that each step requires, and the size of the resulting formula. This analysis identifies the three factors that cause the blow-up (Sections~\ref{sec:todaBrief} and~\ref{sec:alganal}).
    \item We give an exact analysis of how errors propagate through the reduction, and a balanced allocation of the error budget among the steps that, for formulas with four or fewer quantifiers, yields smaller formulas than the textbook allocation (Section~\ref{sec:innerComp}).
    \item We replace the Valiant-Vazirani hash function with two hash functions whose success probability is a constant close to $1/2$, namely modular addition, which is novel in this context, and the construction of Naik, Regan and Sivakumar~\cite{NRS95}, which we encode as a Boolean formula. We prove that no hash function of this kind can succeed with probability noticeably better than $1/2$ (Section~\ref{sec:replacing-vv}).
    \item We identify the parameter $C$ of the input formula that determines the union bound factor, show that with respect to $C$ this factor is tight and cannot be improved, and give a simple syntactic bound on $C$ (Section~\ref{sec:union-bound}).
    \item We show that replacing counting modulo $2$ by counting modulo any other $q$ only increases the size of the formulas (Section~\ref{sec:modp}).
    \item We describe the engineering choices behind turning the reduction into a working tool, and report preliminary results of a prototype implementation (Section~\ref{sec:implementation}).
\end{enumerate}

\subsection{Organization}\label{sec:organization}

The paper goes as follows. In Section~\ref{sec:summary} we describe our results in detail. We give preliminaries in Section~\ref{sec:prelim}. In Section~\ref{sec:todaBrief} we describe Toda's reduction and algorithm, and describe our algorithm analysis in Section~\ref{sec:alganal}. In Section~\ref{sec:innerComp} we discuss the error allocation, while in Section~\ref{sec:replacing-vv} we present alternatives to Valiant-Vazirani. Then in Section~\ref{sec:union-bound} we discuss the union-bound factor. In Section~\ref{sec:modp} we ask what changes if counting modulo $2$ is replaced by counting modulo an arbitrary $q$. In Section~\ref{sec:implementation} we discuss the engineering challenges of building a tool based on the reduction, and report our results. We conclude in Section~\ref{sec:discussion}. Full proofs appear in the appendix.

\section{Summary of Results}\label{sec:summary}

In this section we describe our results in more detail. As explained in the introduction, the blow-up of the reduction comes from the number of times that each step repeats the formula, and this number depends on three factors. We first describe the analysis that identifies these factors, and then our results on each of them.

\subsection{Analysis of the Reduction}\label{sec:summary-analysis}

Our first set of results is an exact analysis of the reduction. The reduction appears in the literature as a proof and not as an algorithm, so we first write each of its steps explicitly, as pseudocode (Section~\ref{sec:todaBrief}). We then compute exact values, and not asymptotic bounds, for the parameters of the algorithm (Section~\ref{sec:alganal}). To state the main outcome of this analysis we first add some necessary notation. Consider the step of the reduction that eliminates a quantifier over $n$ variables, and let $\epsilon_i$ be the error that this step is allowed to contribute to the overall error bound $\epsilon$. Let $k_i$ be the number of times that this step repeats the formula, and let $p$ be the probability that a single application of the hash function succeeds, namely leaves an odd number of solutions. For the Valiant-Vazirani hash function the guarantee is $p=1/8n$. Finally, the formula that the step produces still contains the variables of the quantifiers that were not eliminated yet, which we call the \emph{outer variables}, and we denote their number by $s$. The hash function is drawn once, but the step must be correct for every assignment to the outer variables, and a union bound over their $2^s$ assignments gives the following.

\begin{proposition}[Restatement of Proposition~\ref{clm:repetitionUNbound}]
Setting $k_i \geq \log_{1-p}(\epsilon_i/2^s)$ ensures that the probability for error at step $i$ is at most $\epsilon_i$.
\end{proposition}

Since $\log_{1-p}(x)$ is close to $\ln(1/x)/p$ when $p$ is small, the number of repetitions is roughly $(\ln(1/\epsilon_i) + s\ln 2)/p$. In this formula we can see the three factors from the introduction. The number of repetitions grows linearly with $1/p$, which for the Valiant-Vazirani hash function means linearly with $n$, and it grows linearly with the number $s$ of outer variables, which is the union bound factor. It grows only logarithmically with $1/\epsilon_i$, but the textbook allocation makes $\epsilon_i$ exponentially small for the deeper steps, since it halves the budget from step to step. We also compute the size of $F_0$ exactly (Section~\ref{sec:sizeComp}), and show that it grows with the product of the numbers of repetitions of all the steps, so the number of repetitions at each step is what has to be reduced.

\subsubsection{Constructing a Prototype}

To see what can be gained, we also ran a prototype of the algorithm in which we set the parameters as if the three factors were already handled, for example with $p=1/2$ instead of the Valiant-Vazirani guarantee. Such a setting gives no reliability guarantee, but the prototype could handle competitive QBF benchmarks, and even ``solve'' benchmarks that other solvers could not (Section~\ref{sec:efficiency}). The question is therefore whether a similar performance can be reached while keeping the reliability guarantee.

\subsection{Refining the Error Allocation}\label{sec:summary-allocation}

For the error allocation, we do two things. First, we divide the error budget evenly among the steps, which we call the \emph{balanced} allocation, in place of the textbook division (Section~\ref{sec:innerComp}). Second, we analyze how errors propagate through the reduction. One observation of this analysis is that an existential step can only turn a true formula into a false one, and a universal step can only turn a false formula into a true one. From this analysis we derive the exact relation between the per-step budgets and the overall error, both when the success probability of every step is known exactly (Theorem~\ref{clm:innerprobnotnaive}), and in the realistic case where only a lower bound on it is available. Since the relation is exact, the budget that each step gets is the largest possible when all steps get the same budget. For the second case, which is the one that applies to all the hash functions that we use, we show the following, where a step succeeds if it results in a logically equivalent formula.

\begin{theorem}[Restatement of Theorem~\ref{thm:boundedalloc}]
Assume that every step of the reduction succeeds with a probability of at least $q$, where the actual probability may be any value in $[q,1]$, independently at every step. Then setting $\epsilon_i=1-q$ for every $i$, where $q=(1-\epsilon)^{1/\lceil\frac{d}{2}\rceil}$, ensures the overall error bound $\epsilon$ is achieved.
\end{theorem}

We also prove that for formulas with four or fewer quantifiers, every step gets a larger budget than under the textbook division, and hence the resulting formulas are smaller (Proposition~\ref{prop:benefit}).

\subsection{Replacing Valiant-Vazirani}\label{sec:summary-sieves}

For the second factor, both of our methods replace the Valiant-Vazirani hash function by another randomly chosen constraint that is added to the formula as a conjunct, and we use the term \emph{sieve} for any such constraint. As observed in earlier work~\cite{NaorNaor1993,Gupta1998,NRS95}, isolating a unique solution, which is what the Valiant-Vazirani hash function aims at, is more than what is needed. The reduction only requires the number of surviving solutions to be \odd\kern-0.35em, and for this weaker goal a \emph{constant} success probability is possible. Isolation itself, on the other hand, cannot be improved, as Dell et al.~\cite{DellKMW13} proved that no procedure that must commit to its random choices before seeing the formula can isolate a unique solution with probability better than $O(1/n)$. We present two sieves that avoid the factor of $n$ altogether.

The first sieve, which we call \emph{modular addition}, is, as far as we know, novel in this context. It combines $l$ independent draws of a hash function into a single constraint, in such a way that the number of solutions of the combined formula is the sum of the numbers of solutions that the $l$ draws would leave separately. The count of the combined formula is therefore odd exactly when an odd number of the draws succeed. This is the same idea as generating an unbiased bit from a biased one by sampling it many times and taking the parity. Modular addition therefore amplifies the success probability of any weak hash function arbitrarily close to $1/2$, and it duplicates only the hash function and not the formula.

\begin{theorem}[Informal version of Theorem~\ref{thm:main} and Proposition~\ref{prop:MAuse}]
Let $h$ be a randomly drawn hash function that, applied to a satisfiable formula, leaves an \odd number of solutions with probability $p$, where $0<p\leq 1/2$. Then adding to the formula, as a conjunct, the modular addition of $l$ independent draws of the hash function leaves an \odd number of solutions with probability $\frac{1}{2}\left(1-(1-2p)^{l}\right)$, which approaches $1/2$ as $l$ grows.
\end{theorem}

The second sieve is the one that Naik, Regan and Sivakumar~\cite{NRS95} designed for this task, which we encode as a Boolean formula that is asymptotically smaller than the Valiant-Vazirani hash function.

\begin{theorem}[Informal version of Theorem~\ref{thm:parity-isolation}, construction of~\cite{NRS95}]
For every $n$ and every constant $\delta>0$ there is a sieve, encoded as a Boolean formula of size $O(n\log^2 n)$, that for every satisfiable formula over $n$ variables leaves an \odd number of solutions with probability at least $\frac{1}{2}-\delta$.
\end{theorem}

We use $\delta=1/16$, that is a success probability of $7/16$, throughout the paper. For both sieves we compute the effect on the size of the resulting formula.

\subsubsection{A Barrier at One Half}

Both sieves stop a little below $1/2$. We prove that for sieves that only remove solutions from the formula this is unavoidable. No such sieve, however large, makes the number of surviving solutions odd with probability exceeding $1/2$ by more than a doubly exponentially small amount, and this value is essentially attained. Formally, we have the following.

\begin{theorem}[Informal version of Theorem~\ref{thm:half-barrier} and Proposition~\ref{prop:half-optimal}]
For every sieve that only removes solutions, over formulas with $n$ variables, there is a satisfiable formula on which the probability of leaving an \odd number of solutions is at most $\frac{1}{2}+\frac{1}{2(2^{2^n}-1)}$, and there is a sieve that attains this value on every satisfiable formula.
\end{theorem}

The barrier applies to every sieve that only removes solutions and is drawn before seeing the formula. Dell et al.~\cite{DellKMW13} showed that for procedures that may look at the formula, a success probability above $2/3$ would imply $\mathrm{NP}\subseteq\mathrm{P}/\mathrm{poly}$, and left open whether anything between $1/2$ and $2/3$ can be reached. Our result shows that for sieves that do not see the formula the exact value is $1/2$, unconditionally.

\subsection{The Union Bound Factor}\label{sec:summary-unionbound}

Once the sieve succeeds with constant probability, the number of repetitions is dominated by the third factor, the union bound. Indeed, in the formula of Section~\ref{sec:summary-analysis} the term $s\ln 2$ is then much larger than $\ln(1/\epsilon_i)$, as practical QBF instances have hundreds or thousands of outer variables. This factor received little attention in the literature (Section~\ref{sec:union-bound}). We first show that the parameter that determines it is not the number $2^s$ of assignments to the outer variables, which we call outer assignments. Each outer assignment turns the rest of the formula into a function of the remaining variables, which we call the inner variables, and we call this function the \emph{inner function} of the assignment. No sieve can tell apart two outer assignments with the same inner function, so there is no reason to count both. The parameter that matters is therefore the number of \emph{distinct} inner functions, which we denote by $C$. It can be as small as $1$ and as large as $2^s$. We prove the following.

\begin{theorem}[Informal version of Theorem~\ref{thm:equivclass}]
Two outer assignments succeed and fail together on every hash function if and only if they induce the same inner function. Consequently, the union bound need only range over the $C$ distinct inner functions, and no coarser grouping is sound.
\end{theorem}

We then show that once $C$ is fixed, the factor cannot be improved. On a natural family of formulas the union bound over the $C$ groups is essentially exact, and no sharper bound that uses only the pairwise correlations between the groups saves even a single repetition (Proposition~\ref{prop:pairwise-optimal}). On the positive side, we give a simple syntactic bound on $C$. If $B$ is the set of outer variables that appear in some clause together with an inner variable, then $C\leq 2^{|B|}$ (Proposition~\ref{thm:boundaryC}). This bound is only syntactic and can be far from tight. Sharper criteria that are still cheap to check are left for future work. In our view, bounding $C$ cheaply, or recognizing formulas with a small $C$, is the central question for a practical Toda-based solver.

\subsection{Toda's Reduction Modulo \texorpdfstring{$q$}{q}}\label{sec:summary-modq}

We also examine the generalization of the reduction from counting modulo $2$ to counting modulo an arbitrary $q$, due to Toda and Ogiwara~\cite{TodaOgiwara1992} and Tarui~\cite{Tarui1993}, and write it in terms of Boolean formulas (Section~\ref{sec:modp}). Counting modulo a larger $q$ can allow a larger success probability per repetition, but using it in Toda's reduction adds a cost to the construction, and we weigh the one against the other. Formally, the analogue of the $1/2$ barrier for a modulus $q$ is $\frac{q-1}{q}$, and we prove the following.

\begin{theorem}[Informal version of Theorem~\ref{thm:modq-barrier}]
For every modulus $q\geq 3$, prime or composite, no sieve modulo $q$ attains the natural ceiling of $\frac{q-1}{q}$.
\end{theorem}

We conclude that, for hash functions of the same size, the modulus $2$ produces the smallest formulas (Proposition~\ref{prop:mod2-optimal}).

\subsection{Engineering and Implementation}\label{sec:summary-engineering}

The last part of this work deals with the engineering that is needed to turn the reduction into a working tool (Section~\ref{sec:implementation}). The analysis leaves a number of choices open, and we discuss each of them. Model counters expect their input in clausal form, and we show that the construction can keep the formula in this form throughout, at the cost of one added literal per clause each time repetitions are combined. As a result, only the sieve needs a Tseitin encoding, and since the values of the variables of this encoding are determined by the other variables, the counter can be asked to ignore them when counting, which is called projected counting. We show that the repetitions of the outermost step can be run as independent oracle calls, in parallel, and combined outside the counter. We also describe several smaller savings, such as merging adjacent quantifier blocks, removing an innermost universal block, and cancelling redundant operations when repetitions are combined. In addition, we discuss trade-offs, such as replacing one run with a small error by a majority vote over several runs with a larger error. Finally, we report preliminary results of a prototype implementation that uses a state-of-the-art model counter as an oracle.

Table~\ref{tab:summary} collects the results of this section in one place.

\begin{table}[!htb]
\centering
\small
\begin{tabular}{p{2.1cm}|p{3.2cm}|p{3.6cm}|p{3.1cm}}
\hline
\textbf{Component} & \textbf{Textbook reduction} & \textbf{This work} & \textbf{Provable limit} \\
\hline
Error allocation & geometric: $\epsilon_i=\epsilon/2^{i}$ & balanced: $\epsilon_i=1-(1-\epsilon)^{1/\lceil d/2\rceil}$ (Thm.~\ref{thm:boundedalloc}) & exact worst case, attained by an adversary \\
\hline
Sieve success probability & $1/8n$ per hash function, or $1/2$ at the price of $n$ copies of the formula & arbitrarily close to $1/2$, via the sieve of~\cite{NRS95} in our encoding (Thm.~\ref{thm:parity-isolation}) or via modular addition (Prop.~\ref{prop:MAuse}) & $1/2$, unconditionally (Thm.~\ref{thm:half-barrier}) \\
\hline
Union bound & $2^{s}$, one term per outer assignment & $C$, the number of distinct inner functions (Thm.~\ref{thm:equivclass}), and $C\leq 2^{|B|}$ (Prop.~\ref{thm:boundaryC}) & tight: no pairwise refinement saves a repetition (Sec.~\ref{sec:tightUB}) \\
\hline
Modulus & $2$ & $2$, proved optimal (Prop.~\ref{prop:mod2-optimal}) & ceiling $(q{-}1)/q$ unattainable for $q\geq 3$ (Thm.~\ref{thm:modq-barrier}) \\
\hline
Algorithm and encoding & proof-level description & explicit algorithm (Sec.~\ref{sec:todaBrief}), clausal form maintained throughout, projected counting, parallel oracle calls (Sec.~\ref{sec:implementation}) & \\
\hline
\end{tabular}
\caption{The components of Toda's reduction: the value in the textbook presentation, the value achieved in this work, and the provable limit. The first three rows are the three blow-up factors. The constant-probability sieve itself is due to Naik, Regan and Sivakumar~\cite{NRS95}. Its Boolean encoding, the modular addition technique, and the bounds in the last column are from this work.}
\label{tab:summary}
\end{table}

\section{Preliminaries}\label{sec:prelim}
\paragraph*{Boolean Formulas}
A Boolean formula $F(\vec{x})$ over a set of variables $\x = \{x_1, \ldots, x_m\}$ is a logical expression that evaluates to either \true or \false. The basic Boolean operators are conjunction ($\land$), disjunction ($\lor$), and negation ($\neg$). Since we also work with non-CNF formulas we denote by $|F|$ the \emph{size} of $F$, which is the length of the string representing $F$. An assignment is an evaluation of every variable $x_i$ to either \true or \false.  A \emph{solution} (also called a \emph{satisfying assignment} or a \emph{model}) of $F$ is an assignment for $F$ that evaluates $F$ to be $\true$. The number of solutions for $F$ is denoted by $\#F$.
A quantified Boolean formula (QBF), in prenex normal form, is a formula of the form:
$ F\equiv Q_1 \x_1 \ldots Q_d \x_d \: \hat{F}(\x_1, \ldots, \x_d)$.
Each $Q_i$ is either an existential quantifier ($\exists$) or a universal quantifier ($\forall$), and $\hat{F}$ is a quantifier-free Boolean formula. The string of quantifiers $Q_1 \vec{x}_1 \ldots Q_d \vec{x}_d$ is called the \emph{prefix} of $F$. The set of variables $\vec{x}_i$ is called the \emph{block} of $Q_i$. Throughout the paper we write $d$ for the number of blocks in the prefix and, following common practice, refer to $d$ loosely as the number of \emph{alternations}, although strictly speaking a prefix with $d$ blocks alternates $d-1$ times. Variables in $\hat{F}$ that are not quantified are called \emph{free variables} and they can be assigned to \true or \false. When all the variables in $\hat{F}$ are quantified, then $F$ is a sentence that can be evaluated to \true or \false.  Two Boolean formulas $F,G$ are \emph{logically equivalent}, denoted $F\equiv G$, if $F$ and $G$ have an identical solution space.

\paragraph*{Boolean Formula Operators}
We define several operators that we make use of throughout the paper. For Boolean formulas $F(\x)$ and $G(\y)$, where $\x,\y$ are not necessarily disjoint, the term $F(\x) + G(\y)$ denotes the formula  $[F(\x) \land x_{new}] \lor [\neg x_{new} \land G(\y)]$ where $x_{new}$ is a fresh variable. Note that $\#[F(\x)+G(\y)] = \#F(\x)+\#G(\y)$. The + operator can be naturally extended to a summation of $n$ formulas, using $\lceil \log(n) \rceil$ fresh variables.
Assume $\x=(x_1,\ldots, x_m)$ is a list of variables of $F$, not necessarily all of them. Then the term $F + 1(\x)$ denotes the formula $[F \land x_{new}] \lor [\neg x_{new} \land x_1\land \ldots\land x_m]$, which for every fixed assignment to the other variables of $F$ has exactly one solution more than $F$ has. Sometimes when we convert $F$ to $F+1$, we say that we perform the \emph{$+1$ operation} on $F$.
Next, when $\x$ and $\y$ are disjoint variable sets, the term $F(\x) \times G(\y)$ denotes the formula $F(\x)\land G(\y)$. Observe that $\#[F(\x) \times G(\y)] = \#F(\x) \times \#G(\y)$.
 Finally, we define the \emph{parity quantifier} $\parity$, which we use throughout the paper.  Let $F(\x)$ be a formula with free variables $\x$. Then the quantifier $\parity{\x}$ is defined as follows. $\parity\x F(\x) = True$ if $\#F(\x)\mod{2}=1$ and $\parity\x F(\x) = False$ if $\#F(\x)\mod{2} = 0$. In the following claim, $F(\x)$ and $G(\y)$ are Boolean formulas over disjoint variable sets $\x$ and $\y$. See Appendix~\ref{app:TodaElements} for the proof.

\begin{restatable}{claim}{parityproperties}\label{clm:many}

We have the following:

\begin{enumerate}
    \item  $\neg\parity\x F(\x) = \parity\x(F(\x)+1)$ \label{prop:parity:negation}
    
     \item $\parity\x \parity \y F(\x,\y) =\parity \x, \y  F(\x,\y)$. \label{prop:parity:collapse_nested}
     
    \item $\parity\x F(\x)\land \parity\y G(\y) = \parity(\x,\y)[F(\x)\times G(\y)]$.

    \item $\parity\x F(\x)\lor \parity\y G(\y) = \neg(\parity(\x,\y)[(F(\x)+1) \times (G(\y)+1)]) = \parity \x,\y [(F(\x)+1) \times (G(\y)+1)]+1$.\label{prop:parity:disjunction}

    \item $[\parity\x F(\x,\y)]\land G(\y) =\parity\x [F(\x,\y)\land G(\y)]$. \label{prop:parity:insert_h}

\end{enumerate}

\end{restatable}

\paragraph{Valiant-Vazirani Hash Function}\label{sec:vvthm}

A key component in Toda's reduction is the use of hash functions to probabilistically isolate a single solution, out of possibly many, for a formula $F$. The technique being used relies on the Valiant-Vazirani Theorem~\cite{ValiantVazirani} that we describe below, in the form in which it appears in~\cite{arora2006computational}.
We define a family of hash functions $H_{n,m} = \{h: \{0,1\}^n \rightarrow \{0,1\}^m\}$, where each such function takes the form $h(\x) = A\x + \vec{b}$, where $A$ is an $m \times n$ matrix, and $\vec{b}$ is an $m$-dimensional vector (all operations are performed modulo 2). Naturally, $h$ can be described as a Boolean formula. To simplify the notation we assume throughout the paper, unless mentioned otherwise, that when we refer to $h(\x)$ as a Boolean formula, we mean the Boolean formula $h(\x)=0^m$.

\begin{theorem}[Valiant--Vazirani]\label{thm:vv}
Let $F:\{0,1\}^n\rightarrow\{0,1\}$ be a Boolean formula. Choose $m$ uniformly at random from $\{2,\ldots,n+1\}$ and $h$ uniformly at random from $H_{n,m}$, and define $\widetilde{F}(\x) = F(\x) \land h(\x)$. Then:
(i) if $\exists\x\, F(\x) = \mathit{True}$, then
\begin{equation}\label{eq:vvsingle}
\Pr\left(\#\widetilde{F} = 1\right) \;\geq\; \frac{1}{8n};
\end{equation}
(ii) if $\exists\x\, F(\x) = \mathit{False}$, then $\#\widetilde{F} = 0$.
\end{theorem}

Theorem~\ref{thm:vv} says that if there \emph{exists} a satisfying assignment to $F$, then with probability at least $1/8n$ the hash function leaves exactly one satisfying assignment. If, however, $F$ has no solutions, then it leaves none, whatever the draw. We sometimes refer to the theorem and the method it describes as simply Valiant-Vazirani, or an \emph{isolation lemma}, or a \emph{sieve}.

Note that the construction $h(\x) = A\x + \vec{b}$ described above requires $mn$ gates.

Finally, throughout the paper we use the term ``success probability'' with respect to Valiant-Vazirani or to any other hash function. We define it here properly. Given a satisfiable formula $F$ and a randomly chosen hash function $h$, the \emph{success probability} of using $h$ is the probability that $F \land h$ has an odd number of satisfying assignments (specifically a single assignment in the case of Valiant-Vazirani).

\section{Toda's Reduction: From Theory to Algorithm}\label{sec:todaBrief}

In order to convert Toda's reduction to an algorithm, we first need to uncover the fine details of the reduction. In a high level, the reduction converts a given QBF formula $F$ to a quantifier-free formula $F_0$, followed by checking (by using e.g. a model counter) the parity of the number of solutions for $F_0$. Then with high probability, $F$ is satisfiable if and only if the number of solutions for $F_0$ is \emph{odd}. Naturally, all the theorems and proofs in this section originate from~\cite{Toda91}. Since the reduction is well-known and proved in many places, we do not provide a formal proof of the reduction, but rather focus on the details that are required to convert Toda's reduction to an algorithm.

\subsection{Toda's Reduction Step by Step}\label{sec:explToda}

Building upon the Valiant-Vazirani technique of probabilistically isolating a single solution, Toda's reduction presents a systematic method for transforming any quantified Boolean formula into a formula that uses only a single parity quantifier ($\parity$).
The reduction follows a series of transformations that can be applied repeatedly to eliminate all existential and universal quantifiers. The process works as follows. 

We start with a given formula $F_d = Q_1\x_1, \cdots Q_d\x_d \hat{F}(\x_1, \cdots, \x_d)$, and $0<\epsilon< 1/2$, where $\hat{F}$ is a quantifier-free formula, and $\epsilon$ is our maximum allowed error. Note that $\epsilon < 1/2$ since otherwise we could just replace the whole process with a single-time coin tossing.
First, all universal quantifiers ($\forall\vec{x}$) are replaced with negated existential quantifiers ($\neg\exists\vec{x}\neg$) using standard negation rules. This results in a formula with a prefix with only existential quantifiers and negations of the form $\neg\exists\x_1\neg\exists\x_2\neg\cdots\neg\exists\x_d\neg \hat{F}$ where the outermost (resp. innermost) $\neg$ is removed in case that $Q_1$ (resp. $Q_d$) is existential. Note that if $Q_d$ is universal then we can propagate the innermost negation to $\hat{F}$. Thus we may assume without loss of generality that the innermost $\neg$ is removed.

We next describe the process inductively. We start in step $d$ and decrease to step $0$, where $F_0$ is the resulting formula. Suppose that we are in step $i$ where $1\leq i\leq d$ and we have the formula:

$$F_i=\neg\exists\x_1\neg\cdots\neg\exists\x_i\parity\y_i G_i(\x_1,\cdots,\x_i, \y_i)$$

where $G_i$ is quantifier-free and $\y_i$ is a combination of the variables $\x_{i+1},\dots,\x_{d}$ and fresh variables that were added in the process. Note that for every assignment $\vec{\sigma}$ to $\x_1\cdots \x_i$: the formula $\parity\y_iG_i$ is \true if and only if there is an odd number of assignments for $\y_i$ that make  $G_i(\vec{\sigma},\y_i)$ \true. 
For simplicity we denote the set of variables $(\x_1,\x_2,\cdots,\x_{i-1})$ by $\x$, which we call the \emph{outer variables} and the whole prefix of $\neg\exists\x_1\neg\exists\x_2\cdots\neg\exists\x_{i-1}$ by $\vec{Q}\x$, which we call the \emph{outer prefix}. If $i=0$ then $\vec{Q}\x$ and $\x$ are just empty.

We next show how to eliminate the innermost existential quantifier $\exists\x_i$ of $F_i$. To do this we transform $\exists\x_i$ to a parity quantifier $\parity\x_i$ by conjuncting $\parity\y_iG_i$ with a randomly chosen hash function $h$ from $H_{n_i,m}$ over the variables $\x_i$, where $n_i$ is the number of variables of $\x_i$. Specifically, Valiant-Vazirani ensures us that for every given assignment $\vsigma$ for $\x$:

\begin{itemize}
    \item if $\parity \y_i G_i$ has no solutions (with $\x_i$ being now its only free variables the meaning for solution is assignment for $\x_i$  such that $\parity \y_i G_i$ is \true) then $(\parity\y_i G_i(\vsigma,\x_i,\y_i) \land h_1(\x_i))$ has no solutions, which makes $\parity \x_i \,(\parity \y_i G_i(\vsigma,\x_i,\y_i) \land h_1(\x_i))$ \false.
    \item If on the other hand $\parity \y_i G_i$  has satisfying assignments then  $(\parity \y_i G_i(\vsigma,\x_i, \y_i) \land h_1(\x_i))$ has a single assignment with some small probability, which makes $\parity \x_i \,(\parity \y_i G_i(\vsigma,\x_i,\y_i) \land h_1(\x_i))$ \true\kern-0.35em.
\end{itemize}

We now need to amplify the low probability from Valiant-Vazirani. For that, we repeat the process above $k$ times (the number $k$ is polynomial, and Section~\ref{sec:numcomp} gives its exact value) by duplicating $\parity\y_iG_i$ where each time $h_i$ is independently chosen. We ask that the process succeed at least once, by taking the disjunction of the different duplicates. All this results in a formula of the form:

\begin{equation}\label{eq:beginF_i}
\begin{split}
F'_i\equiv\vec{Q}\x\neg(&[\parity\x_i \,([\parity\y_i G_i(\x,\x_i,\y_i)] \land h_1(\x_i))] \lor \ldots \lor \\ &[\parity\x_i \, ([\parity\y_i G_i(\x,\x_i,\y_i)] \land h_k(\x_i))]).
\end{split}
\end{equation}

Next, since we are eventually interested in the parity of a quantifier-free formula, we convert the formula to an overall form of: 

\begin{equation}\label{eq:endF_i}
\begin{split}
{F''}_i\equiv\vec{Q}\x\neg\parity \x^1_i, \y^1_i,\cdots \x^k_i, \y^k_i [&[(G_i(\x,\x^1_i, \y^1_i)\land
h_1(\x^1_i))+1(\x^1_i,\y^1_i)] \\
\times \cdots \times &[(G_i(\x, \x^k_i,\y^k_i)\land h_k(\x^k_i))+1(\x^k_i,\y^k_i)]]
\\ & +1(\x^1_i,\y^1_i\cdots \x^k_i,\y^k_i).
\end{split}
\end{equation}

The conversion from Equation~\ref{eq:beginF_i} to Equation~\ref{eq:endF_i} is a direct application of the properties of Claim~\ref{clm:many}. First use Property 5, then 2, and finally 4.
See Appendix~\ref{app:TodaElements} for more details.

Note that to be able to use Property 4, we renamed the inner variables in each of the repetitions. Specifically, every repetition $j$ now needs its own fresh variables $x^j_i$ and $y^j_i$. As a result, this conversion introduces many more fresh variables. The outer variables $\x$ however, still remain the same in all duplications.

Finally we propagate the left-most negation, right after the outer prefix of ${F''}_i$, into the formula using Property 1.
Note that if we are eliminating the outermost quantifier, and it is existential, then this $+1$ addition at the end is not used. Altogether, we get:

\begin{fleqn}
\begin{equation}
\begin{split}
F_{i-1}\equiv\vec{Q}\x\parity \x^1_i, \y^1_i,\cdots \x^k_i \y^k_i [&[(G_i(\x,\x^1_i, \y^1_i)\land h_1(\x^1_i))+1(\x^1_i,\y^1_i)]
\\ \times\cdots \times &[(G_i(\x,\x^k_i,\y^k_i)\land h_k(\x^k_i))+1(\x^k_i,\y^k_i)]]
\\ +1(\x^1_i,\y^1_i&\cdots \x^k_i,\y^k_i) +1(\x^1_i,\y^1_i\cdots \x^k_i,\y^k_i).
\end{split}
\end{equation}
\end{fleqn}

Note that we have $+1(\x^1_i,\y^1_i\cdots \x^k_i,\y^k_i) +1(\x^1_i,\y^1_i\cdots \x^k_i,\y^k_i)$ and since this is under $\parity \x^1_i, \y^1_i,\cdots \x^k_i \y^k_i$ then $+1+1=+2=0$ so we can remove both of them. The original proof does not perform this cancellation, and to stay faithful to the original reduction we present it as is for now. Later, when we implement it, we of course cancel out the two consecutive $+1$'s and remove them.

The next iteration proceeds with $F_{i-1}$, where the parity quantified variables
$\x^1_i, y^1_i,\cdots \x^k_i, \y^k_i$ are now denoted by $\y_{i-1}$.
The promise that we get at the end of each step conversion is as follows: For every given $\epsilon_i>0$,  where $\epsilon_i=2^{-t}$ for some $t$, there is a number $k_i$ of repetitions, polynomial in the size of $F_i$ and $t$, such that if $F_i$ is \true then $F_{i-1}$ is \true with probability at least $1-\epsilon_i$, and if $F_i$ is \false then $F_{i-1}$ is \false with probability at least $1-\epsilon_i$. As with the case of $\epsilon$, we also assume here that $\epsilon_i < 1/2$ for every $i$. We explain in Section~\ref{sec:inner_epsilon_naive} how to allocate these $\epsilon_i$'s, so that
all in all, we get the result in the main theorem of Toda's reduction (see Lemma 17.17 in~\cite{arora2006computational}) that we state below.

\begin{theorem}[{\cite{arora2006computational}}]
\label{thm:major}
Let $F$ be a QBF formula with $d$ alternations, and $0<\epsilon<1/2$. Then we have the following: If $F$ is true, then with probability at least $1-\epsilon$, $F_0$ is true. If $F$ is false, then with probability at least $1-\epsilon$, $F_0$ is false.
\end{theorem}

\subsection{Toda's Reduction based Algorithm}
Following the description of Toda's reduction in the previous section, we next present Algorithm~\ref{alg:Toda1} that implements the reduction. 
The algorithm proceeds as follows. We are given a QBF formula $F$, with the number of alternations $d$, and a probability measure $\epsilon>0$. First, in Line 2, we compute an \emph{inner-error} probability value $\epsilon_i$ for each quantifier level $i$. This allocation of error probability at every step ensures that the overall algorithm maintains the desired error bound $\epsilon$. Lines 3-5, 9-11 and 13-15 correspond to the first part, where we converted the formula to the form $\neg\exists\x_1\cdots\neg\exists\x_d\neg \hat{F}$. Lines 3-5 supply the innermost negation when $Q_d$ is universal, Lines 9-11 the negation in front of each of $\exists\x_2,\ldots,\exists\x_d$, and Lines 13-15 the leading negation when $Q_1$ is universal, the latter two as a $+1$ by Property~\ref{prop:parity:negation}. The core of the algorithm is in Lines 7-12, where we iterate through all quantifiers from innermost to outermost. We start in Step $d$ and decrease to Step $1$. For each quantifier we apply the \textbf{Amplify} function, which applies a hash function and amplifies using repetitions as defined in the reduction. Finally, in Line 16 we return $F_0$, that contains only parity quantifiers, having eliminated all existential and universal quantifiers from the original formula.

We next describe the function \textbf{Amplify}, shown in Algorithm~\ref{alg:Amplify}.
In Line 2 we calculate $s$, the total number of variables in the outer quantifiers. Using $s$, $\epsilon$ and the $\x_i$ variables, we then determine the number of repetitions $k_i$ needed to achieve the desired error probability by using the \textbf{ComputeRepetitions} function, described in the next section.
Then in Line 5, for each repetition, we AND the function with a \textbf{hash}, where the hash is the hash function from Valiant-Vazirani, as described in the Preliminaries. In line 6 we \textbf{Rename} the variables from the current and previous quantifiers, giving each repetition its own fresh copy of variables. Finally, in Line 8 we \textbf{Combine} all of the repetitions as in Equation~\ref{eq:endF_i}. The further $+1$ that propagates the negation is added by Algorithm~\ref{alg:Toda1}.


\begin{algorithm}
\SetAlgoLined
\KwIn{QBF formula $\hat{F}=Q_1\x_1,\cdots Q_d\x_d {F}(\x_1,\cdots,\x_d)$, $d$ alternations,error bound $\epsilon>0$ }
\KwOut{Parity-quantified formula}
\SetKwFunction{ComputeInnerError}{ComputeInnerError}
\SetKwFunction{Amplify}{Amplify}
\Begin{
    $\epsilon_1,\dots,\epsilon_{d}\gets$ \ComputeInnerError{$\epsilon, d$}\

     \If{$Q_d\equiv\forall$}{
        
                ${F} \gets \neg{F}$\
            }
     
      ${F_d} \gets {F}$\
      
     \For{$i \gets d$ \KwTo $1$}{
      ${F_{i-1}}(\x_1,\cdots,\x_{i-1},\y_{i-1}) \gets \Amplify({F_{i}},\x_1,\cdots,\x_{i},\y_i,\epsilon_{i})$\

      \If{$i > 1$}{

                ${F_{i-1}} \gets {F_{i-1}}+1$\
                }
         
    }
    \If{$Q_1\equiv\forall$}{
                
                ${F_{0}} \gets {F_{0}}+1$\
                }
    \Return $ \parity \y_{0} \hat{F_0}(\y_0)$           
}
\caption{\TodaOne }\label{alg:Toda1}
\end{algorithm}



\begin{algorithm}
\SetAlgoLined
\KwIn{Boolean formula $F(\x_1,\cdots,\x_{i},\y_i)$, $\x_1,\cdots,\x_i,\y_i,\epsilon>0$, }
\KwOut{Boolean Formula $\hat{F}$}
\SetKwFunction{hash}{hash}
\SetKwFunction{computeReps}{ComputeRepetitions}
\SetKwFunction{rename}{Rename}
\SetKwFunction{combine}{Combine}
\Begin{
    $s\gets \left(\sum_{j=1}^{j=i-1}|\x_j|\right)$
    
   $k_i\gets$ \computeReps{$s, |\x_i|, \epsilon$} \Comment{See next section} \
   
     \For{$j \gets 1$ \KwTo $k_i$}{
          $F_j= (F \land \hash(\x_i))$ \
          
          $\rename(F_j,\x_i,\y_i)$ \
        }
 \Return  $[(F_1+1)\times\cdots \times (F_{k_i}+1)]+1$
  }
\caption{Amplify}\label{alg:Amplify}
\end{algorithm}


\section{Toda's Algorithm Analysis}\label{sec:alganal}

We next analyze some key components of Toda's reduction algorithm. Specifically, we explain how the inner-error probabilities and the number of repetitions are calculated, deriving accurate values rather than the asymptotic bounds found in textbooks. We then assess the size of the resulting formula and its dominating factor.

\subsection{Parameters Derivation}\label{sec:param_derivation}
\subsubsection{Computing the Inner-Error Probability}\label{sec:inner_epsilon_naive}

To achieve the overall error bound $\epsilon$, the error has to be bounded across each step of Toda's reduction. While we follow the error analysis in~\cite{arora2006computational}, we simplify the proof and use a bound over the entire reduction sequence rather than an inductive approach. Note though that we keep the shape of the allocation used there: the inductive accounting of~\cite{arora2006computational} grants the outermost step the largest share of the error budget and geometrically less to each deeper step, and our $\epsilon_i = \epsilon/2^{i}$ does the same. We therefore call this method of allocation \emph{geometric allocation}.

We first add some necessary formalism.
For each step $i$, we define $Tr_i$ as the event that formula $F_i$ is \true, and $Fl_i$ as the event that formula $F_i$ is \false.
Since either the given formula $F=F_d$ is \true or \false, we have that $\Pr(Tr_d) \in \{0,1\}$. Our objective is to ensure that
$\Pr(Fl_0 \mid Tr_d) \leq \epsilon$ and $\Pr(Tr_0 \mid Fl_d) \leq \epsilon$, since this achieves the overall error bound $\epsilon$. As such, we define the inner-error probability at step $i$ as $\epsilon_i = \max(\Pr(Fl_{i-1} \mid Tr_{i}),\Pr(Tr_{i-1} \mid Fl_{i}))$.

\begin{proposition}\label{clm:innerprobnaive}
 Setting $\epsilon_i = \frac{\epsilon}{2^{i}}$ for every step $i$, ensures that the overall error bound $\epsilon$ is achieved.
\end{proposition}

\begin{proof}
Assume that the formula $F$ is \true (the argument for $F$ being \false is similar). We then have the following bound: If the formula $F$ that we start with is \true (event $Tr_d$) then in order to end up with an incorrect result (event $Fl_0$), we must make our first error at some step $i \in \{1,2,\dots,d\}$. This is the probability that we maintain correctness through the previous steps, multiplied by the probability of making an error when transitioning from step $i$ to step $i-1$.  Formally we have:
\newline $\Pr(\text{first error at step } i \mid Tr_d ) = \Pr(Tr_i \land\cdots\land Tr_{d-1} \mid Tr_d) \cdot \Pr(Fl_{i-1} \mid Tr_i)$,
and in total, we have: 


\begin{align*}
\Pr(Fl_0 \mid Tr_d) & \leq \sum_{i=1}^{d} \Pr(Tr_i \land\cdots\land Tr_{d-1}\mid Tr_d) \cdot \Pr(Fl_{i-1} \mid Tr_i) \leq\\
& \leq \sum_{i=1}^{d} \Pr(Fl_{i-1} \mid Tr_{i}) \leq \sum_{i=1}^{d} \epsilon_i.
\end{align*}

Therefore, we assign $\epsilon_i = \frac{\epsilon}{2^{i}}$ as the inner-error bound in Step $i$. This gives a total bound of at most $\epsilon$, since then we have that $\Pr(Fl_0 \mid Tr_d) \leq \epsilon \sum_{i=1}^{d} \frac{1}{2^{i}} \leq \epsilon$.
\end{proof}

\subsubsection{Computing the Number of Repetitions}\label{sec:numcomp}

For each quantifier elimination step, we need to determine how many repetitions of the Valiant-Vazirani technique are required to achieve the desired bound over the inner-error probability. Define the number of repetitions in step $i$ to be $k_i$. Note that $k_i$ is the result returned by \textbf{ComputeRepetitions} from Algorithm~\ref{alg:Amplify}. 
Let $p$ be the success probability of a single application of Valiant-Vazirani. More formally, $p$ is the probability that, when $h_1$ is drawn as in Theorem~\ref{thm:vv}, the formula $\parity\x_i(\parity\y_i G_i(\vsigma,\x_i,\y_i) \land h_1(\x_i))$ is logically equivalent to $\exists\x_i\parity\y_iG_i(\vsigma,\x_i,\y_i)$. Then by Theorem~\ref{thm:vv}, and specifically by Equation~\ref{eq:vvsingle}, $p$ is at least $1/8n_i$. Note that here $\vsigma$ is an assignment to the outer variables $\x$.
However, we need to ensure success across all possible assignments $\vec{\sigma}$ to $\x$ (the outer variables). Specifically, if there are $s$ such outer variables (resulting in $2^s$ possible assignments), then we need the failure probability to be sufficiently low across all of them.

\begin{proposition}\label{clm:repetitionUNbound}
Setting $k_i \geq \log_{1-p}(\epsilon_i/2^s)$ ensures that the probability for error at step $i$ is at most $\epsilon_i$. This decomposes as: $k_i \;\geq\; \underbrace{\log_{1-p}(\epsilon_i)}_{\text{error term}} \;+\; \underbrace{s \cdot \log_{1-p}(1/2)}_{\text{union bound term}}$
\end{proposition}

\begin{proof}
Assume we duplicate $\parity \x_i \,(\parity y_i G_i(\vsigma,\x_i,\y_i) \land h_1(\x_i))$
for $k_i$ independent repetitions. The probability of failing in all of the repetitions is at most $(1-p)^{k_i}$.
To achieve our desired error rate of no more than $\epsilon_i$, we need to ensure that the actual error rate is lower than $\epsilon_i$. In other words we need to ensure that $(1-p)^{k_i} \leq \epsilon_i$. Taking the logarithm of both sides we get $k_i \geq \log_{1-p}(\epsilon_i)$.
Therefore, we need at least $k_i \geq \log_{1-p}(\epsilon_i)$ repetitions.

This analysis applies to a specific assignment $\vsigma$ to $\x$. However, we need to ensure that we bound the error rate over \emph{all} possible assignments to $\vsigma$.

Recall that $s=|\x|$ is the number of outer variables. Then we have $2^s$ possible assignments for $\x$. Using union bound, we get that:
\[
    \Pr(\text{Failure}) \leq 2^s \cdot \Pr(\text{Failure for a single assignment}) = 2^s(1-p)^{k_i}
\]

Similarly to the single assignment case, to maintain the acceptable error rate $\epsilon_i$ for the current step $i$, we require that $\epsilon_i \geq 2^s(1-p)^{k_i}$ which implies $(1-p)^{k_i} \leq \frac{\epsilon_i}{2^s}$.
Taking logarithms, we get $k_i \geq \log_{1-p}(\epsilon_i/2^s)$.
\end{proof}

Proposition~\ref{clm:repetitionUNbound} reveals three independent factors $\epsilon_i$, $p$ and $s$, while also decomposing the term into two independent terms effectively separating $\epsilon_i$ from $s$

To get a concrete sense of scale, we replace $p$ with $1/8n$, which by Equation~\ref{eq:vvsingle} is the success probability bound for a single application of the Valiant-Vazirani hash function. Then, we convert this expression to natural logarithms and use the approximation $\ln(1-p) \approx -p$ (which holds since $p$ is small). All in all we get, approximately, $k_i \geq 8n\ln(1/\epsilon_i) + 5.5\,ns$.
In other words: the union bound and hash probability factors ($s$ and $p$) affect the number of repetitions linearly, while the error allocation factor ($\epsilon_i$) enters only logarithmically.




\subsection{Formula Size Analysis}\label{sec:sizeComp}

We now turn our attention to the size of the resulting parity-quantified formula, $F_0$. Following Algorithm~\ref{alg:Amplify}, at every step $i$ of the reduction, we take the formula from the previous step $F_{i+1}$, append a hash function component $h_i$, and then clone the entire package $k_i$ times (each time with a different hash function), where $k_i$ is derived from Proposition~\ref{clm:repetitionUNbound}. Overall, the size of the final formula $F_0$ is:


\begin{multline}\label{eq:exactSize}
    |F_0| = k_1 \cdot \biggl( k_2 \cdot \bigl( \dots k_d \cdot (|F| + |h_d| + |+1|) + |h_{d-1}| +
    \\ |+1| \dots \bigr) + |h_1|  + |+1| \biggr)
\end{multline}





Observing Equation~\ref{eq:exactSize}, we note that the components that play a role in the size of $F_0$ are the repetitions, the hash functions and the $+1$ function. A careful look, however, shows that their part in the overall size of $F_0$ is not equal: at every step $i$ the size of the $+1$ function and the hash function $|h_i|$ depend only on the number of variables in the current quantifier block. In particular, it does \emph{not} depend on the size of the current formula $|F_i|$. On the other hand, each of the $k_i$ repetitions applies to the entire formula. Thus, a linear increase in the number of repetitions ($k$) results in an exponential divergence in final size as the steps progress. In contrast, the size of the hash function $|h|$ is merely an additive term inside the parentheses before the multiplication happens. Therefore, for the remainder of this paper, we effectively ignore the optimization of $|h|$. We assume $|h|$ is ``reasonable'', and focus on minimizing the number of repetitions at every step.
 Our objectives are then to find a better allocation for the $\epsilon_i$ parameters, increase the probability factor $p$, and decrease the outer variable parameter $s$.

\subsubsection{Practical Motivation}\label{sec:efficiency}

To better illustrate the above analysis, we consider a small example. Using Proposition~\ref{clm:repetitionUNbound}, we observe that even for a simple formula with three quantifiers over fifteen variables, with six variables in the outermost block, five in the middle one, and four in the innermost one, bounding the error to only $10\%$ requires $k_1 = 143$, $k_2 = 310$, and $k_3 = 379$ repetitions. This results in a formula more than $16{,}800{,}000$ times larger than the original. Even starting with a small formula of just $20$ clauses, the transformed formula would exceed $336{,}000{,}000$ clauses.

Such a huge blowup in the size of the formula raises the question of the practical capabilities of Toda's algorithm, even if all the aforementioned factors are reduced. To explore this question, we implemented a naive version of Algorithm~\ref{alg:Toda1} using a state-of-the-art model counter that can  handle parity queries as an oracle~\cite{SM2025}. We then set $p = 1/2$, despite the fact that Valiant-Vazirani provides no such guarantee, and scaled the outer-variable union bound parameter using $\alpha \cdot s$ with $\alpha = 0.2$. We also adopted an improved error allocation method for the $\epsilon_i$ values, which we introduce in the next section. We set the total probability error to be $\epsilon = 0.3$. Naturally, none of the obtained results meets the theoretical reliability guarantee. Nevertheless, we observed that with these tuned factors Toda's algorithm could handle competitive QBF benchmarks, and even ``solve'' benchmarks that other solvers could not. We note that without relaxing the factors, our implementation ran out of memory on a dedicated 12GB machine while constructing the formula.
These findings, although they do not constitute a proof, further motivate us to study methods that reduce these blow-up factors while maintaining reliability.

\section{Refining the Inner-Error Allocation}\label{sec:innerComp}

While our earlier analysis used the simplified geometric allocation with $\epsilon_i = \frac{\epsilon}{2^i}$, we provide here a more precise analysis that captures the different error behaviors of the existential and universal quantifiers. Additionally, we assign the same inner-error probability for all steps. In that sense, we are ``balancing the load'' between the various steps, and hence we name this method \emph{balanced allocation}. We note that using the exact same inner-error for all steps may not be optimal. The optimal balance for a given formula could be found using the analysis and method described below in combination with numerical solvers. We leave that for future work. Towards the end of this section, we discuss the benefits of balanced allocation.

In this section, we present two different analyses. In both of them we denote by $q$ the probability that a single step of the reduction, i.e. the elimination of a single quantifier, results in a logically equivalent formula, so that $\epsilon_i = 1-q$. The two analyses differ in what we may assume about $q$. The first assumes that every step succeeds with probability \emph{exactly} $q$, and the second, which is the one that applies to all the sieves that we use in this paper, assumes only the guarantee that every step succeeds with probability \emph{at least} $q$. We begin with the exact case.

\subsection{Exact balanced allocation}
\begin{theorem}\label{clm:innerprobnotnaive}
 Let $q$ be such that $1-\epsilon\leq (q^{\left(2\left \lceil{\frac{d}{2}}\right \rceil \right)} + q)/(q+1)$. Then
 setting $\epsilon_i = 1-q$ for every $i$ ensures the overall error bound $\epsilon$ is achieved.
\end{theorem}

We give a proof outline, see Appendix~\ref{sec:Appanalysis} for full proof.

\begin{proof}[Proof Outline]

As mentioned above we consider the case where $\epsilon_1=\cdots = \epsilon_d$, and $q=1-\epsilon_i$ is the probability that a single step of the reduction results in a logically equivalent formula. We show that if $q$ satisfies $1-\epsilon\leq (q^{\left(2\left \lceil{\frac{d}{2}}\right \rceil \right)} + q)/(q+1)$, then the overall bound $\epsilon$ is achieved.
For that, we make use of the fact that existential and universal quantifiers behave differently in terms of error propagation: (i) For existential quantifiers: If $F_{i+1}$ is \false\kern-0.35em, then $F_i$ remains \false with probability $1$. If $F_{i+1}$ is \true\kern-0.35em, then $F_i$ remains \true with a probability $q$, defined by the amplification in the reduction. (ii) For universal quantifiers, since $\forall x f(x) \equiv \neg\exists x \neg f(x)$, we get: If $F_{i+1}$ is \true\kern-0.35em, then $F_i$ remains \true with probability $1$. If $F_{i+1}$ is \false\kern-0.35em, then $F_i$ remains \false with probability $q$.
Our method to track the error propagation is to group the quantifiers into pairs of $\forall\exists$ and treat each pair as a single step, as depicted in Figure~\ref{fig:pair-propagation}. For a single $\forall\exists$ pair, see that for every $i$:
$\Pr(Tr_i|Tr_{i+2}) = q + (1-q)^2$, $\Pr(Fl_i|Tr_{i+2}) = (1-q)q$,
$\Pr(Tr_i|Fl_{i+2}) = 1-q$, and $\Pr(Fl_i|Fl_{i+2}) = q$.

Assume for now that we have an even number of quantifiers, and that the first quantifier is universal. Thus, we can view the sequence of quantifiers as pairs of  $\forall\exists$. Then for example we have that:

\begin{align*}
\Pr(Tr_0|Tr_{2l}) = \Pr(Tr_0|Tr_2)\Pr(Tr_2|Tr_{2l})+
\Pr(Tr_0|Fl_{2})\Pr(Fl_{2}|Tr_{2l}).
\end{align*}

Using this, for every $1\leq l\leq d/2$ (recall that $d$ is even), we obtain the following set of recurrence relations:

\begin{align*}
\Pr(Tr_0|Tr_{2l}) &= \Pr(Tr_2|Tr_{2l})(q + (1-q)^2) + \Pr(Fl_2|Tr_{2l})(1-q) \\
\Pr(Fl_0|Tr_{2l}) &= \Pr(Fl_2|Tr_{2l})q + \Pr(Tr_2|Tr_{2l})(1-q)q \\
\Pr(Tr_0|Fl_{2l}) &= \Pr(Tr_2|Fl_{2l})(q + (1-q)^2) + \Pr(Fl_2|Fl_{2l})(1-q) \\
\Pr(Fl_0|Fl_{2l}) &= \Pr(Fl_2|Fl_{2l})q + \Pr(Tr_2|Fl_{2l})(1-q)q.
\end{align*}

The above set of recurrence relations can be solved through standard linear algebra and matrix multiplications, resulting in the following:

\begin{equation*}
\begin{pmatrix}
\Pr(Tr_0|Tr_{2l}) \\
\Pr(Fl_0|Tr_{2l}) \\
\Pr(Tr_0|Fl_{2l}) \\
\Pr(Fl_0|Fl_{2l})
\end{pmatrix} =
\begin{pmatrix}
\frac{q^{2l+1} + 1}{q+1} \\
\frac{-q^{2l+1} + q}{q+1} \\
\frac{-q^{2l} + 1}{q+1} \\
\frac{q^{2l} + q}{q+1}
\end{pmatrix}_{\text{\normalsize .}}
\end{equation*}

From here we can find $\Pr(Tr_0|Tr_{d}),\Pr(Fl_0|Fl_{d})$, which we require because we want $1-\epsilon \leq \min(\Pr(Tr_0|Tr_{d}),\Pr(Fl_0|Fl_{d}))$.
We refer the reader again to Appendix~\ref{sec:Appanalysis} for the complete proof details and the cases where the quantifier pattern is different. In total,
the result of all the cases could be expressed as:

$$1-\epsilon\leq \min(\Pr(Tr_0|Tr_{d}),\Pr(Fl_0|Fl_{d})) = \frac{q^{\left(2\left \lceil{\frac{d}{2}}\right \rceil \right)} + q}{q+1}.$$
\end{proof}


\begin{figure}[h!]
\centering
\begin{tikzpicture}[
    true_node/.style={circle, draw, fill=green!20, minimum size=1.1cm, align=center, font=\small},
    false_node/.style={circle, draw, fill=red!20, minimum size=1.1cm, align=center, font=\small},
    my_arrow/.style={->, >=stealth, shorten >=1pt},
    rescue_arrow/.style={->, >=stealth, shorten >=1pt, very thick, densely dashed},
    header_text/.style={font=\small, align=center}
]

    \node[true_node] (T1) at (0, 2.6) {\true};
    \node[false_node] (F1) at (0, 0) {\false};
    \node[header_text] at (0, 3.9) {$F_{2j}$};

    \node[true_node] (T2) at (5, 2.6) {\true};
    \node[false_node] (F2) at (5, 0) {\false};
    \node[header_text] at (5, 3.9) {$F_{2j-1}$\\ after $\exists$ step};

    \node[true_node] (T3) at (10, 2.6) {\true};
    \node[false_node] (F3) at (10, 0) {\false};
    \node[header_text] at (10, 3.9) {$F_{2j-2}$\\ after $\forall$ step};

    \draw[my_arrow] (T1) -- node[above] {$q_{\exists,j}$} (T2);
    \draw[my_arrow] (T1) -- node[pos=0.34, above right, inner sep=1.5pt] {$1-q_{\exists,j}$} (F2);
    \draw[my_arrow] (F1) -- node[below] {$1$} (F2);

    \draw[my_arrow] (T2) -- node[above] {$1$} (T3);
    \draw[rescue_arrow] (F2) -- node[pos=0.66, below right, inner sep=1.5pt] {$1-q_{\forall,j}$} (T3);
    \draw[my_arrow] (F2) -- node[below] {$q_{\forall,j}$} (F3);

\end{tikzpicture}
\caption{Error propagation through a single $\forall\exists$ pair, where the inner $\exists$ quantifier is eliminated first. Each step is one-sided: an $\exists$ step can only turn \true into \false\kern-0.35em, and a $\forall$ step can only turn \false into \true\kern-0.35em. The dashed edge is the \emph{rescue} path, along which a \true formula that was wrongly turned \false by the $\exists$ step is turned back to \true by a second error in the $\forall$ step.}
\label{fig:pair-propagation}
\end{figure}


Theorem~\ref{clm:innerprobnotnaive} provides us with a simple method for inner-error allocation, that we call \emph{balanced allocation}: for every $i$, set $\epsilon_i = 1-q$, where $q$ be such that $1-\epsilon =  (q^{\left(2\left \lceil{\frac{d}{2}}\right \rceil \right)} + q)/(q+1)$. Then, the theorem ensures that by this allocation the overall error bound $\epsilon$ is achieved.

This formula, however, is derived under the assumption that we can control the value of $\epsilon_i$ exactly, while this may not always be the case. In Proposition~\ref{clm:repetitionUNbound} for example we have only the promise of a bound over the value of $\epsilon_i$, as it hugely depends on the hash function success probability, for which in turn we also only have a bound. Indeed, all three sieves that we use in this paper, guarantee only a lower bound on their success probability. The actual probability at step $i$ can therefore be any value in $[q,1]$, and can differ between steps.

At first glance a step that succeeds with a probability \emph{higher} than promised seems harmless. A closer look at Figure~\ref{fig:pair-propagation} reveals that this is not the case. The term $(1-q_{\exists,j})(1-q_{\forall,j})$ along the dashed rescue path is precisely what makes $\Pr(Tr_{i}|Tr_{i+2}) = q+(1-q)^2$ larger than $q$ in the analysis above, and it is what Theorem~\ref{clm:innerprobnotnaive} exploits in order to allow a generous per-step budget. This term, however, vanishes when $q_{\forall,j}=1$, which the $\geq q$ guarantee permits. In other words, a universal step that works \emph{better} than its guarantee removes the chance of a compensating error, and is therefore adversarial for the true-to-false direction. Symmetrically, an existential step that works better than its guarantee is adversarial for the false-to-true direction. We now derive the analysis for the general case in which we have only bounds.

\subsection{Bounded balanced allocation}\label{sec:boundedalloc}

As mentioned, the actual probability of each step can be any value in $[q,1]$, and can differ between steps. Let $q_{\exists,j}, q_{\forall,j} \in [q, 1]$ be the actual probabilities of the two steps in the $j$-th pair of quantifiers, and let $l=\lceil \frac{d}{2}\rceil$ be the number of pairs. To find the worst case, we model the situation as a game against an \emph{adversary}: an opponent who chooses all the actual probabilities $q_{\exists,j}, q_{\forall,j} \in [q, 1]$, one per step, so as to maximize the overall error. Since the given formula $F$ is either \true or \false\kern-0.35em, the overall error is:
\[
  \text{Overall error} \;=\; \max\!\Big(
    \underbrace{\Pr(Fl_0|Tr_{d})}_{\text{True} \to \text{False error}},\;\;
    \underbrace{\Pr(Tr_0|Fl_{d})}_{\text{False} \to \text{True error}}
  \Big).
\]
We analyze each direction separately, and obtain the following.

\begin{theorem}\label{thm:boundedalloc}
Assume that every step of the reduction succeeds with a probability of at least $q$, where the actual probability may be any value in $[q,1]$, independently at every step. Then setting $\epsilon_i = 1-q$ for every $i$, where $q = (1-\epsilon)^{1/\left\lceil \frac{d}{2}\right\rceil}$, ensures the overall error bound $\epsilon$ is achieved.
\end{theorem}

Before proving the theorem, we isolate the single observation on which the whole analysis rests. It states that, whatever the adversary does inside a pair, the pair as a whole is no more reliable than one single step that succeeds with probability exactly $q$.

\begin{claim}\label{clm:pairbound}
For every pair $j$ and every choice of $q_{\exists,j}, q_{\forall,j} \in [q,1]$, the probability that pair $j$ corrupts the truth value of the formula that enters it is at most $1-q$. That is,
$$\Pr(Fl_{2j-2}|Tr_{2j}) \;\leq\; 1-q \qquad\text{and}\qquad \Pr(Tr_{2j-2}|Fl_{2j}) \;\leq\; 1-q .$$
Moreover, both bounds are attained: the first by setting $q_{\forall,j}=1$ and $q_{\exists,j}=q$, and the second by setting $q_{\exists,j}=1$ and $q_{\forall,j}=q$.
\end{claim}

We prove Claim~\ref{clm:pairbound} in Appendix~\ref{sec:Appanalysis}, and use it here to outline the proof of the true-to-false direction. The false-to-true direction is symmetric, and the complete proof of both directions, together with the cases where the quantifier pattern is not a sequence of $\forall\exists$ pairs, appears in Appendix~\ref{sec:Appanalysis} as well.

\begin{proof}[Proof Outline]
Since the error in the true-to-false direction is $\Pr(Fl_0|Tr_d) = 1-\Pr(Tr_0|Tr_{d})$, it suffices to bound $\Pr(Tr_0|Tr_{d})$ from below by $q^{l}$. We do so by peeling off one pair at a time, showing that each pair costs us a factor of at most $q$.

We first set up the recurrence. Conditioning on the truth value of $F_{2j-2}$, that is, on the output of pair $j$, we obtain the very same relation that was used in the proof of Theorem~\ref{clm:innerprobnotnaive}, only that the transition probabilities of pair $j$ are now determined by $q_{\exists,j}$ and $q_{\forall,j}$ rather than by a single value $q$:

\begin{equation}\label{eq:T-rec}
\Pr(Tr_0|Tr_{2j}) \;=\; \Pr(Tr_0|Tr_{2j-2})\Pr(Tr_{2j-2}|Tr_{2j}) \;+\; \Pr(Tr_0|Fl_{2j-2})\Pr(Fl_{2j-2}|Tr_{2j}).
\end{equation}

The first term of Equation~\ref{eq:T-rec} is the direct path, in which pair $j$ preserves \true and the remaining pairs preserve it as well. The second term is the rescue, in which pair $j$ erroneously outputs \false\kern-0.35em, but the remaining pairs happen to flip it back to \true\kern-0.35em. It is convenient to read these two transition probabilities off Figure~\ref{fig:pair-propagation}: a \true formula entering pair $j$ is corrupted only along the single path in which the $\exists$ step fails and the $\forall$ step then succeeds, so that $\Pr(Fl_{2j-2}|Tr_{2j}) = q_{\forall,j}(1-q_{\exists,j})$.

We can now conclude in two moves. First, we simply discard the rescue term of Equation~\ref{eq:T-rec}, which we may do since it is a probability and hence non-negative. This is the step that makes the argument work for \emph{every} choice of the adversary at once, rather than for one particular choice, and it costs us nothing in the end because, as noted below, the resulting bound is attained. Second, we apply Claim~\ref{clm:pairbound} to the surviving term. Together they give, for every $1\leq j\leq l$:

\begin{align*}
\Pr(Tr_0|Tr_{2j}) &\;\geq\; \Pr(Tr_0|Tr_{2j-2})\Pr(Tr_{2j-2}|Tr_{2j}) \\
&\;=\; \Pr(Tr_0|Tr_{2j-2})\big[1-\Pr(Fl_{2j-2}|Tr_{2j})\big] \\
&\;\geq\; \Pr(Tr_0|Tr_{2j-2})\big[1-(1-q)\big] \;=\; q\cdot\Pr(Tr_0|Tr_{2j-2}).
\end{align*}

Since no pairs remain at level $0$, we have $\Pr(Tr_0|Tr_0)=1$, and applying the above inequality for $j=1,\ldots,l$ yields $\Pr(Tr_0|Tr_{2l}) \geq q^{l}$, that is, $\Pr(Fl_0|Tr_d)\leq 1-q^{l}$.

This bound is not merely an upper estimate but the exact worst case, as it is attained by the adversary that sets $q_{\forall,j}=1$ and $q_{\exists,j}=q$ at every pair, which is the choice of Claim~\ref{clm:pairbound}. The reason is easy to see. By succeeding with certainty, every universal step preserves the truth value it receives, and is therefore as if it were not performed at all. This is exactly what closes the rescue path of Figure~\ref{fig:pair-propagation}, and it leaves us with a chain of $l$ existential steps, each of which can only turn \true into \false\kern-0.35em. Along such a chain \false is absorbing, so that once the truth value has been lost it can never be recovered, and surviving as \true requires all $l$ steps to succeed, which happens with probability exactly $q^{l}$.

The false-to-true direction is symmetric, with the roles of $\exists$ and $\forall$ exchanged. There the adversary lets every existential step succeed with certainty and every universal step be as weak as the guarantee allows, leaving a chain of $l$ universal steps for which \true is the absorbing value, and giving $\Pr(Tr_0|Fl_d) \leq 1-q^{l}$ as well. Requiring $1-q^{l}\leq\epsilon$ in both directions yields $q \geq (1-\epsilon)^{1/l}$.
\end{proof}

Theorem~\ref{thm:boundedalloc} gives us a simple method for inner-error allocation, that we call \emph{bounded balanced allocation}: for every $i$, set $\epsilon_i = 1-(1-\epsilon)^{1/\lceil d/2\rceil}$. Since none of the sieves that we use provides an exact success probability, the setting of Theorem~\ref{clm:innerprobnotnaive} is the less realistic of the two. We therefore adopt the bounded balanced allocation from here on, and devote the rest of this section to comparing it against the geometric allocation of Section~\ref{sec:inner_epsilon_naive}.

\subsection{The benefits of balanced allocation}
This tightened analysis enables more precise error bounds while maintaining the correctness guarantees, and it also allows us to distribute the error budget between the different steps evenly. To illustrate the practical impact, consider a QBF formula with 2 quantifiers (also called a 2QBF formula), where each block has 5 variables, and a target error rate of $\epsilon = 0.3$. Using the geometric allocation, we would need approximately $240$ and $75$ repetitions for the first and second steps, respectively, resulting in a formula roughly $18{,}000$ times larger than the original. With the balanced allocation, however, we only need about $185$ and $48$ repetitions, yielding a formula that is about $8{,}900$ times larger than the original, a reduction of approximately $\times2$.

The source of this gain is depicted in Figure~\ref{fig:alloc-compare}. The geometric allocation spends half of the entire error budget on the outermost step alone, and is then forced to halve the budget again at every subsequent step, so that the deeper steps are left with an exponentially small share. The balanced allocation instead hands the same share to every step. Since by Proposition~\ref{clm:repetitionUNbound} the number of repetitions at step $i$ grows with $\log(1/\epsilon_i)$, and since the size of the resulting formula is governed by the \emph{product} of the repetitions over all the steps, the steps that are starved by the geometric allocation are precisely the ones that drive the blow-up.


\begin{figure}[h!]
\centering
\begin{tikzpicture}[x=1cm,y=1cm]
    \draw[->] (0,0) -- (8.1,0) node[below right, font=\small] {step $i$};
    \draw[->] (0,0) -- (0,3.15) node[above, font=\small] {$\epsilon_i$};

    \foreach \lab/\y in {0.025/0.6, 0.05/1.2, 0.075/1.8, 0.1/2.4}{
        \draw (-0.1,\y) -- (0.1,\y);
        \node[left, font=\small] at (-0.15,\y) {$\lab$};
        \draw[gray!25] (0.1,\y) -- (6.7,\y);
    }

    \foreach \lab/\x in {1/1.0, 2/2.7, 3/4.4, 4/6.1}{
        \draw (\x,-0.1) -- (\x,0.1);
        \node[below, font=\small] at (\x,-0.15) {$\lab$};
    }

    \draw[very thick, blue!65!black] (1.0,2.5337) -- (6.1,2.5337);
    \foreach \x in {1.0,2.7,4.4,6.1}{\filldraw[blue!65!black] (\x,2.5337) circle (2.2pt);}
    \node[font=\small, blue!65!black] at (3.55,2.86) {balanced allocation};

    \draw[very thick, red!70!black, densely dashed] (1.0,2.4) -- (2.7,1.2) -- (4.4,0.6) -- (6.1,0.3);
    \foreach \p in {(1.0,2.4),(2.7,1.2),(4.4,0.6),(6.1,0.3)}{\filldraw[red!70!black] \p circle (2.2pt);}
    \node[font=\small, red!70!black] at (5.15,1.35) {geometric allocation};

\end{tikzpicture}
\caption{The inner-error budget $\epsilon_i$ that each allocation assigns to every step, for a QBF formula with $d=4$ quantifiers and an overall error bound of $\epsilon=0.2$. The balanced allocation assigns every step the same budget $1-(1-\epsilon)^{1/2}\approx 0.106$, whereas the geometric allocation starts at $\epsilon/2 = 0.1$ and halves the budget at every step, down to $0.0125$ at the innermost one. Here the balanced allocation is the more generous of the two at every single step, which by Proposition~\ref{prop:benefit} is the case for every formula with $4$ or fewer quantifiers.}
\label{fig:alloc-compare}
\end{figure}


While the choice of which inner-error allocation method gives better results is perhaps formula dependent, we next prove that the balanced allocation method grants, for QBF formulas with 4 or fewer quantifiers, a larger error budget to each and every step than the geometric allocation does. This means that using balanced allocation on such QBF formulas, results in smaller output formulas. Formally, we have the following.

\begin{proposition}\label{prop:benefit}
Given a formula with $4$ or fewer quantifiers, let $\epsilon_i$ be the inner-error value allocated by the bounded balanced allocation of Theorem~\ref{thm:boundedalloc}, and let $\zeta_i$ be the inner-error value allocated by the geometric allocation. Then we have that $\zeta_i<\epsilon_i$ for every $1\leq i \leq 4$.
\end{proposition}

\begin{proof}
For every $1\leq i \leq 4$, we have by definition that $\zeta_i = \frac{\epsilon}{2^i} \leq \frac{\epsilon}{2}$, and that $\epsilon_i = 1-(1-\epsilon)^{1/l}$ where $l = \lceil\frac{d}{2}\rceil$. Since $d\leq 4$ we have that $l\in\{1,2\}$, and we treat the two cases separately.

For $l=1$ we have $\epsilon_i = \epsilon$, and therefore $\zeta_i \leq \frac{\epsilon}{2} < \epsilon = \epsilon_i$, as required.

For $l=2$ we have $\epsilon_i = 1-\sqrt{1-\epsilon}$, so that it suffices to show that $\frac{\epsilon}{2} < 1-\sqrt{1-\epsilon}$. The inequality is equivalent to $\sqrt{1-\epsilon} < 1-\frac{\epsilon}{2}$. Both sides being positive, we may square them, and get the equivalent inequality $1-\epsilon < 1-\epsilon+\frac{\epsilon^2}{4}$, which in turn amounts to $0 < \frac{\epsilon^2}{4}$, and is always \true since $\epsilon>0$.
\end{proof}

We note that the restriction to 4 or fewer quantifiers in Proposition~\ref{prop:benefit} is not an artifact of the proof. For $d\geq 5$, and hence $l\geq 3$, we have $\epsilon_i \approx \frac{\epsilon}{l} < \frac{\epsilon}{2} = \zeta_1$, so from that point on the geometric allocation does grant a larger budget to the outermost step, and a step-by-step comparison is no longer possible.

This restriction, however, costs us little in practice, for two reasons. First, many of the QBF families that are of interest in practice have only a small number of quantifier alternations, often no more than the four that Proposition~\ref{prop:benefit} covers. Second, and more fundamentally, Toda's reduction is polynomial only when the number of alternations is constant, as its running time is exponential in $d$. The number of alternations must therefore be bounded in any case for the whole approach to be meaningful, and formulas with many alternations are far out of reach of a Toda-based algorithm regardless of which allocation is used. Consequently, we regard the range covered by Proposition~\ref{prop:benefit} as the range that matters, and do not pursue a general algebraic comparison for larger $d$.

\section{Handling the Valiant-Vazirani Factor}
\label{sec:replacing-vv}

The Valiant-Vazirani isolation lemma, with the $p=1/8n$ success probability of Equation~\ref{eq:vvsingle}, lies at the heart of Toda's reduction, and is one of the main causes for the number of repetitions. That its guarantee decays like $1/n$ is not an accident of the construction. Dell, Kabanets, van Melkebeek and Watanabe~\cite{DellKMW13} study procedures that isolate a \emph{unique} solution, and show two things. First, in a black-box setting, where the procedure has to commit to a distribution over hash functions before seeing the formula, as Valiant-Vazirani does, the success probability of \emph{any} such procedure is $O(1/n)$, and this is unconditional. Second, if the procedure is allowed to inspect the formula, and even to use non-uniform advice, then a success probability above $2/3$ exists if and only if $\mathrm{NP}\subseteq\mathrm{P}/\mathrm{poly}$, which would collapse the polynomial hierarchy. Isolation is thus a dead end.

For Toda's reduction, however, we do not need to isolate a unique solution. We need only to ensure that with a reasonable probability, whenever any solutions exist, their count is odd. This is a much weaker requirement than isolation, and in this relaxed setting a \emph{constant} success probability is attainable. The gain is substantial: the number of repetitions at every step of the reduction is proportional to $1/\ln(1-p)$, so replacing a probability of order $1/n$ by a constant divides it by a factor linear in $n$.

\paragraph{Related work.}
The observation that Toda's formulation~\cite{Toda91} asks only for an odd count but takes its probability from the isolation lemma is not ours, and neither are most of the constructions that exploit it. Naor and Naor~\cite{NaorNaor1993} isolate an odd number of elements implicitly, in the course of their work on small-bias probability spaces. Gupta~\cite{Gupta1998} makes the observation explicit, gives a tighter analysis of their construction, reaching a constant success probability that approaches $1/4$, and applies it to improve, among other results, the bounds in Toda's theorem. Naik, Regan and Sivakumar~\cite{NRS95} attain a success probability arbitrarily close to $1/2$, using concatenated error correcting codes. Goldreich~\cite{Goldreich_book-F} suggests, in the same spirit, replacing the isolation lemma by a hash function that is based on a small-biased generator, such as the powering construction $\x\mapsto\langle s^{\x},y\rangle$ over $GF(2^m)$ of Alon, Goldreich, H{\aa}stad and Peralta~\cite{AlonGHP1992}.

Of the two techniques that we present below, the second is an instance of the construction of~\cite{NRS95}, and we claim nothing for it beyond the choice of its parameters and its encoding as a Boolean formula. The first, which we call modular addition, is ours, although the limit theorem on which it rests is not. The technique could be seen as a general amplification technique, which takes any hash function carrying a weak guarantee and produces one whose guarantee approaches $1/2$.

The value $1/2$ that recurs above is not an accident either. All of these sieves only ever \emph{remove} solutions, so that the solutions of the sieved formula are a subset of the original ones, and for such reductions~\cite[Corollary 7]{DellKMW13} show that even in the parity setting a success probability above $2/3$ would imply $\mathrm{NP}\subseteq\mathrm{P}/\mathrm{poly}$. In Section~\ref{sec:half-barrier} we prove a sharper and unconditional statement: no distribution over hash functions whatsoever, however large and however described, can exceed $1/2$ by more than a doubly exponentially small amount, and $1/2$ is essentially attained. This ceiling does not, however, bound the reduction as a whole. Toda's amplification combines $k$ sieved copies over \emph{fresh} variables, as in Equation~\ref{eq:endF_i}, so that the solutions of the combined formula are no longer a subset of the solutions of the original one. That is exactly what places it outside the scope of the restriction above, and what allows the reduction to reach an error as small as we please.

We proceed as follows. We start, in Section~\ref{sec:sweep}, with what Valiant and Vazirani actually prove, which is more than the textbook form of Section~\ref{sec:prelim} states. A single hash function reaches $1/2n$ rather than $1/8n$, and the \emph{sweep} form of the lemma, in which all prefixes are kept, reaches $1/2$ at the price of repeating the formula once per prefix, which turns out to be a factor of $2\ln 2$ better still. It is the best that can be had from the papers of Valiant and Vazirani and of Toda, and it is what we compare against. The two sieves that follow reach a comparable constant without repeating $F$ at all, which is what removes the factor of $n$. In Section~\ref{sec:Modualr_Addition} we present \emph{modular addition}. We apply it to the Valiant-Vazirani hash function, obtaining a success probability of almost $1/2$ at a reasonable cost with respect to the size of the formula in question, but the technique applies to any hash function, and, as we discuss there, to moduli other than $2$. In Section~\ref{subsec:multilinear-hash} we use a sieve of~\cite{NRS95} that achieves a \emph{constant} success probability, independent of the formula size, with a Boolean formula whose size is comparable to, and asymptotically smaller than, that of Valiant-Vazirani. We describe both techniques and how to translate them to Boolean formulas, and show that both significantly decrease the number of repetitions. Finally, in Section~\ref{sec:half-barrier}, we show that neither of them, nor any other sieve of this kind, can be pushed beyond $1/2$.

\subsection{Using the Sweep}\label{sec:sweep}

The $1/8n$ of Theorem~\ref{thm:vv} is the textbook form of the lemma, and it is weaker than what Valiant and Vazirani prove. They do not draw a single hash function. They draw $n$ vectors $\vec{w}_1,\ldots,\vec{w}_n$ from $\{0,1\}^n$ at once, and consider all $n+1$ \emph{prefixes} of that list. For $0\leq m\leq n$ let $h^{(m)}$ be the member of $H_{n,m}$ whose matrix consists of the rows $\vec{w}_1,\ldots,\vec{w}_m$ and whose offset is zero, so that $h^{(0)}$ is the empty prefix, which imposes no constraint. Their theorem is the following.

\begin{theorem}[{Valiant--Vazirani~\cite{ValiantVazirani}}]\label{thm:vvprefix}
Let $F:\{0,1\}^n\rightarrow\{0,1\}$ be a Boolean formula and write $\widetilde{F}^{(m)}(\x) = F(\x)\land h^{(m)}(\x)$. Then:
(i) if $\exists\x\, F(\x) = \mathit{False}$, then $\#\widetilde{F}^{(m)} = 0$ for every $m$;
(ii) if $\exists\x\, F(\x) = \mathit{True}$, then
\[
\Pr\Bigl[\,\#\widetilde{F}^{(m)} = 1 \ \text{ for some } 1\leq m\leq n \,\Bigr] \;\geq\; \frac{1}{4},
\]
and, when the $\vec{w}_i$ are drawn linearly independent,
\[
\Pr\Bigl[\,\#\widetilde{F}^{(m)} = 1 \ \text{ for some } 0\leq m\leq n \,\Bigr] \;\geq\; \frac{1}{2}.
\]
\end{theorem}

Which prefix is the successful one depends on the number of solutions of $F$, which is not known in advance, so the theorem can be used in two ways. One may draw a single prefix at random and use it alone. Valiant and Vazirani themselves draw one of the $n$ prefixes $h^{(1)},\ldots,h^{(n)}$, which gives the $1/4n$ that their paper states, twice the textbook value. The alternative is to keep every prefix, which is what Toda does~\cite{Toda91}, and what we call the \emph{sweep}: at a step with $n$ inner variables one conjoins $F$ with each of $h^{(0)},\ldots,h^{(n)}$ separately, and asks that at least one of the $n+1$ resulting formulas has an odd number of solutions. The guarantee is then the full $1/4$ of the theorem, at the price of $n+1$ copies of $F$ in place of one.

The prefix $h^{(0)}$ imposes no constraint, so the copy that it produces is the formula $F$ itself, whatever the vectors $\vec{w}_1,\ldots,\vec{w}_n$ happen to be, and there is no reason to produce it in every repetition. Toda and Ogiwara~\cite{TodaOgiwara1992} produce it once for the whole step, as a single \emph{empty repetition} shared by all the others, and let the repetitions range over $h^{(1)},\ldots,h^{(n)}$ alone, which leaves $n$ copies of $F$ per repetition in place of $n+1$. We do the same, and with the empty repetition in place the single-hash draw stays over the $n$ prefixes $h^{(1)},\ldots,h^{(n)}$. Finally, drawing the vectors linearly independent costs nothing and raises the guarantee of the sweep to $1/2$, and with it that of the single-hash form to $1/2n$. We use these values from here on.

The comparisons in this section ask which of two constructions produces the smaller formula. The empty repetition adds a single copy of $F$, which is negligible next to the repetitions and does not change any of these answers, so we leave it out and keep the comparisons simple to state.

We presented the reduction in Section~\ref{sec:todaBrief} one hash function at a time, and not one prefix at a time, deliberately: a step that applies a single sieve to a single copy of the formula is the form in which the sieve can be exchanged for another, which is exactly what the rest of this section does. That presentation is, however, a weaker use of Valiant-Vazirani than Toda's own, and comparing new sieves against it would be inaccurate. We therefore first introduce the sweep, and use it as the baseline for the comparisons in Sections~\ref{sec:compareMA} and~\ref{sec:compareMultilinear}.

\begin{proposition}\label{prop:sweep}
Fix a step with $n$ inner variables, $s$ outer variables and inner-error $\epsilon_i$. Let $Size_{_{VV}}$ be the size of the formula produced by the single-hash form and $Size_{_{SW}}$ the size produced by the sweep. Then
\[
\frac{Size_{_{SW}}}{Size_{_{VV}}} \;=\; n\cdot\frac{\ln\left(1-\frac{1}{2n}\right)}{\ln\left(\frac12\right)}
\;\xrightarrow[n\to\infty]{}\; \frac{1}{2\ln 2} \;\approx\; 0.721 .
\]
\end{proposition}

\begin{proof}
By Proposition~\ref{clm:repetitionUNbound} the single-hash form needs $k_{_{VV}} = \log_{1-\frac{1}{2n}}(\epsilon_i/2^s)$ repetitions, each of size $|F|+|h|$, so $Size_{_{VV}} = k_{_{VV}}(1+\alpha)|F|$ with $\alpha = |h|/|F|$ as in Section~\ref{sec:compareMA}. The sweep succeeds with probability $1/2$ and therefore needs $k_{_{SW}} = \log_{\frac12}(\epsilon_i/2^s)$ repetitions, but each repetition now carries $n$ copies of the formula together with its hash, so $Size_{_{SW}} = k_{_{SW}}\cdot n(1+\alpha)|F|$. Taking the ratio, $\epsilon_i$, $s$, $|F|$ and $\alpha$ all cancel:
\[
\frac{Size_{_{SW}}}{Size_{_{VV}}}
= \frac{\log_{\frac12}(\epsilon_i/2^s)}{\log_{1-\frac{1}{2n}}(\epsilon_i/2^s)}\cdot n
= n\cdot\frac{\ln\left(1-\frac{1}{2n}\right)}{\ln\left(\frac12\right)} .
\]
Since $\ln(1-x) = -x + O(x^2)$, the numerator is $-\frac{1}{2n}+O(n^{-2})$, and the expression tends to $\frac{1}{2\ln 2}$.
\end{proof}

The sweep is thus smaller by a factor of $2\ln 2 \approx 1.386$. The gain is modest, and the reason is the following. The sweep spends $n$ copies of $F$ in order to reach a success probability of $1/2$. Spending the same $n$ copies on $n$ independent repetitions of the single-hash form instead, each succeeding with probability $1/2n$, already reaches
\[
1-\left(1-\frac{1}{2n}\right)^{n},
\]
which is bounded from below by $1-e^{-1/2}\approx0.393$. The single-hash form therefore never falls far behind the sweep once both are given the same number of copies of the formula: the two are competing at $0.393$ against $0.5$.

Neither avoids the $n$ copies of $F$. What we would want is a sieve that reaches a high success probability \emph{without} repeating $F$ many times, so that the factor of $n$ is removed from the construction altogether. That is what the next two sections provide.

\subsection{Amplifying Valiant-Vazirani by Modular Addition}\label{sec:Modualr_Addition}

A main step in Toda's reduction is sieving the ongoing formula $F_i$ with a hash function in the hope of isolating a single assignment. To do this we replace $F_i$ with $F_i\land h$ where $h$ is a randomly chosen hash function. To increase the probability of doing so successfully, we introduce duplications of $F_i\land h$, each time with $h$ randomly chosen, and combine these together. Since we have such a low guarantee on the hash function $h$, we introduce many such duplications, and consequently the resulting formula is long.
To tame the number of duplications, we present a more general technique for taking objects that have low guarantee and constructing from them other objects with higher guarantee (in our case the objects are the hash functions). We call this method \emph{modular addition}, which, as far as we know, is novel in this context. The method is based on an idea that is easy to explain: a simple way to generate an unbiased bit from a biased bit is to sample the biased bit independently $n$ times and take the parity of the sum of all samples (i.e. $1$ if odd, $0$ if even). The bigger $n$ is, the closer one gets to a uniform sample. We note that this method can also be seen as an alternative for Von Neumann's well-known method for obtaining the same thing~\cite{bib:Von-Neumann}. 
This example is of course a special case of a more general theorem: 

\begin{restatable}{theorem}{maintheorem}\label{thm:main}
Let $ X $ be a discrete random variable over a bounded range such that $X=0$ and $X=1$ have positive probabilities each. Let $X_j$ denote the independent $j$-th roll of $X$, and let $k>0$ be a positive integer.
Then for every $t$ such that $0\leq t \leq k-1$ the following holds:
\begin{equation*}
\lim_{n \to \infty} \text{Pr}\left((\sum_{j=1}^{n} X_j) \bmod{k} = t\right) = 1/k
\end{equation*}
\end{restatable}

Theorem~\ref{thm:main} can be seen as a corollary in~\cite{DiaconisModularAddition} (Theorem 1, p.23), a study in the area of group representations in probability. The proof there, however, may not be self-contained. We therefore provide a proof that is perhaps also simpler. See proof outline below, and Appendix~\ref{sec:ModularAddition} for full detail. After completing this work, we became aware of an independent arXiv pre-print in physics that establishes the same result~\cite{gueron2025}.

\begin{proof}[Proof Outline]
Our proof makes use of the following definition.
For a positive integer $k$, let $d_{k}$ be the first root of order $k$ of $1$. That is, 
$d_{k} = e^{\left(\frac{2\pi i}{k}\right)}$ where $i$ denotes the imaginary unit $\sqrt{-1}$. For simplicity denote $Pr(X \bmod{k}=j)$ by $p_j$.

We first show by standard combinatorial techniques that for every $t\in[k]$:

\begin{equation}\label{eq:prob_to_multinom}
\text{Pr}\left(\Sumnk{n}{k} = t\right) = (1/k)\sum_{m=0}^{k-1}d_{k}^{-mt}\left(\sum_{j=0}^{k-1} d_k^{mj}p_j\right)^n
\end{equation}

Then, to complete the proof we show that 

\begin{equation*}
 \lim_{n\to\infty} (1/k)\sum_{m=0}^{k-1}d_{k}^{-mt}\left(\sum_{j=0}^{k-1} d_k^{mj}p_j\right)^n = 1/k
\end{equation*}

The latter follows from two observations. First, for $m=0$, the term $d_k^{-0t}\left(\sum_{j=0}^{k-1} d_k^{0j}p_j\right)^n = \left(\sum_{j=0}^{k-1} p_j\right)^n = 1$. Second, for every $m \in \{1,2,\ldots,k-1\}$, we have that $\left|\sum_{j=0}^{k-1} d_k^{mj}p_j\right| < 1$ due to the triangle inequality holding strictly when $p_0, p_1 > 0$, causing these terms to vanish as $n \to \infty$. 
\end{proof}

\subsubsection{The Modular Addition Technique}\label{sec:sModAdd}

We next show how to apply modular addition to construct a hash function and a corresponding formula. Our definition goes as follows.

\begin{definition}[Modular addition hash function]\label{def:HMA}
Let $H_n$ be the family of hash function used by Valiant-Vazirani.
For $l > 0$, we define $H_{_{MA}}^l$ the family of modular addition hash functions, by
\[
H_{_{MA}}^l = \{h_1 + h_2 + \cdots + h_l \mid h_i \in H_{n}\}
\]
where each $h_i$ is independently sampled from $H_{n}$, in the single-hash form of Section~\ref{sec:sweep}.
\end{definition}

Next we show how to describe $H_{_{MA}}^l$ as a Boolean formula.
For $l>0$ let $D_l$ be the following Boolean formula: $D_l\equiv F(\x)\land h_{_{MA}}^l(\x)$ where $h_{_{MA}}^l$ is randomly sampled from $H_{_{MA}}^l$. 

\begin{proposition}\label{prop:MAuse}
If $F$ is satisfiable then $Pr (\parity \x D_l\equiv true)=1/2$ when $l\rightarrow\infty$. Otherwise $Pr(\parity \x D_l\equiv false)=1$ for every $l$.
\end{proposition}

\begin{proof}
Let $\varphi$ be a satisfiable Boolean formula over the variables $\x$. Let $\Psi$ be a family of Boolean formulas over the same variables as $\varphi$. We define $X_{\Psi}$ as a discrete random variable over the range $\{0,\cdots, 2^{|\x|}\}$ as follows. For every $i$ in the range, we define $Pr(X_{\Psi}=i) = Pr_{\psi \in \Psi} (\#\{\x : \varphi(\x) \land \psi(\x)\} = i)$. Note that by this definition $Pr (X_\Psi\mod{2} = 1)$ captures the probability of sampling $\psi$ from $\Psi$, such that $\parity \x[\varphi(\x) \land \psi(\x)] \equiv \true$.

Now, for a given $l$, consider $\varphi$ as our formula $F$, and the family 
$\Psi = \{h_1+h_2+\cdots+ h_l | h_i\in H_{n}\}$.  Sampling $\psi\in\Psi$ is then identical to sampling $l$ functions $h_i\in H_{n}$ independently, which means that $X_\Psi = X_{H_{n}+ \overset{l \text{ times}}{\cdots} + H_{n}}$.
Claim~\ref{clm:smallProp} below shows us that 

$$X_\Psi = X_{H_{n}+ \overset{l \text{ times}}{\cdots} + H_{n}} = X_{H_{n}} + \overset{l \text{ times}}{\cdots} +X_{H_{n}}$$ 

Therefore we can use Theorem~\ref{thm:main} to have that if $F$ is satisfiable then $X_{\Psi}$ approaches a uniform distribution. This means that the probability of sampling $\psi=(h_1 + \cdots +h_l)\in\Psi$ such that $\parity \x D_l = \parity \x[F(\x) \land \psi(\x)] \equiv \true$ approaches $1/2$ as $l$ increases. On the other hand if $F$ is unsatisfiable then $F(\x) \land \psi(\x)$ has zero satisfying assignments which means that $\parity \x D_l = \parity \x[F(\x) \land \psi(\x)] \equiv \false$  for every $l$. \end{proof}

The Claim below completes the Proof for Proposition~\ref{prop:MAuse}.

\begin{claim}\label{clm:smallProp}
    $X_{H_{n}+ \overset{l \text{ times}}{\cdots} + H_{n}} = X_{H_{n}} + \overset{l \text{ times}}{\cdots} +X_{H_{n}}$.
\end{claim}

\begin{proof}
We show for $l=2$, the extension to larger $l$ is straightforward. For $l=2$ we have.
\begin{align*}
X_{H_{n} + H_{n}} &= \#[F(\x) \land (h_1(\x) + h_2(\x))] \\
&= \#[F(\x) \land ((h_1(\x) \land x_{new}) \lor (h_2(\x) \land \neg x_{new}))] \\
&= \#[(F(\x) \land h_1(\x) \land x_{new}) \lor (F(\x) \land h_2(\x) \land \neg x_{new})] \\
&= \#[F(\x) \land h_1(\x)] + \#[F(\x) \land h_2(\x)]
= X_{H_{n}} + X_{H_{n}}
\end{align*}
\end{proof}

Our modular addition technique is therefore to construct $D_l$ for large enough $l$ and use it in place of the Valiant-Vazirani sieve, where Proposition~\ref{prop:MAuse} ensures a success probability close to $1/2$.

We record one requirement on $l$ that we use below. Specializing Equation~\ref{eq:prob_to_multinom} to $k=2$ gives the exact success probability of $D_l$ as $\frac{1}{2}\left(1-(1-2p)^l\right)$, where $p$ is the probability that a single $F\land h$ has an odd number of solutions. In practice $p$ itself is only known through a lower bound, and substituting that lower bound is legitimate only if the expression is monotone in $p$ over the whole range that the bound leaves open. Since $\frac{\partial}{\partial p}\frac{1}{2}\left(1-(1-2p)^l\right) = l(1-2p)^{l-1}$, this holds for every $p\in[0,1]$ when $l$ is \emph{odd}, whereas for even $l$ the expression decreases again once $p$ exceeds $1/2$. We therefore always take $l$ to be odd, which costs nothing, and this is what allows us to substitute the Valiant-Vazirani guarantee of $1/2n$ for $p$ in the comparison below.
Note that $D_l$ contains $l$ repetitions of the hash function $h$, but no repetitions of $F$. Since $F$ tends to be much bigger than $h$, if we manage to tame $l$, we can drastically decrease the size of the formula. We do note that the optimal $l$ in terms of the resulting formula size depends on many factors, such as the current formula size and the number of variables. Nevertheless, we show below that even choosing $l=n$ can have a huge effect on the resulting formula size.

\subsubsection{Comparison with the Sweep}\label{sec:compareMA}

To get a better sense of how much modular addition affects the size of the resulting formula, we did a comparison using the analysis of the number of repetitions in Proposition~\ref{clm:repetitionUNbound}.
For that, we set a formula $F$ with a number of outer variables $s$, and a number of inner variables (to be hashed) $n$. We also assume $h$ to be a Valiant-Vazirani hash function that provides a guarantee of at least $1/2n$, and that the size of $h \sim H_{n,m}$, $|h|$, is the worst possible. Note that in this case $|h|=n^2$, since $h \sim H_{n,n}$ and each of the $n$ constraints contains all of the $n$ inner variables (see Section~\ref{sec:vvthm}). Since the size of $F$ and $h$ varies with $n$, we set a parameter $\alpha= |h|/|F|$ that fixes the ratio between the size of $|h|$ and $|F|$, so that $|h| = \alpha|F|$.
We denote by $Size_{_{SW}}$ the size of the resulting formula after using the sweep of Section~\ref{sec:sweep}, on a given formula $F$, to reach the inner-error guarantee $\epsilon_i$. We compare against the sweep, and not against a single hash function, because by Proposition~\ref{prop:sweep} the sweep is the stronger of the two forms of Valiant-Vazirani and is therefore what we compare against. We denote by $Size_{_{MA}}$ the size of the resulting formula after using $H_{_{MA}}^n$ on $F$, instead of Valiant-Vazirani, and reaching the same success probability of $\epsilon_i$.
Note that $H_{_{MA}}^n$ is the family $H_{_{MA}}^l$ as described in Definition~\ref{def:HMA}, setting $l=n$ for a clean computation rather than for optimality reasons, where in accordance with the remark in Section~\ref{sec:sModAdd} we read $l=n$ as the nearest odd integer when $n$ is even, which affects none of the figures below. Nevertheless, choosing $l=n$ gives a success probability, as derived from Equation~\ref{eq:prob_to_multinom}, of $p_{_{MA}} = \frac{1}{2} - \frac{1}{2}(1-\frac{1}{n})^n \geq 0.316$ for all $n$.

We are interested to see how the ratio between $Size_{_{SW}}$ and $Size_{_{MA}}$ changes, for different number of inner variables $n$, and different $\alpha$. In all this, we ignore formula parameters such as the $+1$ operators and the formula symbols such as Boolean operators and parentheses, which do not contribute much to the size of either formula.

Using Proposition~\ref{clm:repetitionUNbound} together with Proposition~\ref{prop:sweep}, we compute $Size_{_{SW}}$ to be

\begin{equation}\label{eq:sizeVV}
Size_{_{SW}} = k_i\cdot n(|F|+|h|)= n\,k_i(1+\alpha)|F| = n\left(\log_{\frac{1}{2}}(\epsilon_i/2^s)\right)(1+\alpha)|F|
\end{equation}

To compute $Size_{_{MA}}$ we first note that the size of each $h_{_{MA}} \sim H^n_{MA}$ is $n$ times larger than $h$. As such each of the repetitions is of size $|F|+n \alpha|F| = (1+n\alpha)|F|$. The number of repetitions according to Proposition~\ref{clm:repetitionUNbound} is $\log_{1-p_{_{MA}}}(\epsilon_i/2^s)$. From these we compute:

\begin{equation}\label{eq:sizeMA}
Size_{_{MA}} = (\log_{1-p_{_{MA}}}(\epsilon_i/2^s))(1+n\alpha)|F|
\end{equation}

In total, this gives a ratio of

\[
\frac{Size_{_{SW}}}{Size_{_{MA}}} = \frac{\ln(1 - p_{_{MA}})}{\ln\left(\frac{1}{2}\right)} \cdot \frac{n(1 + \alpha)}{1 + n\alpha}
\]

While the ratio $\frac{Size_{_{SW}}}{Size_{_{MA}}}$ is somewhat unpleasant to analyze, we note that the parameters $\epsilon_i$, $|F|$ and $s$ cancel, and we are left with an expression that depends on $n$ and $\alpha$. Since expressing the limit for various $\alpha$ parameters, as $n$ goes to infinity, says little about actual QBF formulas, we depict a graph, in Figure~\ref{fig:ratio_naive_vs_ma}, that describes the ratio $\frac{Size_{_{SW}}}{Size_{_{MA}}}$ for various $\alpha$ values.

\begin{figure}[h!]
    \centering
    \includegraphics[width=1.0\linewidth]{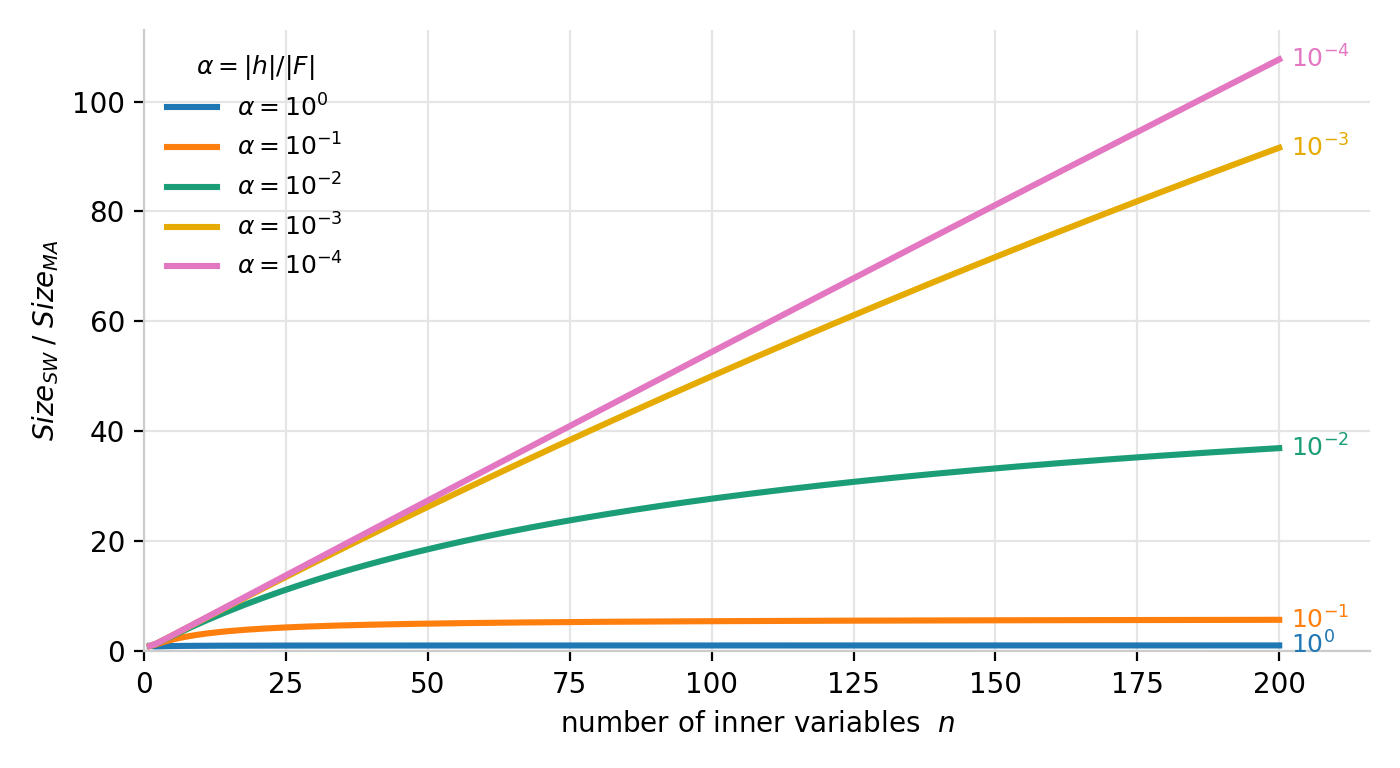}
    \caption{Ratio of the formula size between the Valiant-Vazirani sweep and the Modular addition approach.}
    \label{fig:ratio_naive_vs_ma}
\end{figure}

As we can see, when $|F|\approx|h|$ there is no noticeable improvement in using modular addition. However, the gap becomes substantial as $\alpha$ gets smaller and the number of inner variables $n$ grows. In addition, recall that with every step in the reduction the formula size grows substantially, so $|F|$ becomes much larger than $|h|$. For example, a formula in later steps with $n=125$ inner variables can easily become $1000$ times larger than the hash function. Then using modular addition results in a formula roughly $60$ times smaller than when using the sweep. We note that model counters nowadays can handle such formulas of millions of clauses.

\subsection{Parity Isolation via Multilinear Hashing}
\label{subsec:multilinear-hash}

While the modular addition technique showed us how to utilize the Valiant-Vazirani hash function to construct better hash functions, we next follow a suggestion from~\cite{Goldreich_book-F} and employ a sieve that is not an isolation lemma at all, but is designed directly to make the number of solutions odd. Such a sieve achieves a \emph{constant} success probability, independent of formula size, and thereby removes the $1/n$ factor from the number of repetitions altogether. The construction that we use is an instance of a family introduced by Naik, Regan and Sivakumar~\cite{NRS95}. We discuss its origins, and its relation to the $\epsilon$-biased sample spaces of~\cite{AlonGHP1992,NaorNaor1993}, in Section~\ref{sec:multilinear-attribution}. We first describe the hash function, then prove its guarantee, and finally explain how to present it as a Boolean formula of a reasonable size.

\subsubsection{The Construction}\label{sec:multilinear-construction}

Our hash function uses arithmetic in a finite field. For a positive integer $m$, the Galois field $GF(2^m)$ is a finite field $\mathbb{F}$ of $2^m$ elements, each represented as a binary string of length $m$. We refer the reader to Appendix~\ref{app:field} for a concrete description of this field and its operations, and record here only the three properties that we use. First, the addition operation $+$ in $\mathbb{F}$ is bitwise XOR, so in particular $a+a=0$ for every $a$. Second, $\mathbb{F}$ has a multiplication operation $\cdot$ with a multiplicative identity $1_{\mathbb{F}}$, so that the power $a^e$ is well defined for every $a\in\mathbb{F}$ and every non-negative integer $e$, with $a^0 = 1_{\mathbb{F}}$. Third, we use $\langle A, B\rangle$ to denote the inner product modulo $2$ of two elements of $\mathbb{F}$: if $A = [a_1,\cdots,a_m]$ and $B = [b_1,\cdots,b_m]$ then $\langle A, B\rangle \equiv \bigoplus_{i=1}^{m}(a_i\land b_i)$, which is a single bit.

We also make use of the following standard result.

\begin{theorem}[Schwartz--Zippel~\cite{Schwartz1980,Zippel1979}]\label{thm:sz}
Let $P(z_1,\ldots,z_n)$ be a nonzero polynomial of total degree at most $d$ over a finite field $\mathbb{F}$. If $a_1,\ldots,a_n$ are chosen independently and uniformly at random from $\mathbb{F}$, then $\Pr[P(a_1,\ldots,a_n) = 0] \leq \frac{d}{|\mathbb{F}|}$.
\end{theorem}

We can now define our hash function. To sample one, we draw $n+1$ elements of $\mathbb{F}$ uniformly and independently at random: one element $a_i$ per inner variable $x_i$, and one further element $v$.

\begin{definition}[Multilinear hash function]\label{def:multilinearHash}
Let $a = (a_1,\ldots,a_n)$ and $v$ be elements of $\mathbb{F}$ as above. The multilinear hash function $h_{a,v}: \{0,1\}^n \to \{0,1\}$ is defined as:
\[
h_{a,v}(\x) = \left\langle\, \prod_{i=1}^{n} a_i^{x_i},\;\; v \,\right\rangle
\]
\end{definition}

Since every $x_i$ is either $0$ or $1$, the $i$-th factor of the product equals $a_i$ when $x_i=1$ and $1_{\mathbb{F}}$ when $x_i=0$. In other words, $\prod_{i=1}^{n} a_i^{x_i}$ is simply the field product of those $a_i$ for which $x_i=1$.

\subsubsection{The Probability Guarantee}\label{sec:multilinear-guarantee}

We then have the following.

\begin{theorem}
\label{thm:parity-isolation}
Let $\# [f(\x) \land h(\x)]$ denote the number of solutions to $f \land h$.
If $h = h_{a,v}$ is sampled uniformly as in Definition~\ref{def:multilinearHash} then we have:
\begin{enumerate}
    \item If $f$ is UNSAT: $\# [f(\x) \land h(\x)] = 0$ always (even).
    \item If $f$ is SAT: $\Pr[\# [f(\x) \land h(\x)] \text{ is odd}] \geq \frac{1}{2} - \frac{n}{2^{m+1}}$.
\end{enumerate}
\end{theorem}

\begin{proof}
Case 1 is immediate: if $f$ has no solutions, $f \land h$ has no solutions.
For Case 2, assume that $f$ is satisfiable with a solution set $S_f$, where $|S_f| \geq 1$. We analyze the parity of $\# [f(\x) \land h(\x)]$.

From Definition~\ref{def:multilinearHash} we have that $h_{a,v}(\x) = \langle \prod_{i=1}^{n} a_i^{x_i}, v \rangle$. Therefore we have:

\[
\# [f(\x) \land h(\x)] \bmod 2 =  \sum_{\x \in S_f} \left\langle \prod_{i=1}^{n} a_i^{x_i}, \, v \right\rangle \pmod{2}
\]

Since inner product is linear, i.e. $\langle A, v \rangle + \langle B, v \rangle = \langle A + B, v \rangle$ (where $+$ is XOR). We can pull the sum inside:
\[
\sum_{\x \in S_f} \left\langle \prod_{i=1}^{n} a_i^{x_i}, \, v \right\rangle = \left\langle \sum_{\x \in S_f} \prod_{i=1}^{n} a_i^{x_i}, \, v \right\rangle \pmod{2}
\]
where the sum on the right-hand side is addition in $\mathbb{F}$, that is, bitwise XOR of field elements.

Define $\beta = \sum_{\x \in S_f} \prod_{i=1}^{n} a_i^{x_i} \in \mathbb{F}$.
We have shown that $\# [f(\x) \land h(\x)] \bmod 2 = \langle \beta, v \rangle$.

Next, consider the polynomial in the $n$ formal variables $z_1,\ldots,z_n$ over $\mathbb{F}$:
\[
P_f(z_1,\ldots,z_n) = \sum_{\x \in S_f} \prod_{i=1}^{n} z_i^{x_i}
\]
This polynomial is determined by $S_f$, and $\beta = P_f(a_1,\ldots,a_n)$. We claim that $P_f$ is nonzero. Indeed, every $\x\in S_f$ contributes the monomial $\prod_{i=1}^{n} z_i^{x_i}$, which is \emph{multilinear}: as each $x_i$ is either $0$ or $1$, every variable occurs in it with exponent $0$, meaning that it is absent from the product, or $1$, meaning that it occurs exactly once. Two distinct elements $\x \neq \x'$ of $S_f$ differ in at least one coordinate $j$, and then exactly one of their two monomials contains $z_j$. Distinct elements of $S_f$ thus contribute distinct monomials, so that no cancellation can occur and each monomial appears in $P_f$ exactly once with coefficient $1_{\mathbb{F}}$. Since $S_f\neq\emptyset$, at least one monomial is present, and therefore $P_f$ is nonzero. Its total degree is at most $n$, which is attained when some $\x\in S_f$ has all of its coordinates equal to $1$.

Then to avoid $\# [f(\x) \land h(\x)] = 0\bmod 2$ we need to avoid $\beta = 0$. This happens exactly when $(a_1,\ldots,a_n)$ is a root of $P_f$. Since $P_f$ is a nonzero polynomial of total degree at most $n$, and since $a_1,\ldots,a_n$ are chosen uniformly and independently from $\mathbb{F}$, a field of size $2^m$, Theorem~\ref{thm:sz} gives us that:

\[
\Pr_{a_1,\ldots,a_n \in \mathbb{F}}[\beta = 0] \leq \frac{n}{2^m}
\]
we denote $(a_1,\ldots,a_n)$ a \emph{good} choice if $\beta \neq 0$. Note that most choices are good.

Next, assume that $(a_1,\ldots,a_n)$ is good, so $\beta$ is a nonzero $m$-bit string. For a uniformly random $v \in \mathbb{F}$, we show that $\Pr[\langle \beta, v \rangle = 1] = 1/2$. To see that, note that since $\beta \neq 0$, then $\beta$ has at least one nonzero bit. The inner product $\langle \beta, v \rangle$ is then a non-zero linear function of $v$ to $\mathbb{F}_2$, and as a result returns $0$ for exactly half of all values of $v$, and returns $1$ for the remaining half. In other words, we have:

\[
\Pr_{v \in \mathbb{F}}[\langle \beta, v \rangle = 1 \mid \beta \neq 0] = \frac{1}{2}
\]

Combining these results, gives us the following:
\begin{align*}
\Pr[\#[f(\x) \land h(\x)] \text{ is odd}] &= \Pr[\langle \beta, v \rangle = 1] \\
&\geq \Pr[\beta \neq 0] \cdot \Pr[\langle \beta, v \rangle = 1 \mid \beta \neq 0] \\
&\geq \left(1 - \frac{n}{2^m}\right) \cdot \frac{1}{2} = \frac{1}{2} - \frac{n}{2^{m+1}} \qedhere
\end{align*}
\end{proof}

Note that the success probability of $h_{a,v}(\x)$ approaches $1/2$ as $m$ increases, and that the field has to be only mildly larger than $n$. Concretely, setting $m = \lceil \log_2 n\rceil + c$ for an integer $c\geq 1$ gives $2^{m+1} \geq n2^{c+1}$, and hence a success probability of at least $\frac{1}{2} - 2^{-(c+1)}$: every additional field bit halves the distance to $1/2$. Taking $c=3$, that is
\[
m = \lceil \log_2 n\rceil + 3,
\]
is already enough for a probability $\geq 7/16 \approx 0.44$, which is close enough to $1/2$ for our purposes, and we use this setting from here on.

It remains to encode $h_{a,v}$ as a Boolean formula. As we detail in Appendix~\ref{sec:multilinearTechnique}, the product $\prod_{i=1}^{n} a_i^{x_i}$ is computed by a sequence of $n$ conditional multiplications by field constants, each costing at most $m(m+1)$ gates, followed by the inner product with $v$, which costs at most $m-1$ further gates. This gives a formula of size
\[
|h_{a,v}| \leq nm(m+1) + m - 1 = O(n\log^2 n),
\]
which uses $nm = O(n\log n)$ auxiliary variables for the running product. The bound is attained only when every multiplication matrix is all ones, and a typical draw costs about half of it. For comparison, the worst-case size of the Valiant-Vazirani hash function is $n^2$ (see Section~\ref{sec:vvthm}). The multilinear hash function is therefore \emph{asymptotically smaller} than the Valiant-Vazirani one, and comparing the worst-case sizes of the two it is smaller for every $n\geq133$, all while providing a constant success probability in place of $1/2n$. We quantify the effect on the size of the resulting formula in Section~\ref{sec:compareMultilinear}.

\subsubsection{Relation to Prior Work}\label{sec:multilinear-attribution}

Neither the hash function of Definition~\ref{def:multilinearHash} nor its analysis is new. Both are due to Naik, Regan and Sivakumar~\cite{NRS95}, who present the construction in the language of concatenated error correcting codes: a generalized Reed--Muller code over $GF(2^m)$, concatenated with the Hadamard code, the latter being the inner product of Definition~\ref{def:multilinearHash}. In their construction the input bits are grouped into blocks of $b$ bits, each block is read as an integer $e_t$, and the field element $\prod_t a_t^{e_t}$ is formed with one random element per block. They set $b=\lceil\log_2 n\rceil$, which keeps the number of random bits below $2n$ and the running time quasilinear. Our Definition~\ref{def:multilinearHash} is their construction with single-bit blocks, and Theorem~\ref{thm:parity-isolation} is their bound in that case. The proof is theirs as well: the minimum-distance property that they invoke for the Reed--Muller code is the Schwartz--Zippel lemma, and their concatenation with the Hadamard code is our inner product with $v$. We have only carried the argument out directly on $P_f$, which makes it self-contained given Theorem~\ref{thm:sz} alone, and spelled out why $P_f$ is nonzero. They note, moreover, that applied to a given SAT instance $\varphi$ their construction yields a formula of the form $\varphi\land\psi$. The explicit gate-level construction of $\psi$, and the accounting of its size, are what Appendix~\ref{sec:multilinearTechnique} supplies.

Our only substantive departure is the choice of single-bit blocks, which lowers the total degree of $P_f$ from roughly $n^2/\log n$ to $n$, and with it the field size from $m\approx2\log_2 n$ to $m\approx\log_2 n$. As a field multiplication costs quadratically in the bit-length, this shrinks the encoded hash function by roughly a factor of four, at the price of more random bits. That trade is the right one here and the wrong one in~\cite{NRS95}: they count random bits and Turing machine steps, whereas the only quantity that matters to us is the size of the Boolean formula that the model counter receives, which is also why we analyze the construction at the level of Boolean gates. The same consideration explains why we do not use the powering construction $h_{s,y}(\x)=\langle s^{\x},y\rangle$ of~\cite{AlonGHP1992}. There, $\x$ is read as an integer and placed in the exponent, so that $P_f(z)=\sum_{\x\in S_f}z^{\x}$ is univariate of degree up to $2^n-1$ and the field must satisfy $m>n$, giving a formula of size $\Theta(n^3)$ in place of $O(n\log^2 n)$. This gain is due to~\cite{NRS95}, who identify scaling down the field, by passing from a univariate to a multivariate polynomial, as the main point of their construction.

\subsubsection{Comparison with the Sweep}\label{sec:compareMultilinear}
We repeat the comparison detailed in Section~\ref{sec:compareMA} with the same notations and terminology. This time we compare the sweep with the multilinear hash function $h_{a,v}$. We set $m=\lceil\log_2 n\rceil+3$, which is a formula $h_{a,v}$ of size at most $nm(m+1)+m-1$ and a success probability of $p_{_{ML}} = 7/16$. Using Proposition~\ref{clm:repetitionUNbound}, the number of repetitions required to increase this probability to $\epsilon_i$ is $\log_{1-p_{_{ML}}}(\epsilon_i/2^s)$.

We denote by $Size_{_{ML}}$ the size of the resulting formula after using the hash function $h_{a,v}$, on the formula $F$, to reach the inner-error guarantee $\epsilon_i$.
To compute $Size_{_{ML}}$, denote first $r=|h_{a,v}|/|h|$, the ratio of the two worst-case sizes. Since the worst-case size of $h$ is $n^2$ (see Section~\ref{sec:vvthm}), we have that $r = \frac{nm(m+1)+m-1}{n^2}$. Therefore, while each repetition in the sweep is of size $n(|F|+|h|) = n(1+\alpha)|F|$, in the multilinear approach it is $|F|+r|h| = (1+r\alpha)|F|$. This gives us overall

\begin{equation}\label{es:SizeML}
Size_{_{ML}} = \left(\log_{1-p_{_{ML}}}(\epsilon_i/2^s)\right)(1+r\alpha)|F|.
\end{equation}

From Equation~\ref{eq:sizeVV} we then get:

\[
\frac{Size_{_{SW}}}{Size_{_{ML}}} = \frac{\ln(1-p_{_{ML}})}{\ln\left(\frac{1}{2}\right)} \cdot \frac{n(1 + \alpha)}{1 + r\alpha}
\]

The picture we get here is markedly different from the one we got for modular addition. There, the hash function was $n$ times larger than that of a single Valiant-Vazirani hash, so the second factor was substantial and the benefit depended strongly on $\alpha$. Here, in contrast, $r$ is close to $1$, and in fact tends to $0$ as $n$ grows, so that the second factor is close to $n$ for \emph{every} $\alpha$. The ratio is thus a product of two contributions that are both essential: a factor of $n$, which is the sweep's $n$ copies of the formula that the multilinear hash function does not need, and a constant $\ln(9/16)/\ln(1/2)\approx0.830$, which is the small price of settling for $7/16$ rather than the sweep's $1/2$. In other words, the size penalty of the hash function has essentially disappeared, and the multilinear hash function saves a factor of roughly $0.83n$ in the size of the resulting formula, nearly independently of how large $F$ is compared to $h$. Figure~\ref{fig:ratio_naive_vs_multilinear} depicts the ratio $\frac{Size_{_{SW}}}{Size_{_{ML}}}$ for various values of $\alpha$. As expected, the curves nearly coincide. The mild deviations reflect the fact that $r$ exceeds $1$ for small $n$ and drops below $1$ for large $n$, and the small jumps are the points where $\lceil\log_2 n\rceil$ increases.

\begin{figure}[h!]
    \centering
    \includegraphics[width=1.0\linewidth]{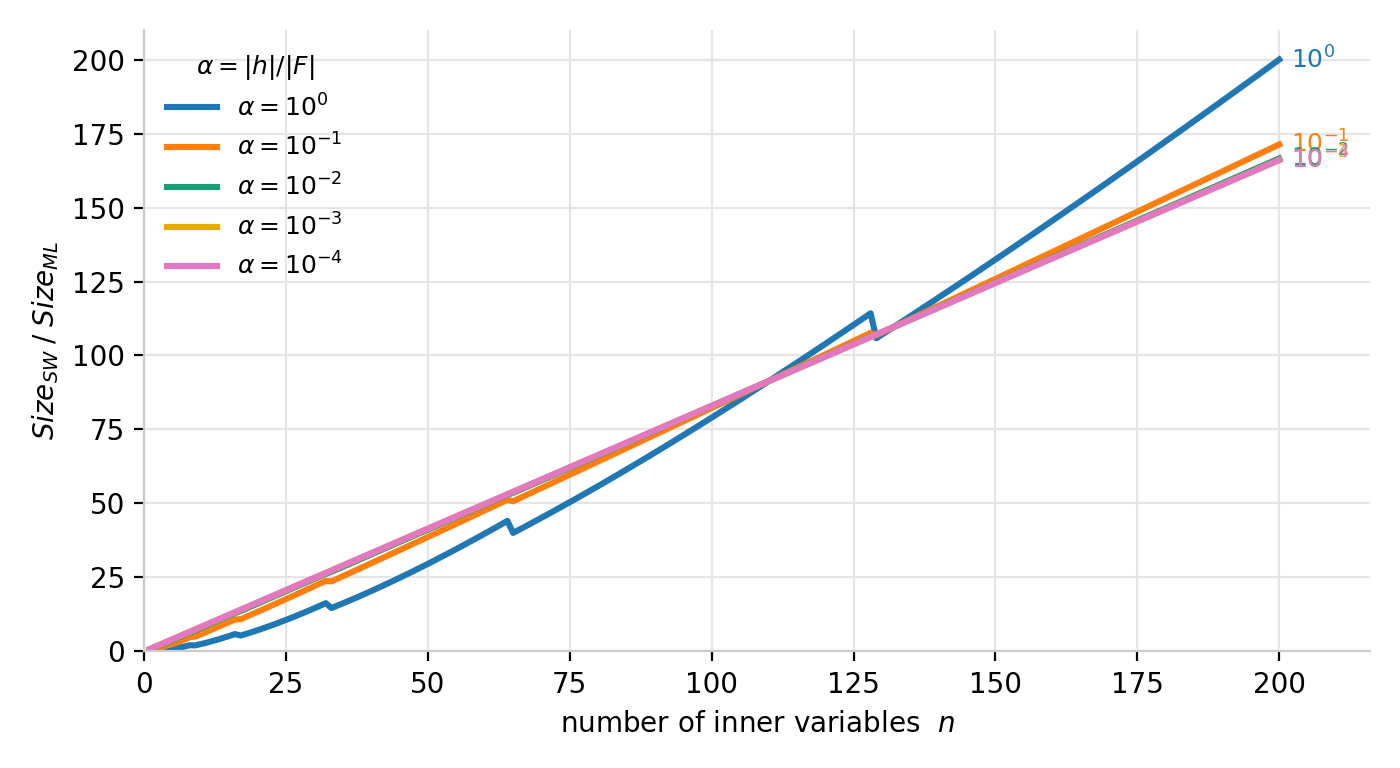}
    \caption{Ratio of the formula size between the Valiant-Vazirani sweep and the multilinear hashing approach.}
    \label{fig:ratio_naive_vs_multilinear}
\end{figure}







\paragraph{Multilinear Hashing vs. Modular Addition} Comparing Figures~\ref{fig:ratio_naive_vs_ma} and~\ref{fig:ratio_naive_vs_multilinear} shows that on size the multilinear hash function dominates modular addition on both counts: it attains a higher success probability, $7/16$ against the $0.316$ of $H^n_{_{MA}}$, with a smaller formula, $O(n\log^2 n)$ against $\Theta(n^3)$. Combining the two does not help either. Over $GF(2^{\lceil\log_2 n\rceil+c})$ the multilinear hash function has, as we saw, a success probability of at least $\frac12-2^{-(c+1)}$, that is $1-2p\leq2^{-c}$, so combining $l$ independent copies by modular addition gives a success probability of at least $\frac{1}{2}-\frac{1}{2}2^{-cl}$, which is exactly the guarantee that a single hash function over $GF(2^{\lceil\log_2 n\rceil+cl})$ already provides. Comparing the sizes of the two, the combination is cheaper only when $c\sqrt{l}>\lceil\log_2 n\rceil$, which for realistic $n$ occurs only at success probabilities so close to $1/2$ that the saving is worthless.

Modular addition nevertheless remains of interest, for two reasons that have nothing to do with size. The first is structural. It consists purely of linear constraints over the inner variables, with no field arithmetic and no auxiliary field elements, whereas the multilinear hash function requires $n$ successive field multiplications and $O(n\log n)$ auxiliary variables. Model counters are highly affected by the structure of a formula and not only by its size, and in our preliminary experiments, formulas that combined modular addition with the multilinear hash function were counted faster than a multilinear hash function carrying the same guarantee, despite being larger. We do not have an explanation that we can support with evidence. Section~\ref{subsec:experiments} measures the effect on the benchmark, and we leave its explanation to future work.

The second reason is that modular addition is not tied to parity. Theorem~\ref{thm:main} is stated for an arbitrary modulus $k$, and Proposition~\ref{prop:MAuse} used it only with $k=2$. The same argument gives a number of solutions that is uniform modulo $k$, and hence non-zero modulo $k$ with probability approaching $1-1/k$, that is $2/3$ for $k=3$ and $4/5$ for $k=5$. This is well beyond the $1/2$ that the multilinear hash function of Section~\ref{subsec:multilinear-hash} cannot pass. It matters because Toda's reduction is itself not tied to parity: in Section~\ref{sec:modp} we generalize its first part from the parity quantifier to a mod-$p$ quantifier.

\subsection{A Barrier at One Half}\label{sec:half-barrier}

Both sieves of the previous sections stop just short of a success probability of $1/2$, as do all the constructions surveyed at the beginning of this section. We close the section by showing that this is neither a coincidence nor an artifact of any particular construction: no sieve of this kind can do better, and $1/2$ is attained, up to a negligible additive term, by a distribution that is easy to describe.

The setting is the natural one. A sieve has to be fixed before the formula is seen, so what we are given is a distribution over hash functions, and it has to work for every satisfiable formula. This is exactly the form of the guarantees above: Theorem~\ref{thm:parity-isolation}, for instance, holds for every satisfiable $f$.

\begin{definition}\label{def:sieve-dist}
Let $U=\{0,1\}^n$ and $N=2^n$. A \emph{sieve} is a probability distribution $\mathcal{H}$ over the $2^N$ Boolean functions $h:U\to\{0,1\}$. For a formula $f$ with solution set $S_f\subseteq U$ we write
\[
P_{\mathcal{H}}(f) \;=\; \Pr_{h\sim\mathcal{H}}\bigl[\,\parity\x\,[f(\x)\land h(\x)]\,\bigr],
\]
that is, the probability that $f\land h$ has an odd number of solutions.
\end{definition}

A sieve in this sense is far more general than anything we could implement: it is an arbitrary distribution over \emph{all} Boolean functions on $n$ variables, with no restriction on how $h$ is described, how large it is, or how many random bits are used to draw it. A barrier proved in this setting therefore applies to every concrete construction.

Two restrictions remain, both standard~\cite{DellKMW13}. First, the sieved formula is $f\land h$: it lives on the variables of $f$ and its solutions are solutions of $f$, which is the \emph{pruning} property. Second, $h$ is drawn without looking at $f$, which is the \emph{black-box} setting. The second is mild here: inside Toda's reduction one $h$ must serve $f(\x,\y)$ for all $2^{|\x|}$ assignments to the outer variables at once, so information about $f$ would have to hold for all of the resulting solution sets simultaneously. The first is not mild: modular addition, for instance, does not prune.

The argument rests on the following identity. For a hash function $h$ and a nonempty $S\subseteq U$, let $\sigma_h(S)=(-1)^{|\{\x\in S\,:\,h(\x)=1\}|}$, which equals $-1$ exactly when $h$ succeeds on $S$, and $+1$ when it fails.

\begin{lemma}\label{lem:subset-sum}
Let $h$ be a hash function that is not identically zero. Then $\sum_{\emptyset\neq S\subseteq U}\sigma_h(S) = -1$.
\end{lemma}

\begin{proof}
For $\x\in U$ put $a_{\x}=(-1)^{h(\x)}\in\{+1,-1\}$, so that $\sigma_h(S)=\prod_{\x\in S}a_{\x}$. Expanding the product $\prod_{\x\in U}(1+a_{\x})$, by choosing from each factor either the $1$ or the $a_{\x}$, produces exactly one term $\prod_{\x\in S}a_{\x}$ for every subset $S\subseteq U$, with the empty set contributing $1$. Hence
\[
\prod_{\x\in U}(1+a_{\x}) \;=\; 1+\sum_{\emptyset\neq S\subseteq U}\sigma_h(S).
\]
Since $h$ is not identically zero, there is some $\y\in U$ with $h(\y)=1$, so that $a_{\y}=-1$ and the factor $(1+a_{\y})$ vanishes, making the entire product $0$. The claim follows.
\end{proof}

\begin{theorem}\label{thm:half-barrier}
For every sieve $\mathcal{H}$ there is a satisfiable formula $f$ with
\[
P_{\mathcal{H}}(f)\;\leq\;\frac{1}{2}+\frac{1}{2(2^N-1)}.
\]
\end{theorem}

\begin{proof}
The all-zero hash function always yields a count of $0$ and therefore always fails, so removing it from the support of $\mathcal{H}$ and renormalizing can only increase $P_{\mathcal{H}}(f)$, for every $f$. We may thus assume that $\mathcal{H}$ is supported on hash functions that are not identically zero.

Taking expectations in Lemma~\ref{lem:subset-sum} gives $\sum_{\emptyset\neq S\subseteq U}\mathbb{E}_{h}[\sigma_h(S)] = -1$. This is a sum of $2^N-1$ terms, so at least one of them, say the one indexed by $S^*$, satisfies $\mathbb{E}_{h}[\sigma_h(S^*)]\geq -1/(2^N-1)$, since otherwise the sum would be strictly smaller than $-1$.

Finally, as $\sigma_h(S)$ takes the value $-1$ on success and $+1$ on failure, we have $\mathbb{E}_h[\sigma_h(S)] = 1-2P_{\mathcal{H}}(S)$. Hence for any $f$ whose solution set is $S^*$,
\[
P_{\mathcal{H}}(f)\;=\;\frac{1-\mathbb{E}_h[\sigma_h(S^*)]}{2}\;\leq\;\frac{1}{2}+\frac{1}{2(2^N-1)}. \qedhere
\]
\end{proof}

The bound cannot be improved, as the following shows.

\begin{proposition}\label{prop:half-optimal}
Let $\mathcal{H}^*$ be the uniform distribution over the $2^N-1$ hash functions that are not identically zero. Then $P_{\mathcal{H}^*}(f) = \frac{1}{2}+\frac{1}{2(2^N-1)}$ for \emph{every} satisfiable $f$.
\end{proposition}

\begin{proof}
Let $S_f\neq\emptyset$ and fix some $\y\in S_f$. Pair each hash function $h$ with the function $h'$ that agrees with $h$ everywhere except at $\y$. Flipping a single value twice returns the original function, so this is a perfect matching of the $2^N$ hash functions into $2^{N-1}$ pairs. The counts of two paired functions on $S_f$ differ by exactly one, because they differ on $\y\in S_f$ and nowhere else, so exactly one member of each pair succeeds on $S_f$. Hence exactly $2^{N-1}$ of the $2^N$ hash functions succeed. The all-zero function is not one of them, so all $2^{N-1}$ of them survive its removal, and $P_{\mathcal{H}^*}(f) = 2^{N-1}/(2^N-1)$, which is the stated value.
\end{proof}

Since $N=2^n$, the additive term $1/(2(2^N-1))$ is doubly exponentially small in $n$. Theorem~\ref{thm:half-barrier} therefore says that the success probability of any sieve is at most $1/2$ for every practical purpose, and Proposition~\ref{prop:half-optimal} says that $1/2$ is essentially reached. This places the constructions of this section in context. The multilinear hash function attains $7/16$ over a field of $\lceil\log_2 n\rceil+3$ bits, and $\frac12-2^{-(c+1)}$ over a field of $\lceil\log_2 n\rceil+c$ bits, so it approaches the barrier at the modest price of a logarithmically growing field, and no further work on any such construction can push it past $1/2$. The barrier also explains, from the outside, the ceiling that~\cite{NRS95} observe from within their coding framework. It also complements~\cite{DellKMW13} in two ways. Their unconditional black-box bound of $O(1/n)$ concerns isolating a \emph{unique} solution. For sieves that only remove solutions and are judged by parity, their bound is the conditional one cited at the opening of this section, and their concluding discussion records the range between $1/2$ and $2/3$ as open for such reductions. Theorem~\ref{thm:half-barrier} settles this range within the black-box setting: no distribution over hash functions drawn independently of the formula exceeds $1/2$ by more than a doubly exponentially small amount, unconditionally and with no restriction on the complexity of the distribution. For efficient procedures that may inspect the formula, the range remains open, and their conditional $2/3$ remains the best known bound.

\begin{remark}
Identifying a hash function with its indicator vector in $\{0,1\}^N$, and a solution set with a linear test applied to that vector, the quantity $\mathbb{E}_h[\sigma_h(S)]$ is the bias of $\mathcal{H}$ with respect to the test $S$, and Lemma~\ref{lem:subset-sum} is the standard fact that the Fourier coefficients of a distribution over $\mathbb{F}_2^{N}$ sum to $2^{N}$ times its mass at the origin. The technique is thus not new. What we take from it is the consequence for parity sieving, which we have not found recorded elsewhere.
\end{remark}

Two further comments. First, $\mathcal{H}^*$ is of no practical use, since sampling from it costs $2^n$ random bits and the resulting $h$ has no small description. The content of Proposition~\ref{prop:half-optimal} is only that the upper bound of Theorem~\ref{thm:half-barrier} cannot be improved. Second, the barrier is specific to parity. A sieve that is asked only to make the number of solutions non-zero modulo some $q>2$ is not covered by it, and its ceiling is the larger value $\frac{q-1}{q}$. Theorem~\ref{thm:half-barrier} is therefore one more reason to look at moduli other than $2$, which is the subject of Section~\ref{sec:modp}. There, Theorem~\ref{thm:modq-barrier} shows that for every $q\geq 3$ the larger ceiling is in fact never attained, and Proposition~\ref{prop:mod2-optimal} that the extra room it does leave is not worth its cost.

\section{Analysis of the Union Bound Factor}
\label{sec:union-bound}

Having systematically examined the components of Toda's reduction, we now turn to the last factor that remains, the union bound over outer quantifier variables. This factor has received little attention in the literature, but understanding it is essential to any effort to make Toda's reduction efficient. We first explain the role and the dominance of the union bound factor in Toda's reduction. We then identify the parameter that actually governs it, show that with respect to that parameter the bound is tight and unavoidable, and finally ask how that parameter may be bounded for a given formula.

Toda's reduction eliminates quantifiers one at a time, from innermost to outermost, replacing the innermost quantifier by a parity quantifier by means of a probabilistic sieve. Consider one such step, applied to a formula of the form
\[
\vec{Q}\x\, \exists \y \; G(\x, \y).
\]
Here $\vec{Q}\x$ is the outer prefix and $\x$ the outer variables, in the terminology of Section~\ref{sec:explToda}. The sieve is applied to the inner quantifier $\exists\y$, while the $s=|\x|$ outer variables remain free. This is where the difficulty lies. These variables are not ones that the step may set as it finds convenient: they are still governed by their own quantifiers, which are eliminated only at later steps. The formula that this step produces must therefore be correct simultaneously for every assignment $\vsigma\in\{0,1\}^s$ to them.

The sieve, however, is randomized, and its randomness is drawn once and shared by all of the outer assignments. If a single assignment fails with probability at most $\delta$, then the probability that some assignment fails is at most $2^s\delta$. As a result, to reach an inner-error of $\epsilon_i$ at this step we need $2^s\delta \leq \epsilon_i$, or equivalently $\delta \leq \epsilon_i/2^s$, and this is exactly the factor $2^s$ that appears in Proposition~\ref{clm:repetitionUNbound}.


It is worth seeing how dominant this factor is. Proposition~\ref{clm:repetitionUNbound} already separates the number of repetitions into an error term and a union bound term, the first growing only logarithmically in $1/\epsilon_i$ and the second linearly in $s$. What makes the second decisive is that once the sieve provides a \emph{constant} success probability, as the multilinear hash function of Section~\ref{subsec:multilinear-hash} does with $p=7/16$, its coefficient $\log_{1-p}(1/2)$ is a constant as well, so that for the hundreds or thousands of outer variables that are common in practical QBF instances the union bound term dominates the count completely. At $p=1/2$ and $\epsilon_i = 0.1$, for example, a step with $s = 50$ outer variables requires $k_i = 54$ repetitions, of which $50$ come from the union bound and only $3.32$ from the error budget.

Note also that although $s$ decreases as the reduction proceeds outwards, it is largest at the innermost steps, and the blow-up that it causes there is then multiplied by the repetitions of every subsequent step.

The rest of this section examines this factor. Section~\ref{sec:equivclasses} shows that $2^s$ is in fact the wrong parameter, and replaces it by the number $C$ of distinct functions that the outer assignments induce over the inner variables. Section~\ref{sec:tightUB} shows that $C$ is the right one, in the sense that no argument from pairwise information improves on a union bound over the $C$ classes. Section~\ref{subsec:high-dependency} then asks how $C$ may be bounded for a given formula.


\subsection{From Assignments to Equivalence Classes}\label{sec:equivclasses}

The factor $2^s$ in Proposition~\ref{clm:repetitionUNbound} arises from taking a union bound over all $2^s$ assignments to the outer variables. This is more than is needed. What the sieve actually sees at a given step is not the assignment $\vsigma$ itself, but the function over the inner variables that $\vsigma$ leaves behind. If two assignments leave behind the same function, no hash function can tell them apart, and there is no reason to charge for both.

We would therefore like to charge one term not per assignment, but per distinct \emph{behavior} of the sieve. To do so we must know exactly when two assignments are indistinguishable, since grouping too coarsely would make the bound unsound, and grouping too finely would give away the saving. The theorem below settles this, and shows that the induced inner function is precisely the right invariant.

We first fix the notation. For an outer assignment $\vsigma\in\{0,1\}^s$ write $G_{\vsigma}(\y) = G(\vsigma,\y)$ for the \emph{inner function} that $\vsigma$ induces, and for a hash function $h$ let $X_{\vsigma}(h)=1$ if the sieve succeeds on $G_{\vsigma}$, that is if $\parity\y\,[G_{\vsigma}(\y)\land h(\y)]$ is \true\kern-0.35em, and $X_{\vsigma}(h)=0$ otherwise. Recall from Section~\ref{sec:half-barrier} that a sieve is a distribution over hash functions, and that the guarantee it provides, as in Theorem~\ref{thm:parity-isolation}, is of the following form: for some $p>0$ and \emph{every} nonempty $T\subseteq\{0,1\}^{|\y|}$,
\begin{equation}\label{eq:sieve-guarantee}
\Pr_h\bigl[\,|T\cap h^{-1}(1)| \text{ is odd}\,\bigr] \;\geq\; p .
\end{equation}

\begin{theorem}\label{thm:equivclass}
Let the sieve satisfy Equation~\ref{eq:sieve-guarantee}. Then for any two outer assignments $\vsigma,\vsigma'$,
\[
X_{\vsigma}(h) = X_{\vsigma'}(h) \text{ for every } h \qquad\Longleftrightarrow\qquad G_{\vsigma}\equiv G_{\vsigma'} .
\]
\end{theorem}

\begin{proof}
If $G_{\vsigma}\equiv G_{\vsigma'}$ then the two functions have the same solution set, so for every $h$ the surviving sets coincide, and so do their parities.

For the converse, suppose $G_{\vsigma}\not\equiv G_{\vsigma'}$, and let
\[
T \;=\; \{\,\y \;:\; G_{\vsigma}(\y)\neq G_{\vsigma'}(\y)\,\}
\]
be the set of inner assignments on which the two functions differ, which is nonempty by assumption. Writing $X_{\vsigma}(h)=\sum_{\y\in h^{-1}(1)}G_{\vsigma}(\y) \bmod 2$, and similarly for $\vsigma'$, we obtain
\[
X_{\vsigma}(h)\oplus X_{\vsigma'}(h) \;=\; \sum_{\y\in h^{-1}(1)}\bigl(G_{\vsigma}(\y)+G_{\vsigma'}(\y)\bigr) \bmod 2 \;=\; |T\cap h^{-1}(1)| \bmod 2 ,
\]
because a $\y$ on which the two functions agree contributes either $0+0$ or $1+1$, and hence $0$, whereas a $\y$ on which they differ contributes $1$. Since $T$ is nonempty, applying Equation~\ref{eq:sieve-guarantee} to $T$ gives $\Pr_h[X_{\vsigma}(h)\neq X_{\vsigma'}(h)]\geq p>0$, so that some $h$ separates the two assignments.
\end{proof}

Theorem~\ref{thm:equivclass} thus gives us exactly the grouping we were after: assignments that induce the same inner function fail together on every hash function, and assignments that induce different inner functions are separated by some hash function, so that neither a coarser nor a finer grouping would be correct. Accordingly we let
\[
C \;=\; \bigl|\{\,G_{\vsigma} \;:\; \vsigma\in\{0,1\}^s\,\}\bigr|
\]
be the number of distinct inner functions, and call two outer assignments \emph{equivalent} when they induce the same one. Since all the assignments in a class share a single failure event, the union bound in the proof of Proposition~\ref{clm:repetitionUNbound} need only range over the $C$ classes, which gives
\[
k_i \;\geq\; \log_{1-p}(\epsilon_i/C) \;=\; \log_{1-p}(\epsilon_i) \;+\; \log_2(C)\cdot\log_{1-p}(1/2),
\]
replacing the linear term $s$ by $\log_2 C \leq s$. Two questions follow, and occupy the rest of this section: whether $C$ is really the end of the story, and how $C$ can be bounded for a given formula.

\subsection{Tightness of the Union Bound}\label{sec:tightUB}

Unfortunately the union bound factor of Proposition~\ref{clm:repetitionUNbound} is not only a theoretical estimate, but in some cases reflects the actual number of required repetitions. This is a result of the well known fact that the union bound is tight for independent events~\cite{chan2021introduction}.

It is also not hard to exhibit a family of QBF formulas that exhibit this tight union bound. Consider for example a QBF of the form $\forall \x\, \exists \y\, G(\x, \y)$ where $|\x| = s$ and $|\y| = m$. Then set

\[
G(\x, \y) \;=\; \bigwedge_{\vsigma \in \{0,1\}^s} \Bigl( (\x = \vsigma) \rightarrow g_{\vsigma}(\y) \Bigr)
\]

where for $\vsigma \in \{0,1\}^s$, $g_{\vsigma} : \{0,1\}^m \to \{0,1\}$ is a Boolean formula description of an independently and uniformly sampled random Boolean function over $\{0,1\}^m$.

Since each $g_{\vsigma}$ is sampled independently, so are the inner functions, and the success or failure of the sieve on one outer assignment carries no information about its success or failure on any other. Moreover the $2^s$ inner functions are distinct with overwhelming probability, so that $C=2^s$ and no collapse of the kind discussed in Section~\ref{sec:equivclasses} occurs either. We are therefore in the setting of $C$ independent events, in which the union bound is essentially exact.

\begin{proposition}\label{prop:union-gap}
Let $n$ independent events each fail with probability $q$, where $q$ is chosen so that all of them succeed simultaneously with probability exactly $1-\epsilon$. Then the union bound estimate $1-nq$ of that probability satisfies
\[
0 \;\leq\; (1-\epsilon)-(1-nq) \;=\; nq-\epsilon \;\leq\; -\ln(1-\epsilon)-\epsilon \;\leq\; \epsilon^{2},
\]
and the middle bound is approached as $n\to\infty$.
\end{proposition}

The tightness of the union bound is a well known fact, and we state it here, with a proof in Appendix~\ref{app:union-gap}, only for completeness and to get a sense of scale for our case. The union bound therefore understates the true success probability by at most $-\ln(1-\epsilon)-\epsilon$, a quantity that never exceeds $\epsilon^2$ in the range $[0,0.5]$. As a result, at $\epsilon=0.1$ the union bound has a slack of no more than $0.0054$, and that slack only shrinks as $\epsilon$ gets smaller. On an instance of the above family the $2^s$ factor of Proposition~\ref{clm:repetitionUNbound} is thus not an artifact of a loose analysis, but very nearly the truth, and striving for a more accurate computation is of no help.

One might still hope that on other formulas the correlations between distinct classes could be exploited by a sharper tool than the union bound. We close this section by arguing that they cannot, at least not by any method that uses pairwise information only. The reason is that distinct classes cannot be very correlated to begin with.

For two assignments in distinct classes, let $\rho = \Pr_h[X_{\vsigma}(h)=X_{\vsigma'}(h)]$ be the probability that the sieve treats them alike.

\begin{proposition}\label{prop:rho}
For any sieve satisfying Equation~\ref{eq:sieve-guarantee} and any two assignments in distinct classes, $\rho\leq 1-p$.
\end{proposition}

\begin{proof}
As in the proof of Theorem~\ref{thm:equivclass}, $X_{\vsigma}(h)\oplus X_{\vsigma'}(h)$ is the parity of $|T\cap h^{-1}(1)|$ for the nonempty set $T$ on which $G_{\vsigma}$ and $G_{\vsigma'}$ differ. The two assignments are treated alike exactly when that parity is even, and Equation~\ref{eq:sieve-guarantee} bounds the probability of an odd parity from below by $p$. Hence $\rho = 1 - \Pr[\text{odd}]\leq 1-p$.
\end{proof}

Let $A_1,\ldots,A_C$ be the failure events of our $C$ classes, where $A_i$ is the event that the sieve fails on class $i$ in all $k$ repetitions of the step. Only one step is in question here, so we write $k$ for its repetition count. Each of these has the same marginal,
\[
q \;=\; \Pr[A_i] \;=\; (1-p)^{k} .
\]
For the pairwise terms, take two classes. Since the marginals of $X_{\vsigma}$ and $X_{\vsigma'}$ are both $p$, and the two agree with probability $\rho$, both fail on a single hash function with probability $(1-2p+\rho)/2$. Over $k$ independent repetitions we therefore get, writing $v$ for the resulting common value and using $\rho\leq1-p$ from Proposition~\ref{prop:rho},
\[
v \;=\; \Pr[A_i\cap A_j] \;=\; \left(\frac{1-2p+\rho}{2}\right)^{k} \;\leq\; \left(1-\frac{3p}{2}\right)^{k} \;=:\; r .
\]
These two numbers, $q$ and $v$, are the whole of the pairwise information available about our classes, and the following proposition says exactly what they permit.

\begin{proposition}\label{prop:pairwise-optimal}
Let $C\geq2$, and let $q$ and $v$ satisfy $0\leq v\leq q\leq1$ and $Cq-(C-1)v\leq1$. Among all collections of events $B_1,\ldots,B_C$ for which
\[
\Pr[B_i]=q \ \text{ for every } i, \qquad\text{and}\qquad \Pr[B_i\cap B_j]=v \ \text{ for every } i\neq j ,
\]
the largest value of $\Pr\bigl[\bigcup_{i=1}^{C}B_i\bigr]$ is $Cq-(C-1)v$.
\end{proposition}

\begin{proof}
For the upper bound, split the union into disjoint pieces and drop from each piece all of its conditions but the first:
\[
\Pr\Bigl[\bigcup_{i=1}^{C}B_i\Bigr] \;=\; \Pr[B_1] \;+\; \sum_{i=2}^{C}\Pr\bigl[B_i\cap\overline{B_1}\cap\cdots\cap\overline{B_{i-1}}\bigr] \;\leq\; \Pr[B_1] \;+\; \sum_{i=2}^{C}\Pr\bigl[B_i\cap\overline{B_1}\bigr] .
\]
Every term of the sum equals $\Pr[B_i]-\Pr[B_i\cap B_1]=q-v$, so the right-hand side is $q+(C-1)(q-v)=Cq-(C-1)v$.

That this value is attained, and not merely approached, is what makes it the best possible. Take a probability space with $C+2$ outcomes: an outcome $\omega_{\ast}$ of probability $v$, outcomes $\omega_1,\ldots,\omega_C$ of probability $q-v$ each, and an outcome $\omega_{\varnothing}$ carrying whatever probability is left, namely $1-v-C(q-v)$. The three hypotheses are exactly what makes this a probability space, since $v\geq0$, since $q-v\geq0$, and since $1-v-C(q-v)=1-\bigl(Cq-(C-1)v\bigr)\geq0$. Put $B_i=\{\omega_{\ast},\omega_i\}$. Then $\Pr[B_i]=v+(q-v)=q$ for every $i$, and for $i\neq j$ the outcomes $\omega_i$ and $\omega_j$ are distinct, so $\omega_{\ast}$ is the only outcome that $B_i$ and $B_j$ share and $\Pr[B_i\cap B_j]=v$. Finally $\bigcup_i B_i$ consists of every outcome but $\omega_{\varnothing}$, so its probability is $v+C(q-v)=Cq-(C-1)v$.
\end{proof}

The inequality above is that of Kounias~\cite{Kounias1968} and Hunter~\cite{Hunter1976} with all the pairwise probabilities equal, where it becomes not merely an upper bound but the smallest bound that $q$ and $v$ admit. The most that pairwise information can subtract from the union bound is therefore
\[
\underbrace{Cq}_{\text{union bound}} \;-\; \underbrace{(C-1)\,v}_{\text{correction}} .
\]

We can now ask how many repetitions this correction actually saves. Write $\gamma = v/q$ for the size of the correction relative to a single failure probability, and use $C-1\approx C$. The union bound asks for $Cq\leq\epsilon_i$, whereas the refined bound asks for $Cq(1-\gamma)\leq\epsilon_i$. Substituting $q=(1-p)^{k}$ and taking logarithms in the second of these,
\[
k\,\ln(1-p) \;\leq\; \ln\frac{\epsilon_i}{C} - \ln(1-\gamma) ,
\]
and dividing by $\ln(1-p)$, which is negative and so reverses the inequality,
\[
k \;\geq\; \underbrace{\frac{\ln(\epsilon_i/C)}{\ln(1-p)}}_{\text{what the union bound asks}} \;-\; \frac{\ln(1-\gamma)}{\ln(1-p)} .
\]
Note that the second term is what the correction saves. That is,
\[
\Delta k \;=\; \frac{\ln(1-\gamma)}{\ln(1-p)} .
\]
Both logarithms are negative, so $\Delta k$ is positive, as it should be. Moreover $\Delta k\geq1$ precisely when $\ln(1-\gamma)\leq\ln(1-p)$, that is when $1-\gamma\leq1-p$, that is when
\[
\gamma \;\geq\; p .
\]
So the correction is worth as much as a single repetition exactly when $\gamma$ is at least $p$.

That is a demanding requirement, because $\gamma$ decays geometrically in $k$. Indeed, $\gamma = v/q \leq r/q = \left(\frac{1-\frac{3p}{2}}{1-p}\right)^{k}$, so $\gamma\geq p$ forces $\left(\frac{1-\frac{3p}{2}}{1-p}\right)^{k} \geq p$, or equivalently, taking logarithms and dividing by the negative quantity $\ln\left(\frac{2-3p}{2-2p}\right)$, ${k} \leq \frac{\ln p}{\ln \left(\frac{2-3p}{2-2p}\right)}$.

At $p=7/16$, the value provided by the multilinear hash function of Section~\ref{subsec:multilinear-hash}, we get ${k} < 1.7$, and since $k$ is a positive integer this leaves only $k = 1$. But $k=1$ can never occur under $p \leq 1/2$ and $\epsilon_i \leq \epsilon< 1/2$, as a single repetition with probability $p \leq 1/2$ of success cannot guarantee an overall success of at least $1-\epsilon_i > 1/2$. All in all, no bound that uses pairwise information saves even a single repetition over the plain union bound.

Lowering $p$ does raise the threshold, since $\frac{\ln p}{\ln \left(\frac{2-3p}{2-2p}\right)}\rightarrow\infty$ as $p\rightarrow0^+$, but it raises the number of repetitions faster. The number of repetitions we need is $k \geq \log_{1-p}(\epsilon_i/C) = \ln(\epsilon_i/C) \;/\; \ln(1-p)$, and combining it with the bound from before we get $\frac{\ln(\epsilon_i/C)}{\ln(1-p)} \leq k \leq \frac{\ln p}{\ln \left(\frac{2-3p}{2-2p}\right)}$. Multiplying through by $-\ln(1-p)>0$, this means that a saving is possible only when
\[
{\ln(C/\epsilon_i)} \;\leq\; -\frac{\ln (p)\, {\ln(1-p)}}{\ln \left(\frac{2-3p}{2-2p}\right)} ,
\]
whose right-hand side is decreasing in $p$ and stays below $4.25$ for any reasonable value $p\in[0.1,0.5]$. Since moreover $\epsilon_i \leq \epsilon < 1/2$, this implies $C < e^{4.25}\epsilon_i < e^{4.25}/2 < 36$. For most practical formulas this is far from being the case: it would require the step to induce fewer than $36$ distinct inner functions, whereas $C$ is typically exponential in the number of outer variables.

\subsection{Bounding the Number of Classes}
\label{subsec:high-dependency}

Theorem~\ref{thm:equivclass} makes $C$, rather than $2^s$, the parameter that governs the number of repetitions, and Section~\ref{sec:tightUB} shows that pairwise information yields nothing further once $C$ is fixed. What remains is to bound $C$ for a given formula, and here there is a great deal to gain, as theoretically $C$ can be exponentially smaller than $2^s$.

To see what is at stake, consider first the extreme case, in which all the outer assignments induce the same inner function, that is $G(\vsigma,\y) \equiv G(\vsigma',\y)$ for all $\vsigma,\vsigma'$. Then $C=1$ and the union bound term disappears altogether. This happens when the outer variables appear syntactically in the formula but do not influence the inner ones. Such a formula is of course a rare and degenerate case, and we mention it only to convey the idea: what the union bound charges for is the diversity of the inner functions, and when there is no diversity there is nothing to pay. We also stress that Theorem~\ref{thm:equivclass} speaks of the induced Boolean functions and not of the syntax, so that two assignments may leave behind formulas that look very different and still belong to the same class.

A simple syntactic criterion already gives a useful bound. Call an outer variable a \emph{boundary variable} if it occurs in at least one clause together with an inner variable, and let $B$ be the set of boundary variables. An outer variable outside $B$ occurs only in clauses over outer variables, and therefore imposes no constraint on the inner variables: changing it either leaves $G_{\vsigma}$ untouched, or else falsifies a purely outer clause and sends $\vsigma$ to the single class in which $G_{\vsigma}$ is unsatisfiable. The inner function is thus determined by the restriction of $\vsigma$ to $B$, which gives the following.

\begin{proposition}\label{thm:boundaryC}
$C \leq 2^{|B|}$, where $B$ is the set of boundary variables.
\end{proposition}

Using $|B|$ in place of $s$ is therefore always sound, and is a strict gain whenever some outer variable meets no inner variable in any clause. Returning to the step of Section~\ref{sec:union-bound} with $s=50$, $p=1/2$ and $\epsilon_i=0.1$: if only $12$ of the outer variables are boundary variables then $C\leq 2^{12}$, and the count falls from $54$ repetitions to $16$. The bound is only syntactic, and can be far from tight, since two assignments that activate different sets of clauses may still induce the same Boolean function, and deciding equivalence exactly is of course expensive. Sharper and still cheap criteria for bounding $C$ seem to us a promising direction, and we leave them to future work.







The results of this section can be summarized as follows. The parameter that governs it is $C$, the number of distinct inner functions, and not $2^s$. Once $C$ is known, no tool that draws on pairwise information improves on the union bound over the $C$ classes. Moreover, $C$ itself can be anywhere between $1$ and $2^s$, so that the whole question of whether a Toda-based solver is practical on a given family of formulas turns on how many distinct inner functions its outer assignments induce. Such a collapse, which may occur at every step of the reduction, can yield a substantial improvement in the number of repetitions.




\section{Toda's Reduction Modulo \texorpdfstring{$q$}{q}}\label{sec:modp}

The reduction we have analyzed so far ends with a parity quantifier, that is, with a question about the number of solutions modulo $2$. It is natural to ask what happens when $2$ is replaced by another modulus. This section answers that question, and the answer is negative: no modulus improves on $2$, and the reason is a simple size computation that we give in Section~\ref{sec:modq-vs-mod2}. Throughout this section $q$ denotes the modulus, and we write $\pi$ for the success probability of a sieve, reserving the letters used for probabilities in Sections~\ref{sec:innerComp} and~\ref{sec:union-bound} for their meaning there.

Before anything else we state clearly what is and what is not ours. The theorem that the polynomial hierarchy is randomly reducible to $\mathrm{Mod}_k\mathrm{P}$ for \emph{every} $k\geq 2$, and not only for $k=2$, is due to Toda and Ogiwara~\cite{TodaOgiwara1992} who prove it, and independently to Tarui~\cite{Tarui1993}. Both proofs predate this work by three decades. A third route was communicated to Toda and Ogiwara by Beigel and is recorded as a remark in their paper. It composes the structural results of Hertrampf~\cite{Hertrampf1990} and Beigel and Gill~\cite{BeigelGill1992} with Toda's modulo amplification, and we return to it in Section~\ref{sec:modp-composite}. The algebra that we use below is likewise not new: property by property it is the content of~\cite[Section 6]{BeigelGill1992}, stated there for counting classes rather than for formulas. We indicate the correspondence in Appendix~\ref{app:modp}.

What we do here is narrower, and is of the same nature as the rest of this paper. Both~\cite{TodaOgiwara1992} and~\cite{Tarui1993} work with counting classes ($\mathrm{Gap}\P$ functions in the first case, probabilistic polynomials over $\mathbb{Z}$ in the second), and at that level of abstraction there is no Boolean formula whose size one could measure. Since our concern throughout has been the size of the formula that is eventually handed to a model counter, we write the construction in the language of Boolean formulas and QBF, in the same operators $+$ and $\times$ and the same quantifier notation used in Section~\ref{sec:todaBrief}. This makes the price of the generalization explicit: exactly where the extra copies of the formula appear, and how many. It is that price, and not the theorem, that is the subject of this section.

\subsection{The modulo \texorpdfstring{$q$}{q} quantifier}\label{sec:modp-quantifier}

Let $q$ be a prime and let $F(\x)$ be a Boolean formula. Generalizing the parity quantifier of Section~\ref{sec:prelim}, we define
\[
\modq\x\, F(\x) =
\begin{cases}
\mathit{True} & \text{if } \#F(\x) \not\equiv 0 \pmod{q},\\
\mathit{False} & \text{if } \#F(\x) \equiv 0 \pmod{q},
\end{cases}
\]
so that $\parity$ is exactly $\modq$ for $q=2$. As with $\parity$, the reduction manipulates these quantifiers by a small set of algebraic rules (negation, nesting, conjunction, disjunction, and insertion of free variables), which generalize the five properties of Claim~\ref{clm:many}. One further rule is needed for $q>2$ and has no analogue at $q=2$, and we state it here because it is the source of the entire cost of the generalization.

\begin{claim}[Fermat]\label{clm:modp-fermat}
Let $q$ be prime and let $F(\x)$ be a Boolean formula. Then
\[
\modq\x\, F(\x) \;=\; \modq\x\, F(\x)^{q-1},
\]
where $F^{a}$ denotes the $a$-fold product $F(\x^1)\times\cdots\times F(\x^a)$ over disjoint copies of $\x$, and moreover $\#[F^{q-1}]\equiv 0$ or $1 \pmod q$ for \emph{every} formula $F$.
\end{claim}

The point of Claim~\ref{clm:modp-fermat} is the second half. Modulo $2$ every count is congruent to $0$ or $1$, so the rules for combining quantifiers may be applied without further ado. Modulo $q>2$ a count may be congruent to any of $1,\ldots,q-1$, and the rules for nesting, conjunction and disjunction all require knowing \emph{which} residue occurs. Raising to the power $q-1$ forces the count into $\{0,1\}$ unconditionally, by Fermat's little theorem, and thereby restores the situation that holds for free when $q=2$.

\subsection{The construction}\label{sec:modp-construction}

With Claim~\ref{clm:modp-fermat} in hand the reduction proceeds exactly as in Section~\ref{sec:explToda}. Only the \textbf{Combine} step of Algorithm~\ref{alg:Amplify} changes, and the $+1$ steps of Algorithm~\ref{alg:Toda1} that propagate a negation become the addition of another constant, by Property~\ref{modp:neg} of Appendix~\ref{app:modp}. Recall that at step $i$ we sieve the ongoing formula $G_i$ with $k_i$ independently chosen hash functions $h_1,\ldots,h_{k_i}$ and must express the statement ``at least one of the $k_i$ repetitions succeeded'' as a single $\modq$-quantified formula. Writing $F_j = G_i \land h_j$ for the $j$-th sieved copy, the combination is
\begin{equation}\label{eq:modp-combine}
\mathbf{Combine}_q(F_1,\ldots,F_{k}) \;=\;
\Bigl[\bigl(F_1^{\,q-1} + (q-1)\bigr) \times \cdots \times \bigl(F_{k}^{\,q-1} + (q-1)\bigr)\Bigr] + K_k,
\end{equation}
where $+\,c$ denotes the $+$ operator of Section~\ref{sec:prelim} applied to a formula with exactly $c$ satisfying assignments, and
\[
K_k \;=\; q - \bigl[(q-1)^{k}\bmod q\bigr].
\]
The constant $K_k$ is computed from $q$ and $k$ at construction time and is then a single small number, and nothing of the form $(q-1)^{k}$ ever appears in the formula. Since $q-1\equiv -1 \pmod q$ we have $(q-1)^{k}\equiv(-1)^{k}$, so in fact
\[
K_k = \begin{cases} q-1 & \text{if } k \text{ is even},\\ 1 & \text{if } k \text{ is odd},\end{cases}
\]
and in particular $K_k\not\equiv 0 \pmod q$ always.

The mechanism is easy to read off. By Claim~\ref{clm:modp-fermat} each $\#F_j^{\,q-1}$ is $0$ or $1$ modulo $q$, so each factor of the product is congruent to $q-1$ when the $j$-th repetition failed and to $1+(q-1)=q\equiv 0$ when it succeeded. The product is therefore congruent to $(q-1)^{k}$ exactly when every repetition failed, and to $0$ as soon as one succeeded. Adding $K_k$ inverts this: $(q-1)^{k}+K_k\equiv 0$ in the first case, and $0+K_k = K_k \not\equiv 0$ in the second. Substituting $q=2$ into Equation~\ref{eq:modp-combine} gives $q-1=1$ and $K_k=1$, so it returns $[(F_1+1)\times\cdots\times(F_k+1)]+1$, which is precisely the \textbf{Combine} of Algorithm~\ref{alg:Amplify}.

\begin{remark}\label{rem:modp-todaogiwara}
Equation~\ref{eq:modp-combine} and the construction of Toda and Ogiwara share the same core. Writing $A_j$ for the number of solutions surviving the $j$-th sieve, both are built on the product $\prod_j (A_j-1)$, which is $0$ as soon as one sieve succeeds and $\pm1$ when none does, the sign depending on the parity of the number of factors. Modulo $q$ our factor $F_j^{\,q-1}+(q-1)$ is congruent to $F_j^{\,q-1}-1$, so the two products agree.

What is done with the product differs, because the two constructions want opposite outputs. Toda and Ogiwara are building an indicator of ``no solution'', which must take the value $1$ in that case, so they square the product and discard the sign~\cite[p.~324]{TodaOgiwara1992}. We need the reverse: a count $\equiv 0$ when no repetition fires and $\not\equiv 0$ when one does. Adding $K_k$ does both jobs at once, sending $\pm1$ to $0$ and $0$ to $K_k$. Squaring would give us the wrong polarity.
\end{remark}

\begin{theorem}\label{thm:modp}
Let $q$ be prime. For every QBF formula $F$ with $d$ alternations and every $0<\epsilon<1/2$, the reduction of Section~\ref{sec:explToda}, with \textbf{Combine} replaced by Equation~\ref{eq:modp-combine}, produces a quantifier-free formula $F_0$ such that, with probability at least $1-\epsilon$, $F$ is \true if and only if $\#F_0 \not\equiv 0 \pmod{q}$.
\end{theorem}

The full proof, together with the generalizations of the five properties of Claim~\ref{clm:many} and the step-by-step transformation that corresponds to Appendix~\ref{sec:detreduction}, is given in Appendix~\ref{app:modp}. We note there that Equation~\ref{eq:modp-combine} is not an ad-hoc construction: it is precisely what one obtains by applying the negation and conjunction properties to the disjunction of the $k$ repetitions, in the same way that the parity case obtains its \textbf{Combine} from Properties~\ref{prop:parity:negation} and~\ref{prop:parity:disjunction} of Claim~\ref{clm:many}.

\subsection{The size of the construction}\label{sec:modp-size}

We now account for the size, which is the reason for this section. One ingredient of Equation~\ref{eq:modp-combine} costs something and the rest do not.

\begin{itemize}
\item $F_j^{\,q-1}$ is a product of $q-1$ copies of $F_j = G_i\land h_j$ over disjoint variables, and therefore has size $(q-1)\bigl(|G_i|+|h_j|\bigr)$. This is the whole cost of the generalization.
\item The additions $+(q-1)$ and $+K_k$ cost the same as the $+1$ of the parity construction. Adding a constant $c$ is the $+$ operator of Section~\ref{sec:prelim}: one fresh selector variable $x_{new}$, together with a formula over the variables already present having exactly $c$ satisfying assignments. For $c=1$ that formula is the single conjunction $x_1\land\cdots\land x_m$. For a general $c<q$ it is a disjunction of at most $\lceil\log_2 c\rceil$ such conjunctions, one per set bit of $c$. No copies of $F_j$ are created.
\end{itemize}

We therefore ignore the additions below, exactly as the $+1$'s are ignored in the size analysis of Section~\ref{sec:sizeComp}. Under that convention a single repetition at step $i$ has size
\[
(q-1)\bigl(|G_i|+|h|\bigr),
\]
against $|G_i|+|h|$ for $q=2$: \emph{the modulus multiplies the per-repetition size by exactly $q-1$.} Whether the larger modulus buys back enough repetitions to pay for this is the question of the next section.

\begin{remark}[What each construction asks of its sieve]\label{rem:modp-vs-prefixes}
The $q-1$ copies are best understood as the price of a weaker demand, and it is worth setting the two constructions side by side.

Toda and Ogiwara ask their sieve to isolate \emph{exactly one} solution. Valiant and Vazirani's lemma guarantees only that \emph{some} number of hash constraints achieves this, and which number is the right one depends on the size of the solution set and is not known in advance. Writing $n$ for the number of variables being sieved, they therefore hedge over every prefix $W_{k,1}\cdots W_{k,j}$ with $1\leq j\leq n$ at once, so that each repetition contributes $n$ factors, alongside a single further factor, shared by all repetitions, for the unsieved count, which is the empty repetition of Section~\ref{sec:sweep}. In exchange they obtain a constant success probability: $1/4$ in their own accounting~\cite[Claim 1]{TodaOgiwara1992}, and $1/2$ when the vectors are drawn linearly independent, as in Section~\ref{sec:sweep}. What this gives them, beyond the constant, is generality: a count of exactly one is non-zero modulo \emph{every} $k$, so no analogue of Property~\ref{modp:fermat} is needed and their construction applies to composite moduli as it stands. What it costs them, beyond the prefix sweep, is that a larger modulus can offer nothing further.

We ask only that the count be non-zero modulo $q$, which is a weaker demand and lets more than one residue count as a success. The sieves of Section~\ref{sec:replacing-vv} reach a constant success probability from a \emph{single} hash function, so we pay neither the sweep over prefix lengths nor a compensating factor in the number of repetitions. What we pay instead is Property~\ref{modp:fermat}: the surviving count is now an unknown residue, and restoring it to $\{0,1\}$, so that the rules for combining quantifiers apply at all, costs $q-1$ copies and confines us to prime moduli. Since $q$ is a constant while the number of sieved variables is not, this is much the smaller of the two prices.

Two contrasts follow. In theory, one repetition of ours makes $q-1$ copies of the formula, which is a constant number since $q$ is a constant, while one repetition of theirs makes $n$ of them. Both need a number of repetitions linear in the number of outer variables, so over a whole step we duplicate the formula a \emph{linear} number of times and they duplicate it a \emph{quadratic} number of times. In practice the contrast is sharper still: isolation succeeds or fails identically under every modulus, so a different modulus can bring them nothing at all, whereas the ceiling of Theorem~\ref{thm:modq-barrier} still lies above the $1/2$ that Theorem~\ref{thm:half-barrier} imposes at $q=2$. The main question of this section is whether this extra room can be exploited at higher moduli, and whether the additional factor of $q-1$ is repaid.
\end{remark}

\subsection{Modulo \texorpdfstring{$q$}{q} versus modulo \texorpdfstring{$2$}{2}}\label{sec:modq-vs-mod2}

A larger modulus does relax the demand on the sieve. Modulo $2$ a repetition succeeds only if the sieved formula has an odd number of solutions. Modulo $q$ it suffices that the number of solutions is not divisible by $q$, which is a weaker requirement and can therefore be met more often. The relevant ceiling is $\frac{q-1}{q}$, and it is the modulo $q$ analogue of the barrier of Section~\ref{sec:half-barrier}. A \emph{sieve modulo $q$} is a distribution $\mathcal{H}$ over Boolean functions $h:U\to\{0,1\}$, exactly as in Definition~\ref{def:sieve-dist}, but judged by
\[
P^{q}_{\mathcal{H}}(f)\;=\;\Pr_{h\sim\mathcal{H}}\bigl[\,\#(f\land h)\not\equiv 0 \pmod q\,\bigr],
\]
the probability that the sieved formula has a number of solutions that is not divisible by $q$. For $q=2$ this is the $P_{\mathcal{H}}(f)$ of Section~\ref{sec:half-barrier}.

\begin{theorem}\label{thm:modq-barrier}
Let $q\geq 3$, and let $N\geq q+1$ when $q\geq 4$ and $N\geq 6$ when $q=3$. For every sieve $\mathcal{H}$ modulo $q$ there is a satisfiable formula $f$ with
\[
P^{q}_{\mathcal{H}}(f)\;\leq\;c_q\;<\;\frac{q-1}{q},
\qquad
c_q=\begin{cases}
\dfrac{27}{41}, & q=3,\\[8pt]
\dfrac{q(q+3)}{q^{2}+5q+2}, & q\geq 4.
\end{cases}
\]
\end{theorem}

Since $N=2^{n}$, the hypothesis on $N$ is satisfied for every $n$ of interest. The proof is an averaging argument, as for Theorem~\ref{thm:half-barrier}: a minimum is at most a weighted average, so to bound every sieve at once it is enough to find a few families of solution sets against which no sieve can do well on average. Three families suffice (the singletons, the sets of size $q$, and the whole universe), and on each of them whether a given $h$ succeeds depends only on how many elements $h$ selects, which reduces the question to a short table. The full proof is given in Appendix~\ref{app:modq-barrier}.

At $q=2$ the ceiling is attained: Theorem~\ref{thm:half-barrier} gives $\frac12$ and Proposition~\ref{prop:half-optimal} shows it is reached. Theorem~\ref{thm:modq-barrier} says that for every larger modulus it is not. Using the ceiling rather than an achievable value therefore can only \emph{favour} the larger modulus, so the conclusion below is conservative.

\begin{proposition}\label{prop:mod2-optimal}
Assume that for every $q$ the sieve attains the ceiling $\frac{q-1}{q}$, which by Theorem~\ref{thm:modq-barrier} is an over-estimate for every $q\geq 3$, and that the hash functions for all moduli have the same size. Then, among all moduli $q\geq 2$, the formula produced by a single step of the reduction is smallest for $q=2$.
\end{proposition}

\begin{proof}
Fix a step $i$ with $s$ outer variables and inner-error $\epsilon_i$. By Section~\ref{sec:numcomp} the number of repetitions needed when the sieve succeeds with probability $\pi$ is
\[
k_i \;\geq\; \log_{1-\pi}\!\left(\frac{\epsilon_i}{2^{s}}\right) \;=\; \frac{\ln(\epsilon_i/2^{s})}{\ln(1-\pi)}.
\]
By Section~\ref{sec:modp-size} each repetition has size $(q-1)(|f|+|h|)$, so the step produces a formula of size
\[
k_i\cdot (q-1)\bigl(|f|+|h|\bigr) \;=\; \frac{\ln(\epsilon_i/2^{s})}{\ln(1-\pi)}\,(q-1)\bigl(|f|+|h|\bigr).
\]
Substituting $\pi = \frac{q-1}{q}$, so that $1-\pi = \frac1q$ and $\ln(1-\pi) = -\ln q$, and writing $\ln(\epsilon_i/2^s) = -\ln(2^s/\epsilon_i)$, this becomes
\[
\ln\!\left(\frac{2^{s}}{\epsilon_i}\right)\bigl(|f|+|h|\bigr)\cdot \Psi(q),
\qquad\text{where}\qquad
\Psi(q) \;=\; \frac{q-1}{\ln q}.
\]
The first two factors do not depend on $q$, so minimizing the size over $q$ amounts to minimizing $\Psi$. We claim that $\Psi$ is strictly increasing on $(1,\infty)$. Indeed,
\[
\Psi'(q) \;=\; \frac{\ln q - \frac{q-1}{q}}{(\ln q)^2},
\]
and the numerator is $g(q) = \ln q - 1 + \frac1q$, for which $g(1)=0$ and $g'(q) = \frac1q - \frac1{q^2} = \frac{q-1}{q^2} > 0$ for every $q>1$. Hence $g(q)>0$ and so $\Psi'(q)>0$ for all $q>1$. As $q=2$ is the smallest admissible modulus, $\Psi$ is minimized there.
\end{proof}

Numerically, $\Psi(2)=1.443$, $\Psi(3)=1.821$, $\Psi(5)=2.485$ and $\Psi(7)=3.083$, so a step costs $26\%$ more modulo $3$, $72\%$ more modulo $5$ and $114\%$ more modulo $7$ than it does modulo $2$.

Two remarks qualify this conclusion. First, the comparison assumes that hash functions for the two moduli have the same size. This is the one assumption that could fail in favour of a larger modulus: if a sieve modulo $q$ were found whose Boolean encoding is substantially smaller than the encodings we use for $q=2$ in Section~\ref{sec:replacing-vv}, while still attaining a comparable success probability, the balance could change. Whether such a sieve exists we do not know, and we regard it as the one open question that this section leaves. Second, the comparison is about formula size only. It says nothing about whether a model counter finds counting modulo $q$ easier or harder than counting modulo $2$, a question on which, as noted in Section~\ref{sec:discussion}, theory offers little guidance.

\subsection{Composite moduli}\label{sec:modp-composite}

Fermat's little theorem requires $q$ to be prime, so Equation~\ref{eq:modp-combine} is a construction for prime moduli. For theoretical completeness this suffices, because composite moduli reduce to prime ones. Hertrampf~\cite{Hertrampf1990} and Beigel and Gill~\cite{BeigelGill1992} showed that $\mathrm{Mod}_k\mathrm{P}$ is the union of the classes $\mathrm{Mod}_q\mathrm{P}$ over the primes $q$ dividing $k$, and at the level of formulas the reduction is a single multiplication~\cite[Proposition 29]{BeigelGill1992}: if $q$ divides $k$ then for every integer $m$
\[
\tfrac{k}{q}\,m \not\equiv 0 \pmod{k} \quad\Longleftrightarrow\quad m \not\equiv 0 \pmod{q},
\]
so multiplying the number of solutions by $k/q$ converts a test modulo $q$ into a test modulo $k$. As a formula operation this is a $\times$ with a formula having exactly $k/q$ satisfying assignments, which by Section~\ref{sec:modp-size} is negligible. Running the reduction modulo the smallest prime factor of $k$ and then multiplying therefore covers every $k\geq 2$.

It is worth being precise about what this multiplication does and does not give. It is applied to the finished formula, \emph{after} the sieve has already done its work, so it leaves the success probability of the sieve untouched. The guarantee that was used is still the guarantee of the sieve for the modulus at which the reduction was actually run, that is, for the smallest prime factor of $k$. Consequently this multiplication does not reduce the number of repetitions: a composite modulus reached in this way inherits exactly the $k_i$, and hence exactly the size, of the prime it was built from. The construction is a statement about coverage, not about size, and it therefore plays no part in the comparison of Section~\ref{sec:modq-vs-mod2}.

It also does not disturb that comparison. Proposition~\ref{prop:mod2-optimal} minimizes over all moduli and finds the minimum at $q=2$, so enlarging the range by composite moduli cannot yield anything smaller, and nothing is lost by having excluded them.

\section{Engineering and Implementation}\label{sec:implementation}

Sections~\ref{sec:alganal} to~\ref{sec:union-bound} analyzed the size of the formula that the reduction produces, and the factors that govern it. Turning that analysis into a tool requires a further layer of decisions, which the analysis does not determine, and which is the subject of this section. Some of them are free, in the sense that they reduce the size of the formula while costing nothing elsewhere. Others trade one resource against another, and we say which is which as we go. Throughout we keep the notation of Section~\ref{sec:sizeComp}, and write $K=\prod_{i=1}^{d}k_i$ for the total number of copies of $F$ that the reduction produces, which by Equation~\ref{eq:exactSize} is the quantity that governs $|F_0|$.

We proceed as follows. Section~\ref{subsec:preserve-cnf} shows that the reduction can be kept in conjunctive normal form throughout, so that only the hash function requires a Tseitin encoding. Section~\ref{subsec:multi-call} performs the combination of the outermost repetitions outside the counter. Section~\ref{subsec:savings} collects four smaller savings, among them the cancellation of the two consecutive $+1$ operations that Section~\ref{sec:explToda} announced but deferred. Section~\ref{subsec:trades} weighs two choices that trade one resource for another, among them the parameters of the sieve that Sections~\ref{sec:sizeComp} and~\ref{sec:sModAdd} left open, and Section~\ref{subsec:experiments} reports our results.

\subsection{Preserving Conjunctive Normal Form}\label{subsec:preserve-cnf}

Model counters take their input in conjunctive normal form, so the formula that Equation~\ref{eq:exactSize} describes has to be converted before it can be counted. The obvious way of doing so is to build the whole formula and convert it once, by a Tseitin encoding. This turns out to be unnecessary, because the two operators of Section~\ref{sec:prelim} out of which the reduction is assembled can both be realized directly in conjunctive normal form. Let $F$ and $G$ be in conjunctive normal form. Then, recalling the definitions of Section~\ref{sec:prelim}:

\begin{itemize}
\item $F(\x)\times G(\y)$ is $F\wedge G$, which is the concatenation of the two lists of clauses.
\item $F(\x)+G(\y)$ is defined as $[F\wedge x_{new}]\vee[\neg x_{new}\wedge G]$, which is actually equivalent to
\[
(F\vee\neg x_{new})\wedge(G\vee x_{new}),
\]
as one sees by evaluating both expressions at $x_{new}=1$, where each reduces to $F$, and at $x_{new}=0$, where each reduces to $G$. Now a disjunction with a single literal distributes over a conjunction, and $F$ and $G$ are themselves conjunctions of clauses, so that $F\vee\neg x_{new}$ is obtained by appending $\neg x_{new}$ to every clause of $F$, and $G\vee x_{new}$ by appending $x_{new}$ to every clause of $G$. The sum is then the concatenation of the two lists, and it introduces one variable and one further literal per clause.
\end{itemize}

The $+1$ operation is the special case $G=x_1\wedge\cdots\wedge x_m$ of the sum. The clauses of that $G$ are the $m$ units $(x_j)$, so the rule above turns them into the $m$ binary clauses $(x_{new}\vee x_j)$, which are concatenated with the clauses of $F$, each of these carrying $\neg x_{new}$.

Modular addition needs the sum of $l$ formulas, and Section~\ref{sec:prelim} noted that this takes $\lceil\log_2 l\rceil$ fresh variables. It is done as follows. Iterating the binary sum, as $(((F_1+F_2)+F_3)+\cdots)+F_l$, would work but is wasteful. It uses $l-1$ fresh variables, and since every $+$ appends its own literal to all clauses built so far, the clauses of $F_1$ end up $l-1$ literals longer. Instead take $b=\lceil\log_2 l\rceil$ fresh variables $z_{b-1},\ldots,z_0$, which we call selectors, and give each summand a distinct $b$-bit code. To every clause of $F_i$ append $b$ literals, one per selector: $z_j$ if bit $j$ of the code of $F_i$ is $0$, and $\neg z_j$ if it is $1$. For $l=4$ the codes are $00,01,10,11$, so every clause of $F_1$ gets $(z_1\vee z_0)$, of $F_2$ gets $(z_1\vee\neg z_0)$, of $F_3$ gets $(\neg z_1\vee z_0)$, and of $F_4$ gets $(\neg z_1\vee\neg z_0)$. The appended literals are all false exactly when the selectors spell the code of $F_i$. For that value of the selectors the clauses of $F_i$ are in force, and every other summand has a true literal in each of its clauses and imposes nothing. The value therefore contributes $\#F_i$ models, and the $l$ codes together contribute $\sum_i\#F_i$, as required. If $l$ is not a power of two, the selectors have $2^b-l$ values that are the code of no summand. Such a value satisfies every summand and would add a stray count, so it is ruled out by one clause of width $b$, the clause that is false exactly at that value. For $l=3$ the codes are $00,01,10$ and the value $11$ is ruled out by $(\neg z_1\vee\neg z_0)$. All told, the sum of $l$ formulas costs $b$ selectors, $b$ extra literals on every clause of every summand, and $2^b-l$ exclusion clauses.

The only component of the construction that genuinely requires a Tseitin encoding is therefore the hash function, whose circuit is not a conjunction of clauses, and Appendix~\ref{sec:multilinearTechnique} has already accounted for the size of its encoding. In particular the input matrix is never re-encoded at all, but enters the construction as the formula it already was and remains so throughout.

A Tseitin encoding, by contrast, introduces an auxiliary variable for every gate, and hence one for every clause of each of the $K$ copies of $F$. Writing $W$ for the number of clauses of $F$, converting at the end rather than maintaining the form throughout would cost $K(W+1)$ additional variables and at least $2K(W+1)$ additional clauses.

There is a trap in combining the sum with the Tseitin encoding of the hash functions. A hash function is a circuit, and its encoding introduces an auxiliary variable for every gate, with clauses that fix the auxiliary to the value of its gate. For a fixed assignment to the inputs these clauses leave exactly one consistent value for every auxiliary, which is why the auxiliaries do not change the count. That holds only while the defining clauses are in force. Suppose now that $h_1+\cdots+h_l$ is built by applying the sum above to the encodings of the $h_i$. For the code of $h_i$ the clauses of every other $h_j$ are satisfied by their selector literals, including the clauses that define the auxiliaries of $h_j$, so those auxiliaries are free. Each free variable doubles the count, so the count is even whatever the parity of the sum.

The implementation avoids this by leaving the encodings as they are and not applying the sum to them. Each encoding has an output variable $o_i$ that is true exactly when $h_i$ accepts the inputs, and the sum is applied to the $l$ unit clauses $(o_1),\ldots,(o_l)$ alone. This produces the $l$ clauses $(L_i\vee o_i)$, where $L_i$ is the set of selector literals of the code of $h_i$, together with the $2^b-l$ exclusion clauses, and nothing else. For $l=4$ they are
\[
(z_1\vee z_0\vee o_1),\quad(z_1\vee\neg z_0\vee o_2),\quad(\neg z_1\vee z_0\vee o_3),\quad(\neg z_1\vee\neg z_0\vee o_4).
\]
To see that the count is right, fix the original variables to an assignment that satisfies $F$. This fixes every auxiliary variable and every $o_i$, and the only freedom left is in the selectors, which may spell the code of any hash function that accepts the assignment, and nothing else. The assignment therefore extends to exactly one model per accepting hash function, and the number of models is $\sum_i\#[F\land h_i]$, as it should be. In addition, the formula gets smaller, since the selector literals now appear only in these $l$ short clauses and the exclusion clauses, and do not appear in or lengthen any other clause.

\subsection{Multiple Calls at the Outermost Step}\label{subsec:multi-call}

At the outermost step the $k_1$ repetitions are combined by $\times$, that is, by conjunction over disjoint sets of variables, and what the reduction asks of them is only that \emph{some} repetition succeed. This combination therefore need not take place inside the formula. The implementation instead emits the repetitions as $k_1$ separate instances and combines their answers outside the counter: if $Q_1$ is existential the answer is \true as soon as one instance has an odd number of solutions, and if $Q_1$ is universal the answer is \false as soon as one instance does.

Three benefits follow. Each instance is smaller than the combined one by a factor of $k_1$, and since by Equation~\ref{eq:exactSize} the outermost repetitions multiply everything beneath them, this is the largest single reduction in the size of an instance available anywhere in the construction. The calls are independent, sharing neither variables nor state, so that they may be performed concurrently, and on $k_1$ processors the elapsed time is that of one instance rather than of $k_1$. Finally the combination is one-sided, so that the procedure may stop at the first instance with an odd count, whether $Q_1$ is existential or universal.

\subsection{Further Savings in the Construction}\label{subsec:savings}

\paragraph{Cancelling the $+1$ operations.} Combining the $k_i$ repetitions of a step by Property~\ref{prop:parity:disjunction} of Claim~\ref{clm:many} produces $[(F_1+1)\times\cdots\times(F_{k_i}+1)]+1$, and propagating the negation of the outer prefix by Property~\ref{prop:parity:negation} adds a further $+1$. As Section~\ref{sec:explToda} observes, two consecutive $+1$ operations under a parity quantifier cancel, since $+1+1=+2=0$. The reduction was presented there without the cancellation, to remain faithful to the original proof. The implementation performs it, and emits neither. When $Q_1$ is existential the outermost step has no negation to propagate and its $+1$ stands alone. The implementation does not encode it either. By Property~\ref{prop:parity:negation} it only negates the answer, so the counter is asked about the formula without it and the result is negated outside. With the multiple calls of Section~\ref{subsec:multi-call} the product is never formed and the question does not arise.

\paragraph{Merging quantifier blocks.} Adjacent blocks carrying the same quantifier are merged into one before the reduction begins. This matters more here than it would for a search-based solver, because the quantity that governs the allocation of Section~\ref{sec:innerComp} is the number of blocks: a prefix written as $\exists\x_1\exists\x_2\forall\y$ has three blocks but only two alternations, and merging it removes an entire elimination step, together with the factor $k_i$ that step would have contributed to Equation~\ref{eq:exactSize}.

\paragraph{Removing an innermost universal block.} When $Q_d$ is universal, Section~\ref{sec:explToda} propagates the innermost negation into $\hat{F}$. An alternative is to remove the block altogether, by the universal reduction rule of QBF solving. A clause that contains a variable of the block in both polarities is a tautology and is deleted. From every other clause the literals of the block are deleted, since the block can always be set to falsify them. A clause that becomes empty makes $F$ false. Both operations preserve the Boolean function that the remaining variables induce, and not merely the truth value of $F$, so the block can be removed and the reduction begun on what remains, saving an elimination step.

\paragraph{Projected counting.} The auxiliary variables introduced by the Tseitin encoding of the hash function are functionally determined by the variables from which they are defined: once the inputs of the hash function are fixed, each such variable has exactly one consistent value. They therefore contribute a factor of $1$ to the number of solutions, and the counter need not branch on them. Projecting the auxiliaries away would also have cured the trap of Section~\ref{subsec:preserve-cnf} on its own, but with the encoding used there projection is an optimization and nothing more. The implementation accordingly presents the instance as a \emph{projected} counting problem, in which the count is taken over a declared subset of the variables, a form that current counters exploit directly.

\subsection{Trading One Resource for Another}\label{subsec:trades}

\paragraph{Majority vote.} A standard way of reaching a small error is to run the whole procedure $r$ times with a larger error $\epsilon'$ at each run, and return the majority answer. That this is not an obvious improvement follows from Section~\ref{sec:equivclasses}, where $k_i\geq\log_{1-p}(\epsilon_i/C)$, which for $p=1/2$ reads $k_i=\log_2 C+\log_2(1/\epsilon_i)$: the inner-error enters the count of repetitions only logarithmically, so that sharpening $\epsilon$ is cheap. Taking $C=2^{12}$, $p=1/2$ and the allocation of Theorem~\ref{thm:boundedalloc}, moving from $\epsilon=0.30$ to $\epsilon=0.05$ costs two or three further repetitions per step at every depth up to $d=8$, whereas reaching the same $0.05$ by a majority vote from $\epsilon'=0.30$ requires $r=17$ runs.

Adding up the sizes of those $r$ formulas would make the majority vote much the worse of the two, but that presumes the cost of counting to be linear in the size of the instance, which it need not be. The $r$ runs are disjoint formulas, counted independently, so that the quantity to compare is $r$ times the cost of a small instance against the cost of a single larger one. At $d=4$ the trade is a factor of $17$ in the number of calls against a factor of $2.07$ in the size of each, and since model counting is worst-case exponential it may well be worth making. The runs are also independent, and so may be performed concurrently. We would not expect this to hold in general, and treat the majority vote as an option to be enabled per instance rather than as a default.

\paragraph{The parameters of the sieve.} Section~\ref{sec:sizeComp} set the optimization of $|h_i|$ aside, on the ground that it enters Equation~\ref{eq:exactSize} additively while $k_i$ enters it multiplicatively. An implementation must nevertheless fix the parameters of the sieve, and since these determine $p$ they determine $k_i$ as well.

For modular addition the parameter is the number $l$ of summands of Definition~\ref{def:HMA}. Raising $l$ raises the success probability $\frac{1}{2}(1-(1-2p)^{l})$ and so lowers $k_i$, while each repetition grows to $(1+l\alpha)|F_i|$. Section~\ref{sec:sModAdd} took $l=n$ for a clean computation, and noted that the best $l$ depends on the formula at hand. For the multilinear hash function the parameter is $c$ in $m=\lceil\log_2 n\rceil+c$ of Section~\ref{subsec:multilinear-hash}, where every additional field bit halves the distance from $p$ to $1/2$ while $|h_{a,v}|$ grows quadratically in $m$. The implementation uses the multilinear hash function as the summand of modular addition, and searches for the best pair $(c,l)$ at run time, at every elimination step. Everything the size of the step depends on is known by then, so each candidate pair costs a few arithmetic operations and nothing is built. All of it takes a negligible fraction of a single call to the counter, which is also why $l$ is searched even though Section~\ref{sec:compareMultilinear} found that the combination rarely pays.

\subsection{Experimental Evaluation}\label{subsec:experiments}

Every run had one core of an Intel Xeon Gold 6130, 12 GB of memory and 900 seconds. The limit covers everything the run does: the preprocessor when one is used, the construction of the formula, writing it to disk, and every call to the counter. The counter is GANAK~\cite{SharmaRSM19}, version 2.6.3, in the parity mode of~\cite{SM2025}, which returns the parity of the number of models rather than the number itself. The solvers we compare with are DepQBF 6.03~\cite{DepQBF}, CAQE 4.0.1~\cite{CAQE} and QESTO 1.0~\cite{qesto}, run under the same limits, and the preprocessor is bloqqer~\cite{bloqqer}, given at most 300 of the 900 seconds. The error bound is $\epsilon=0.1$, with one run per instance and no majority vote.

We use two benchmarks. The first is a ladder of 326 crafted instances in 33 families. Each family is a sequence of instances of growing size, long enough for every configuration of the tool to fail on it. Fifteen families come from the QBFFam generator~\cite{QBFFAM}, at $n=2$ to $50$, and include the KBKF, equality, parity, completion, Lonsing and trapdoor formulas. The other eighteen were generated for this work. Five follow a gadget construction from~\cite{SchleitzerBeyersdorff23}, and the rest are our own designs, among them graph colouring, matching and random 3-QBF. The second benchmark is the QDIMACS archive of QBFLIB~\cite{QBFLIB}, 27,929 instances, of which 21,392 have two or three quantifier blocks once adjacent blocks with the same quantifier are merged. The comparison with the solvers covers every family but the two largest, Narizzano and Jordan-Kaiser, which were not run, 13,018 of these instances, 18 of which have a universal block whose variables occur in no clause and are left out, leaving 13,000. Every instance is run without preprocessing, and again after bloqqer whenever bloqqer returns a formula with two or three blocks within its cap. An instance that bloqqer solves on its own is not counted for any solver. An instance that bloqqer reduces to a single quantifier block has become a propositional formula, which is not our subject, so it is not counted either.

We ran the tool in five configurations, adding the contributions of this paper one at a time so that each can be seen on its own. The \emph{textbook} configuration is the reduction of Section~\ref{sec:todaBrief} as a textbook presents it: the hash function of Theorem~\ref{thm:vv} with $p\geq1/8n$, the geometric allocation of Section~\ref{sec:inner_epsilon_naive}, the union bound $2^s$ of Proposition~\ref{clm:repetitionUNbound}, one Tseitin conversion of the whole formula at the end, and one call to the counter. The \emph{sweep} configuration replaces the hash function by the sweep of Section~\ref{sec:sweep}, with its shared empty repetition and linearly independent rows, so $p=1/2$. The \emph{sieve} configuration replaces the sweep by the multilinear hash function of Section~\ref{subsec:multilinear-hash}, with $c$ and $l$ chosen at every step by the search of Section~\ref{subsec:trades}. The \emph{allocation and bound} configuration adds the balanced allocation of Theorem~\ref{thm:boundedalloc} and the bound $C\leq2^{|B|}$ of Proposition~\ref{thm:boundaryC} in place of $2^s$. The \emph{full} configuration adds the engineering of this section: the clausal form of Section~\ref{subsec:preserve-cnf}, the separate calls of Section~\ref{subsec:multi-call} and the savings of Section~\ref{subsec:savings} (other than the merging of blocks, which every configuration performs). The first four keep the single Tseitin conversion and the single call, so that what separates them is the analysis alone.

Table~\ref{tab:ladder} gives the results on the ladder without preprocessing, and Figure~\ref{fig:cactus} shows how many instances each configuration solves within a given time. Each configuration solves everything the one before it solves, and the count climbs from 3 to 14, 21, 23 and 47 of the 326 instances. The lower configurations fail on memory rather than on time: the textbook configuration runs out of memory on 315 instances, the sweep still on 283. The last three columns of the table compare each configuration with the one directly above it, on the instances the two have in common, those on which both built the formula and were not killed for memory for the two sizes, and those both solved for the counting time, and give the median ratio of formula size, largest call and counting time. The sweep gains a factor of 53 on the textbook configuration. The gain grows with the depth, from 23 on the two-block instances to 88 on the three-block ones, as it should, since by Equation~\ref{eq:exactSize} the factors of the steps multiply. The sieve and then the allocation each gain another 1.45. The instances the sweep and the sieve have in common have fewer than twenty variables, and the sweep does not finish on larger ones, so the factor of $n$ that Section~\ref{sec:compareMultilinear} promises is out of reach of this comparison. The full configuration divides the formula by another 5, and the largest formula passed to the counter by 25. From the textbook configuration to the full one, the formula shrinks by a factor of 435 and the largest call by 2,149. One instance may serve as an example. PARITYTrue of size 3 has two blocks and six variables. The five configurations turn it into 953,190, 42,391, 19,645, 10,565 and 6,064 clauses, the last in four calls of at most 1,575 clauses each, and solve it in 163, 2.8, 1.5, 0.7 and 0.4 seconds. The three solvers solve 269, 224 and 210 of the 326 instances. Every instance the full configuration solves, all three solve as well, in about a second at most. The five KBKF families, with prefixes of a hundred blocks and more, are out of reach of every configuration of the tool, while the solvers get to $n=18$ on them.

After bloqqer there is little left to say. The preprocessor solves 249 of the 326 instances on its own and reduces 42 more to a single block. Of the 35 that remain, no configuration of the tool solves any within the limit, and the solvers solve 20, 19 and 20.

\begin{table}[!htb]
\centering
\small
\begin{tabular}{l|r|r|r|r|r|r}
\hline
\textbf{Configuration} & \textbf{Solved} & \textbf{Timeout} & \textbf{Out of memory} & \multicolumn{3}{c}{\textbf{Gain on the row above, median ratio}} \\
 & & & & \textbf{formula size} & \textbf{largest call} & \textbf{counting time} \\
\hline
textbook & 3 & 7 & 315 & & & \\
sweep & 14 & 29 & 283 & 53 (10) & 53 (10) & 63 (3) \\
sieve & 21 & 13 & 292 & 1.45 (30) & 1.45 (30) & 1.06 (14) \\
allocation and bound & 23 & 16 & 287 & 1.45 (33) & 1.45 (33) & 1.74 (21) \\
full & 47 & 170 & 109 & 5.1 (38) & 25 (38) & 3.6 (23) \\
\hline
DepQBF & 269 & 57 & 0 & & & \\
CAQE & 224 & 102 & 0 & & & \\
QESTO & 210 & 116 & 0 & & & \\
\hline
\end{tabular}
\caption{The crafted ladder without preprocessing. The last three columns are median ratios against the row above, over the instances counted in parentheses. The textbook configuration crashed on one further instance.}
\label{tab:ladder}
\end{table}

\begin{figure}[!htb]
    \centering
    \includegraphics[width=0.85\linewidth]{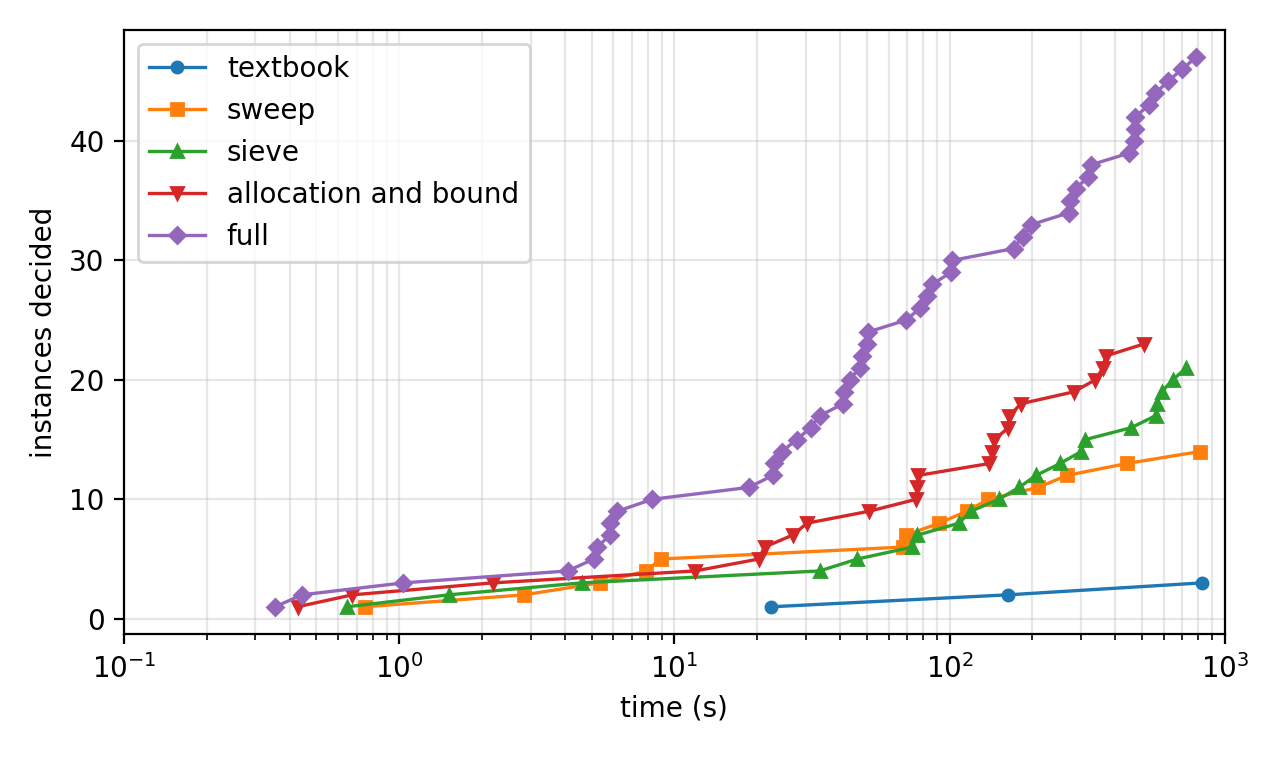}
    \caption{The crafted ladder without preprocessing: how many instances each configuration solves within a given time.}
    \label{fig:cactus}
\end{figure}

Table~\ref{tab:laddermax} gives, for each configuration, the largest instance it solved on the ladder without preprocessing, by variables, by quantifier blocks and by clauses, each maximum taken on its own. The largest instance the full configuration solved is one of our own families with 10 outer and 800 inner variables, and it also gives the maximum in clauses. The largest outermost block it solved has 22 variables and is a Lonsing instance with 26 variables in total.

\begin{table}[!htb]
\centering
\small
\begin{tabular}{l|r|r|r|r}
\hline
\textbf{Configuration} & \textbf{Solved} & \textbf{Variables} & \textbf{Quantifier blocks} & \textbf{Clauses} \\
\hline
textbook & 3 & 6 & 3 & 10 \\
sweep & 14 & 12 & 3 & 26 \\
sieve & 21 & 12 & 3 & 48 \\
allocation and bound & 23 & 26 & 3 & 90 \\
full & 47 & 810$^*$ & 3 & 3,997 \\
\hline
\end{tabular}
\caption{The largest ladder instance each configuration solved, by variables, quantifier blocks and clauses. $^*$Ten outer and 800 inner variables.}
\label{tab:laddermax}
\end{table}

On the 47 instances that the full configuration solves without preprocessing, the counter takes 97.5\% of the time and the construction of the formula 2.5\%. The medians are 50 seconds for the counter and 0.7 seconds for the construction. The search of Section~\ref{subsec:trades} takes 0.2 milliseconds. This is what that section assumed: the construction costs nothing, and only the counter matters. What limits the full configuration is the counter rather than the size of the formula. Of the 170 ladder instances on which it timed out, 138 built their formula, and 120 of those had a smaller largest call than the largest it ever counted, 6.0 million clauses. No instance with more than 22 variables in its outermost block was solved.

Table~\ref{tab:qbfeval} gives the results on QBFEVAL. On general benchmarks the union bound makes our formulas far too large, as Section~\ref{sec:union-bound} predicts, and no engineering changes that. The comparison is there to measure the gap. The full configuration solves 38 of the 13,000 instances without preprocessing and 15 of the 8,506 that bloqqer leaves with two or three blocks, against 6,890 and 4,778 for the best solver on each form. Every instance it solves is solved by all three solvers as well. Without preprocessing the solvers need at most ten seconds on these instances. The 38 are small, 160 variables at most and 40 in the median, with an outermost block of at most 20 variables, and come from eight families, 11 each from Biere's and Kronegger's. Of the 15 solved after preprocessing, ten are Tacchella instances with 177 variables and about 4,200 clauses. Two are Miller-Marin instances of 386 and 469 variables with only two in the outermost block. Three are Tentrup instances of about 265 variables, solved in 219 to 498 seconds, and these are the only ones that cost the solvers real time, between half a minute and five minutes. The lower configurations solve 0, 4, 0 and 0 instances, and the sweep crashes on 119 more that would need billions of copies of the formula. Interestingly, the four that the sweep solves are lost again by the sieve, and the full configuration solves them within 27 seconds, where the sweep needed 236 to 711 seconds.

\begin{table}[!htb]
\centering
\small
\begin{tabular}{l|r|r}
\hline
 & \textbf{Without preprocessing, 13,000 instances} & \textbf{After bloqqer, 8,506 instances} \\
\hline
DepQBF & 6,890 & 4,273 \\
CAQE & 6,046 & 4,778 \\
QESTO & 5,807 & 4,301 \\
\hline
textbook & 0 & 0 \\
sweep & 4 & 0 \\
sieve & 0 & 0 \\
allocation and bound & 0 & 0 \\
full & 38 & 15 \\
\hline
\end{tabular}
\caption{Instances solved on QBFEVAL, without preprocessing and after bloqqer.}
\label{tab:qbfeval}
\end{table}

Table~\ref{tab:qbfevalmax} gives, for each configuration, the largest instance it solved on QBFEVAL by variables, by quantifier blocks and by clauses. It says what the tool handles today: inputs of at most a few hundred variables and a few thousand clauses, with two or three blocks.

\begin{table}[!htb]
\centering
\small
\begin{tabular}{l|r|r|r|r}
\hline
\textbf{Configuration} & \textbf{Solved} & \textbf{Variables} & \textbf{Quantifier blocks} & \textbf{Clauses} \\
\hline
textbook & 0 & & & \\
sweep & 4 & 24 & 3 & 37 \\
sieve & 0 & & & \\
allocation and bound & 0 & & & \\
full & 53 & 469$^*$ & 3 & 4,240 \\
\hline
\end{tabular}
\caption{The largest QBFEVAL instance each configuration solved, by variables, quantifier blocks and clauses. $^*$Two outer variables, as in the second largest, of 386. The largest with a sizable outermost block has 267 variables, 11 of them outermost, and the largest outermost block among the solved instances has 20 variables, in an instance of 160.}
\label{tab:qbfevalmax}
\end{table}

Over all runs of the full configuration on instances with a known answer, every ladder and QBFEVAL run included, 734 in total, 1.4\% of the answers were wrong, well below the 10\% that the configuration allows.

\section{Conclusion}\label{sec:discussion}

In this work we revisited Toda's theorem from an algorithmic perspective. Our objective was to determine under what conditions Toda's reduction can serve as the basis of an algorithm that solves QBF formulas with few alternations. To this end we converted the reduction into an explicit algorithm that uses a model counter as an oracle, derived exact values for each of its parameters, and identified the three factors that drive the blow-up of the formula it constructs: the allocation of the error budget among the steps, the success probability of the sieve, and the union bound over the outer variables. The first two factors are now settled from both sides. The error analysis behind the allocation is exact, with a worst case that is attained. The sieve reaches a success probability arbitrarily close to $1/2$, and no sieve of its kind can do noticeably better. Changing the modulus does not help either, as the modulus $2$ produces the smallest formulas.

What remains is the union bound, and it is the factor that determines whether the approach can be practical. We showed that the parameter that governs it is $C$, the number of distinct inner functions that the outer assignments induce, and that the union bound over the $C$ classes is tight. The consequence is dual. For formulas in which $C$ is large, the blow-up is unavoidable, and no engineering will place the resulting formulas within the reach of any current model counter. For formulas in which $C$ is small, the allocation, the sieves and the engineering of this paper bring the resulting formulas to sizes that modern counters handle. The central question that this work leaves open is therefore to bound $C$, or to recognize formulas with a small $C$, cheaply and in advance. The boundary-variable bound of Section~\ref{subsec:high-dependency} is a first step, and a solver that routes exactly the recognized formulas to our procedure, as part of a portfolio, is in our view the realistic path to practical use. Whether meaningful QBF families with a small $C$ arise in practice is an empirical question, and an affirmative answer would make the procedure a practical tool.

The analysis also leaves sharper questions on its own side. The balanced allocation assigns every step an equal budget, and although the analysis behind it is exact, an unequal division tailored to the formula at hand could yield smaller formulas still, a computation that we leave to numerical methods. Our barrier at $1/2$ covers every sieve that is drawn without seeing the formula, and for efficient procedures that may inspect it, the range between $1/2$ and $2/3$ recorded by~\cite{DellKMW13} remains open. A sieve for counting modulo $q$ whose Boolean encoding is substantially smaller than ours would revive the case for larger moduli, and we do not know whether one exists. Finally, we observed that formulas built with modular addition were sometimes counted faster than smaller formulas carrying the same guarantee, which suggests that the structure of a sieve matters to the counter as much as its size, a phenomenon that deserves a study of its own.

Two further questions concern the procedure itself. Our algorithm implements the first, randomized half of Toda's theorem, and making it deterministic, by making the derandomization half of the theorem algorithmic as well, appears to be a challenge in its own right. The procedure also returns only a yes or no answer, and extending it to produce a witness, or a strategy, would be needed wherever a certificate is expected. Last is the oracle. Model counters can now answer parity queries directly, and faster than they count~\cite{SM2025}, yet theory provides little guidance on the gap between parity counting, counting modulo larger values, and exact counting. We do not even know whether $P^{\oplus P}$ contains $NP$, nor whether the polynomial hierarchy is contained in $P^{\oplus P}$. A better understanding of that gap would tell us how much the oracle side of the procedure can still improve.

\clearpage

\bibliographystyle{\mainbibstyle}
\bibliography{main}

\clearpage

\appendix
\ifarxiv\else\section*{Appendix}\fi
\section{Elements in Toda's Reduction}\label{app:TodaElements}




We discuss details of Toda's reduction as explained in Section~\ref{sec:explToda}. For that, we first give properties of the parity operator.

\subsection{Properties of the Parity Operator}\label{sec:propPar}

We restate Claim~\ref{clm:many} of Section~\ref{sec:prelim}, in which $F(\x)$ and $G(\y)$ are Boolean formulas over disjoint variable sets $\x$ and $\y$, and prove it.

\parityproperties*

Proving these properties is not hard, but we nevertheless provide exact proofs.

\begin{proof}
\thinspace
\begin{enumerate}
     \item 
    The claim follows from the definitions. $\neg\parity\x F(\x) = True $ if and only if $\parity\x F(\x) = False$,
    
    which means that
    $\#F(\x) \mod{2} = 0$. Following the definition of the $+$ operator, we then have that $\#[F(\x)]+1 \mod{2} = 1$, which means that
    $\#[F(\x)+1] \mod{2} = 1$. This is the same as  $\parity\x(F(\x)+1) = True$.
     \item 

     Let $S_x^o=\{\sigma_{\x} \mid \parity\y F(\sigma_{\x},\y) = True \}$ denote the set of assignments for $\x$ for which the number of assignments for $\y$ that make $F(\x,\y)$ \true is odd.
     Similarly, let $S_x^e=\{\sigma_x \mid \parity\y F(\sigma_{\x},\y) = False \}$ denote the set of assignments for $\x$ for which that number is even.
    Then $\{S_x^o,S_x^e\}$ partitions the set of all assignments for $\x$. Now we have $\parity(\x, \y) F(\x,\y) = True$ if and only if  $\#F(\vec{x},\vec{y}) \mod{2} = 1$. That means that 
     $(\sum_{\sigma_{\x}\in S_x^o} \#F(\sigma_{\x},\y) + \sum_{\sigma_{\x}\in S_x^e} \#F(\sigma_{\x},\y)) \mod{2} =1$, which means that $\sum_{\sigma_{\x}\in S_x^o} \#F(\sigma_{\x},\y) \mod{2} +\sum_{\sigma_{\x}\in S_x^e} \#F(\sigma_{\x},\y) \mod{2} =1$. By definition, for every $\sigma_{\x}\in S^e_x$, $\#F(\sigma_{\x},\y) = 0 \mod{2}$. Then $\sum_{\sigma_{\x}\in S_x^e} \#F(\sigma_{\x},\y) \mod{2} = 0 \mod {2}$. So we have that $\sum_{\sigma_{\x}\in S_x^o} \#F(\sigma_{\x},\y) \mod{2} = 1 \mod{2}$. Again, by definition, for every $\sigma_{\x}\in S^o_x$, $\#F(\sigma_{\x},\y) = 1 \mod{2}$. Since an even sum of odds is even, and an odd sum of odds is odd, we conclude that $S^o_x$ is \odd. That means that the number of $\x$ assignments, for which the number of $\y$ assignments that make $F(\x,\y)$ \true is \odd\kern-0.35em, is itself \odd\kern-0.35em. But that is the same as saying $\parity\x \parity \y F(\x,\y)= True$. Since all proof steps are ``if and only if'' the proof is complete.

      \item

 We have that $\parity\x F(\x)\land \parity\y G(\y) = True$ if and only if $(\#F(\x) \mod{2} = 1) \land (\#G(\y) \mod{2} = 1)$. Since a product of odd numbers is odd, this happens if and only if $(\#[F(\x)]\times\#[G(\y)]) \mod{2} = 1$, which happens, since $\x,\y$ are disjoint, if and only if $(\#[F(\x) \times G(\y)]])\mod{2} = 1$. This is the same as $\parity(\x,\y)[F(\x)\land G(\y)] = True$.
      
     \item 

    The proof is straightforward from parts 1 and 3 of this claim and from using De Morgan. Specifically we have:
    \begin{multline*}
      \parity\x F(\x)\lor \parity\y G(\y) = \neg\neg(\parity\x F(\x) \lor \parity\y G(\y)) = \\
      \neg(\neg\parity\x F(\x) \land \neg\parity\y G(\y)) = \neg(\parity\x (F(\x)+1) \land \parity\y (G(\y)+1)) = \\
      \neg(\parity(\x,\y)[(F(\x)+1) \land (G(\y)+1)] = \parity(\x,\y)[(F(\x)+1) \land (G(\y)+1)]+1
    \end{multline*}
     

     \item 

     Note that both formulas have free variables $\y$. Let $\sigma_{\y}$ be an assignment for $\y$ for which $([\parity\x F(\x,\sigma_{\y})]\land G(\sigma_{\y})) = True$. Then that means that both $(\parity\x F(\x,\sigma_{\y})=True$ and $G(\sigma_{\y})) = True$. Then, for every assignment $\sigma_{\x}$ for $\x$ we have that $F(\sigma_{\x},\sigma_{\y})=True$ if and only if $F(\sigma_{\x},\sigma_{\y})\land G(\sigma_{\y})=True$.
     Then it means that the parity of the assignment for $\x$ that make $F(\x,\sigma_{\y})\land G(\sigma_{\y})$ is also \odd, which means $\parity\x [F(\x,\y)\land G(\y)]=True$.
     
     On the other hand, assume that $\sigma_{\y}$ is an assignment for $\y$ for which $([\parity\x F(\x,\sigma_{\y})]\land G(\sigma_{\y})) = False$.
     Then if $G(\sigma_{\y})$ is \false then there are no assignments for $F(\x,\sigma_{\y})\land G(\sigma_{\y})$, hence its parity is \even (zero assignments). Otherwise we have that $G(\sigma_{\y})$ is \true, and again for every assignment $\sigma_{\x}$ for $\x$ we have that $F(\sigma_{\x},\sigma_{\y})=True$ if and only if $F(\sigma_{\x},\sigma_{\y})\land G(\sigma_{\y})=True$.
     Then it means that the parity of the assignment for $\x$ that make $F(\x,\sigma_{\y})\land G(\sigma_{\y})$ is also \even, which means $\parity\x [F(\x,\y)\land G(\y)]=False$.

\end{enumerate}
\end{proof}



\subsection{Details of Toda's Reduction}\label{sec:detreduction}

We now provide detailed steps of transforming formula $F'_i$ from Equation~\ref{eq:beginF_i} to formula $F''_i$ in Equation~\ref{eq:endF_i}. We do that by applying the properties of the parity quantifier from Section~\ref{sec:propPar}.

Starting with:
\begin{multline*}
F'_i\equiv\vec{Q}\x\neg([\parity\x_i \,(\parity\y_i G_i(\x,\x_i,\y_i) \land h_1(\x_i))] \lor \ldots \lor [\parity\x_i \, \\(\parity\y_i G_i(\x,\x_i,\y_i) \land h_k(\x_i))])
\end{multline*}

Denote $\parity\y_i G_i(\x,\x_i,\y_i) \land h_j(\x_i)$ as $H_j(\x,\x_i)$ for brevity. So we have:
\begin{equation*}
F'_i\equiv\vec{Q}\x\neg([\parity\x_i \, H_1(\x,\x_i)] \lor \ldots \lor [\parity\x_i \, H_k(\x,\x_i)])
\end{equation*}

\textbf{Step 1:} Apply De Morgan's law converting the disjunction to a negated conjunction of negated terms.
\begin{equation*}
F'_i\equiv\vec{Q}\x\neg\neg([\neg\parity\x_i \, H_1(\x,\x_i)] \land \ldots \land [\neg\parity\x_i \, H_k(\x,\x_i)])
\end{equation*}

\textbf{Step 2:} Use Property 1 from Claim~\ref{clm:many} to insert the negation next to the parity using the $+1$ operator.
\begin{equation*}
F'_i\equiv\vec{Q}\x\neg\neg([\parity\x_i \, (H_1(\x,\x_i)+1)] \land \ldots \land [\parity\x_i \, (H_k(\x,\x_i)+1)])
\end{equation*}

Expanding $H_j$ back to its original form we get:
\begin{multline*}
F'_i\equiv\vec{Q}\x\neg\neg([\parity\x_i \, ((\parity\y_i G_i(\x,\x_i,\y_i) \land h_1(\x_i))+1)] \land \ldots \\ \land [\parity\x_i \, ((\parity\y_i G_i(\x,\x_i,\y_i) \land h_k(\x_i))+1)])
\end{multline*}

\textbf{Step 3:} Applying Property 5 to push the hash functions into the inner parity quantifier.
\begin{multline*}
F'_i\equiv\vec{Q}\x\neg\neg([\parity\x_i \, ((\parity\y_i [G_i(\x,\x_i,\y_i) \land h_1(\x_i)])+1)] \land \ldots \\ \land [\parity\x_i \, ((\parity\y_i [G_i(\x,\x_i,\y_i) \land h_k(\x_i)])+1)])
\end{multline*}

\textbf{Step 4:} Use Property 2 to combine the nested parity quantifiers.
\begin{multline*}
F'_i\equiv\vec{Q}\x\neg\neg([\parity(\x_i,\y_i) \, ([G_i(\x,\x_i,\y_i) \land h_1(\x_i)])+1)] \land \ldots \\ \land [\parity(\x_i,\y_i) \, ([G_i(\x,\x_i,\y_i) \land h_k(\x_i)])+1)])
\end{multline*}

\textbf{Step 5:} Currently we have conjunction of parity quantified formulas. In order to combine them and result in a single quantified formula we wish to use Property 3. However, this requires disjoint variables. To satisfy this property we duplicate variables and refresh the variables names. For each repetition $j$ we replace $\x_i$ with $\x_i^j$ and $\y_i$ with $\y_i^j$ resulting in :
\begin{multline*}
F'_i\equiv\vec{Q}\x\neg\neg([\parity(\x_i^1,\y_i^1) \, ([G_i(\x,\x_i^1,\y_i^1) \land h_1(\x_i^1)])+1)] \land \ldots \\ \land [\parity(\x_i^k,\y_i^k) \, ([G_i(\x,\x_i^k,\y_i^k) \land h_k(\x_i^k)])+1)])
\end{multline*}

\textbf{Step 6:} We combine the conjunction of the parity quantified formulas by applying Property 3 repeatedly. We get:
\begin{multline*}
F'_i\equiv\vec{Q}\x\neg\neg\parity(\x_i^1,\y_i^1,\ldots,\x_i^k,\y_i^k) \, [(([G_i(\x,\x_i^1,\y_i^1) \land h_1(\x_i^1)])+1) \times \ldots \\ \times (([G_i(\x,\x_i^k,\y_i^k) \land h_k(\x_i^k)])+1)]
\end{multline*}

\textbf{Step 7:} We use Property 1 again to push one of the negations inside the parity quantifier:
\begin{multline*}
F'_i\equiv\vec{Q}\x\neg\parity(\x_i^1,\y_i^1,\ldots,\x_i^k,\y_i^k) \, [(([G_i(\x,\x_i^1,\y_i^1) \land h_1(\x_i^1)])+1) \times \ldots \\ \times (([G_i(\x,\x_i^k,\y_i^k) \land h_k(\x_i^k)])+1)] + 1
\end{multline*}

This is exactly the form of ${F''}_i$ shown in Equation~\ref{eq:endF_i}:
\begin{multline*}
{F''}_i\equiv\vec{Q}\x\neg\parity \x^1_i, y^1_i,\cdots \x^k_i \y^k_i [[(G_i(\x^1_i, y^1_i)\land \\
h_1(\x^1_i))+1] \times\cdots \times [(G_i(\x^k_i,\y^k_i)\land h_k(\x^k_i))+1]]+1
\end{multline*}

The transformation is complete. 

\section{Details in Refined Inner-Error Analysis}\label{sec:Appanalysis}

We give here the full proofs of the two inner-error allocations of Section~\ref{sec:innerComp}. In Section~\ref{sec:AppanalysisExact} we prove Theorem~\ref{clm:innerprobnotnaive}, which handles the case where every step of the reduction succeeds with probability \emph{exactly} $q$. In Section~\ref{sec:AppanalysisBounded} we prove Theorem~\ref{thm:boundedalloc}, which handles the case where every step is only guaranteed to succeed with probability \emph{at least} $q$, and which is the case that applies to all the sieves used in this paper. Throughout this section $q$ denotes the probability that a single step of the reduction results in a logically equivalent formula, so that $\epsilon_i = 1-q$, and $l=\lceil\frac{d}{2}\rceil$ denotes the number of $\forall\exists$ pairs.

\subsection{The Exact Case}\label{sec:AppanalysisExact}

We give a full proof of Theorem~\ref{clm:innerprobnotnaive}.

\begin{theoremStar}
 Let $q$ be such that $1-\epsilon\leq (q^{\left(2\left \lceil{\frac{d}{2}}\right \rceil \right)} + q)/(q+1)$. Then
 setting $\epsilon_i = 1-q$ for every $i$ ensures the overall error bound $\epsilon$ is achieved.
\end{theoremStar}

\begin{proof}

For every $i$, set the inner-error probability $\epsilon_i=1-q$ for some $0\leq q\leq 1$. We show that when $q$ is such that $1-\epsilon\leq (q^{\left(2\left \lceil{\frac{d}{2}}\right \rceil \right)} + q)/(q+1)$, then the overall bound $\epsilon$ is achieved.
For that, we make use of the fact that existential and universal quantifiers behave differently in terms of error propagation: (i) For existential quantifiers: If $F_{i+1}$ is \false\kern-0.35em, then $F_i$ remains \false with probability $1$. If $F_{i+1}$ is \true\kern-0.35em, then $F_i$ remains \true with a probability $q$, defined by the amplification in the reduction. (ii) For universal quantifiers, since $\forall x f(x) \equiv \neg\exists x \neg f(x)$, we get: If $F_{i+1}$ is \true\kern-0.35em, then $F_i$ remains \true with probability $1$. If $F_{i+1}$ is \false\kern-0.35em, then $F_i$ remains \false with probability $q$.
Our method to track down the error propagation is to group the quantifiers into pairs of $\forall\exists$ and treat each pair as a single step. For a single $\forall\exists$ pair, see that for every $i$:
$\Pr(Tr_i|Tr_{i+2}) = q + (1-q)^2$, $\Pr(Fl_i|Tr_{i+2}) = (1-q)q$,
$\Pr(Tr_i|Fl_{i+2}) = 1-q$, and $\Pr(Fl_i|Fl_{i+2}) = q$.

Assume for now that we have an even number of quantifiers, and that the first quantifier is universal. That means that we can view the sequence of quantifiers as pairs of  $\forall\exists$. Then for example we have that

\begin{align*}
\Pr(Tr_0|Tr_{2l}) = \Pr(Tr_0|Tr_2)\Pr(Tr_2|Tr_{2l})+
\Pr(Tr_0|Fl_{2})\Pr(Fl_{2}|Tr_{2l})
\end{align*}

Using this, we have for every $1\leq l\leq d/2$ (recall that $d$ is even) the following set of recurrence relations:

\begin{align*}
\Pr(Tr_0|Tr_{2l}) &= \Pr(Tr_2|Tr_{2l})(q + (1-q)^2) + \Pr(Fl_2|Tr_{2l})(1-q) \\
\Pr(Fl_0|Tr_{2l}) &= \Pr(Fl_2|Tr_{2l})q + \Pr(Tr_2|Tr_{2l})(1-q)q \\
\Pr(Tr_0|Fl_{2l}) &= \Pr(Tr_2|Fl_{2l})(q + (1-q)^2) + \Pr(Fl_2|Fl_{2l})(1-q) \\
\Pr(Fl_0|Fl_{2l}) &= \Pr(Fl_2|Fl_{2l})q + \Pr(Tr_2|Fl_{2l})(1-q)q
\end{align*}

The above set of recurrence relations can be expressed using the matrix multiplication as:
\begin{equation*}
\begin{pmatrix}
\Pr(Tr_0|Tr_{2l}) \\
\Pr(Fl_0|Tr_{2l}) \\
\Pr(Tr_0|Fl_{2l}) \\
\Pr(Fl_0|Fl_{2l})
\end{pmatrix}
= \mathbf{M}_{\text{pair}} 
\begin{pmatrix}
\Pr(Tr_2|Tr_{2l}) \\
\Pr(Fl_2|Tr_{2l}) \\
\Pr(Tr_2|Fl_{2l}) \\
\Pr(Fl_2|Fl_{2l})
\end{pmatrix}
\end{equation*}
Where
\begin{equation*}
\mathbf{M}_{\text{pair}} =
\begin{pmatrix}
q + (1-q)^2 & 1-q & 0 & 0 \\
q(1-q) & q & 0 & 0 \\
0 & 0 & q + (1-q)^2 & 1-q \\
0 & 0 & q(1-q) & q
\end{pmatrix}
\end{equation*}

To express the relation of $l$ pairs we can raise $\mathbf{M}_{\text{pair}}$ to the power of $l$ to get

matrix multiplication as:
\begin{equation*}
\begin{pmatrix}
\Pr(Tr_0|Tr_{2l}) \\
\Pr(Fl_0|Tr_{2l}) \\
\Pr(Tr_0|Fl_{2l}) \\
\Pr(Fl_0|Fl_{2l})
\end{pmatrix}
= (\mathbf{M}_{\text{pair}})^l 
\begin{pmatrix}
\Pr(Tr_{2l}|Tr_{2l}) \\
\Pr(Fl_{2l}|Tr_{2l}) \\
\Pr(Tr_{2l}|Fl_{2l}) \\
\Pr(Fl_{2l}|Fl_{2l})
\end{pmatrix}
= (\mathbf{M}_{\text{pair}})^l 
\begin{pmatrix}
1 \\
0 \\
0 \\
1
\end{pmatrix}
\end{equation*}

Next, we generalize this argument for the cases where our sequence of quantifiers is not strictly pairs of $\forall\exists$. This means that either the innermost quantifier is $\forall$ and/or the outermost  quantifier is $\exists$. We  address both cases by adding a term in the matrix multiplication equation.  Therefore our objective is to find the values of  $\vec{v}_{\text{final}}$ in the following computation:

\begin{equation}
\mathbf{M}_{\text{outer}} \cdot \mathbf{M}_{\text{pair}}^l \cdot \vec{v}_{\text{init}} = \vec{v}_{\text{final}}
\label{eq:matrix-main}
\end{equation}

Where:
\begin{itemize}
    \item $\vec{v}_{\text{final}}=(\Pr(Tr_0|Tr_{2l}), \Pr(Fl_0|Tr_{2l}),\Pr(Tr_0|Fl_{2l}), \Pr(Fl_0|Fl_{2l}))$ is the vector representing the final probabilities.
    \item $\vec{v}_{\text{init}}$ represents the initial probabilities - this is where we will account for the case where the innermost quantifier is $\forall$.
    \item $\mathbf{M}_{\text{pair}}$ encodes transitions for a single $\forall\exists$ pair.
    \item $\mathbf{M}_{\text{outer}}$ addresses the case where the outer most quantifier is $\exists$.
\end{itemize}

We first explain both possible cases of $\vec{v}_{\text{init}}$:

\begin{itemize}
\item In case that the inner-most quantifier is part of the sequence of the $\forall\exists$ pairs then:
\begin{align*}
    \vec{v}_{\text{init}} &= (\Pr(Tr_{2l}|Tr_{2l}), \Pr(Fl_{2l}|Tr_{2l}),\Pr(Tr_{2l}|Fl_{2l}), \Pr(Fl_{2l}|Fl_{2l})) \\
    &= (1,0,0,1)
\end{align*}

\item For the case where the innermost quantifier is $\forall$ , then $\vec{v}_{\text{init}}$ is of the form:
\begin{align*}
    \vec{v}_{\text{init}} &= (\Pr(Tr_{2l}|Tr_{2l+1}), \Pr(Fl_{2l}|Tr_{2l+1}),\Pr(Tr_{2l}|Fl_{2l+1}), \Pr(Fl_{2l}|Fl_{2l+1})) \\
    &= (1,0,1-q,q)
\end{align*}
$$$$

\end{itemize}

For the case  of $\mathbf{M}_{\text{outer}}$ we have to account for the outermost quantifier. 

\begin{itemize}
    \item In case that the outer most quantifier is $\forall$ we do not have to address any more changes to the probabilities, since the outermost quantifier is then a part of a$\forall\exists$ pair. Therefore so we define $\mathbf{M}_{\text{outer}} = \mathbf{I}_4$.
    \item For the case where the outer most quantifier is $\exists$ we have to account for the additional existential quantifier.
    Following the terminology of Section~\ref{sec:explToda}, assume that we reached Step $1$ and our formula is of the form $F_1 = \exists\x_1\parity \y_1 \hat{F_1}(\x_1,\y_1)$.
    If $F_1$ is \false, then the resulting formula $F_0$ stays \false with probability of $1$ since there are no satisfying assignments for $\x_1$ in $\hat{F_1}$. On the other hand, if $F_1$ is \true, then $F_0$ will be \true with probability $q$ and \false with probability $1-q$. This results in the encoding:
\begin{equation*}
\mathbf{M}_{\text{outer}} = \begin{pmatrix}
q & 0 & 0 & 0 \\
1-q & 1 & 0 & 0 \\
0 & 0 & q & 0 \\
0 & 0 & 1-q & 1
\end{pmatrix}
\end{equation*}
\end{itemize}

Finally, we compute the  Equation~\ref{eq:matrix-main} system for all cases. For brevity we report only for the case where we have only $\forall\exists$ pairs, other cases produce the same minimal bound as we report below.

Solving Equation~\ref{eq:matrix-main} system for the case where we have only $\forall\exists$ pairs yields:
\begin{equation*}
\begin{pmatrix}
\Pr(Tr_0|Tr_{2l}) \\
\Pr(Fl_0|Tr_{2l}) \\
\Pr(Tr_0|Fl_{2l}) \\
\Pr(Fl_0|Fl_{2l})
\end{pmatrix} =
\begin{pmatrix}
\frac{q^{2l+1} + 1}{q+1} \\
\frac{-q^{2l+1} + q}{q+1} \\
\frac{-q^{2l} + 1}{q+1} \\
\frac{q^{2l} + q}{q+1}
\end{pmatrix}
\end{equation*}

To ensure our desired error bound $\epsilon$, we require:
\begin{equation*}
1-\epsilon\leq \min(\Pr(Tr_0|Tr_{d}),\Pr(Fl_0|Fl_{d}))
\end{equation*}

All we end with: $1-\epsilon \leq \frac{q^{\left(2\left \lceil{\frac{d}{2}}\right \rceil \right)} + q}{q+1}$
\end{proof}

\subsection{The Bounded Case}\label{sec:AppanalysisBounded}

We now give a full proof of Theorem~\ref{thm:boundedalloc}. Here we no longer assume that a step succeeds with probability exactly $q$, but only that it succeeds with probability at least $q$. Accordingly, we let $q_{\exists,j},q_{\forall,j}\in[q,1]$ denote the actual success probabilities of the two steps of the $j$-th $\forall\exists$ pair, and we allow these to differ between steps.

\begin{theoremStar}
Assume that every step of the reduction succeeds with a probability of at least $q$, where the actual probability may be any value in $[q,1]$, independently at every step. Then setting $\epsilon_i = 1-q$ for every $i$, where $q = (1-\epsilon)^{1/\left\lceil \frac{d}{2}\right\rceil}$, ensures the overall error bound $\epsilon$ is achieved.
\end{theoremStar}

As in the exact case, the two quantifier types behave differently, and it is convenient to record their behavior in a matrix form. For a single step we write a $2\times 2$ \emph{transition matrix}, whose $(a,b)$ entry is the probability that the step outputs the truth value $b$ given that it received the truth value $a$, with rows and columns indexed by $(\true\kern-0.35em,\false\kern-0.35em)$. Following the description of the error propagation in Section~\ref{sec:AppanalysisExact}, these matrices are:

\begin{equation}\label{eq:step-matrices}
  \mathbf{M}_{\exists}(q_{\exists}) \;=\;
  \begin{pmatrix}
    q_{\exists} & 1 - q_{\exists} \\
    0 & 1
  \end{pmatrix},
  \qquad\qquad
  \mathbf{M}_{\forall}(q_{\forall}) \;=\;
  \begin{pmatrix}
    1 & 0 \\
    1 - q_{\forall} & q_{\forall}
  \end{pmatrix}.
\end{equation}

The zero entries in Equation~\ref{eq:step-matrices} are the structural property that drives the whole analysis. The zero in the bottom-left of $\mathbf{M}_{\exists}$ states that an existential step can never turn a \false formula into a \true one, since when there is no witness there is nothing for the sieve to isolate. The zero in the top-right of $\mathbf{M}_{\forall}$ states, symmetrically, that a universal step can never turn a \true formula into a \false one. Both hold for \emph{every} value of $q_{\exists}$ and $q_{\forall}$, including the value $1$. We note in addition that $\mathbf{M}_{\exists}(1) = \mathbf{M}_{\forall}(1) = \mathbf{I}_2$, that is, a step that succeeds with probability $1$ preserves the truth value with certainty, and is thus as if it were not present at all. This observation is what the adversary exploits below.

\begin{proof}
Since the given formula $F$ is either \true or \false\kern-0.35em, it suffices to bound both $\Pr(Fl_0|Tr_d)$ and $\Pr(Tr_0|Fl_d)$ by $\epsilon$. We assume throughout that the quantifier prefix is a sequence of $l$ pairs of $\forall\exists$, and we remove this assumption at the end of the proof. We begin by proving Claim~\ref{clm:pairbound}, which bounds the damage that a single pair can do.

\medskip\noindent\textbf{Proof of Claim~\ref{clm:pairbound}.}
Consider pair $j$, in which the $\exists$ step is applied first, as it eliminates the inner quantifier, and the $\forall$ step second. Suppose first that the formula entering the pair is \true\kern-0.35em. By Equation~\ref{eq:step-matrices}, the $\exists$ step preserves \true with probability $q_{\exists,j}$ and flips it to \false with probability $1-q_{\exists,j}$. In the first case the $\forall$ step preserves \true with certainty, by the structural zero of $\mathbf{M}_{\forall}$, and the pair outputs \true\kern-0.35em. In the second case the $\forall$ step preserves \false with probability $q_{\forall,j}$ and flips it back to \true with probability $1-q_{\forall,j}$. Hence the pair outputs \false only along the single path in which the $\exists$ step fails and the $\forall$ step succeeds:

\begin{equation}\label{eq:app-pair-true}
\Pr(Fl_{2j-2}|Tr_{2j}) \;=\; q_{\forall,j}(1-q_{\exists,j}),
\end{equation}
and correspondingly $\Pr(Tr_{2j-2}|Tr_{2j}) = 1-q_{\forall,j}(1-q_{\exists,j})$.
Since $q_{\forall,j}\leq 1$ and $q_{\exists,j}\geq q$, we obtain from Equation~\ref{eq:app-pair-true} that
$$\Pr(Fl_{2j-2}|Tr_{2j}) \;=\; q_{\forall,j}(1-q_{\exists,j}) \;\leq\; 1\cdot(1-q) \;=\; 1-q,$$
with equality exactly when $q_{\forall,j}=1$ and $q_{\exists,j}=q$, both of which are legal choices as $q\leq 1$.

Suppose now that the formula entering the pair is \false\kern-0.35em. The $\exists$ step preserves \false with certainty, by the structural zero of $\mathbf{M}_{\exists}$, and the $\forall$ step then flips it to \true with probability $1-q_{\forall,j}$. Hence
\begin{equation}\label{eq:app-pair-false}
\Pr(Tr_{2j-2}|Fl_{2j}) \;=\; 1-q_{\forall,j} \;\leq\; 1-q,
\end{equation}
using $q_{\forall,j}\geq q$, with equality when $q_{\forall,j}=q$. Setting in addition $q_{\exists,j}=1$ does not affect Equation~\ref{eq:app-pair-false}, and we shall see below why it is nevertheless the adversary's choice in this direction. This proves Claim~\ref{clm:pairbound}.

\medskip\noindent\textbf{The true-to-false direction.}
We show that $\Pr(Tr_0|Tr_{2l}) \geq q^{l}$, which is equivalent to the required bound $\Pr(Fl_0|Tr_d) \leq 1-q^{l}$.

The hash functions of the different steps are chosen independently of each other, and therefore, once the truth value of $F_{2j-2}$ is given, the outcome of the remaining pairs $j-1,\ldots,1$ does not depend on the truth values at the earlier levels. Conditioning on the truth value of $F_{2j-2}$, that is, on the output of pair $j$, we thus obtain for every $1\leq j\leq l$:

\begin{equation}\label{eq:app-T-rec}
\Pr(Tr_0|Tr_{2j}) \;=\; \Pr(Tr_0|Tr_{2j-2})\Pr(Tr_{2j-2}|Tr_{2j}) \;+\; \Pr(Tr_0|Fl_{2j-2})\Pr(Fl_{2j-2}|Tr_{2j}).
\end{equation}

The second term of Equation~\ref{eq:app-T-rec} is the rescue term, in which pair $j$ erroneously outputs \false but the remaining pairs flip the value back to \true\kern-0.35em. Being a product of two probabilities it is non-negative, and we may therefore discard it. We stress that this is the step that makes the argument uniform over the adversary's choices: the inequality that follows holds for \emph{every} choice of $q_{\exists,m},q_{\forall,m}\in[q,1]$ at \emph{every} pair $m$, and not only for one particular choice. Using Claim~\ref{clm:pairbound} on the surviving term we get, for every $1\leq j\leq l$:

\begin{align*}
\Pr(Tr_0|Tr_{2j}) &\;\geq\; \Pr(Tr_0|Tr_{2j-2})\Pr(Tr_{2j-2}|Tr_{2j}) \\
&\;=\; \Pr(Tr_0|Tr_{2j-2})\big[1-\Pr(Fl_{2j-2}|Tr_{2j})\big] \\
&\;\geq\; \Pr(Tr_0|Tr_{2j-2})\big[1-(1-q)\big] \;=\; q\cdot\Pr(Tr_0|Tr_{2j-2}).
\end{align*}

At level $0$ no pairs remain and the truth value passes through unchanged, so that $\Pr(Tr_0|Tr_0)=1$. Applying the above inequality successively for $j=1,2,\ldots,l$ therefore yields, by induction on $j$, that $\Pr(Tr_0|Tr_{2j})\geq q^{j}$, and in particular $\Pr(Tr_0|Tr_{2l})\geq q^{l}$ as required.

\medskip\noindent\textbf{Tightness of the bound.}
The bound $1-q^{l}$ is attained, and is therefore the exact worst case rather than merely an upper estimate. Let the adversary set $q_{\forall,j}=1$ and $q_{\exists,j}=q$ at every pair, which is the choice of Claim~\ref{clm:pairbound}. Then $\mathbf{M}_{\forall}(1) = \mathbf{I}_2$, so that every universal step preserves the truth value it receives with certainty and is as if it were not performed at all. Consequently a \false formula can never become \true under this choice: the universal steps preserve it, and the existential steps cannot turn \false into \true to begin with, by the structural zero of $\mathbf{M}_{\exists}$. Hence $\Pr(Tr_0|Fl_{2j-2}) = 0$ for every $j$, so that the rescue term of Equation~\ref{eq:app-T-rec} vanishes and discarding it above was lossless. What remains is a chain of $l$ existential steps, each with matrix $\mathbf{M}_{\exists}(q)$. By the structural zero of $\mathbf{M}_{\exists}$ the value \false is absorbing along such a chain, so that once the truth value has been lost it can never be recovered, and surviving as \true requires all $l$ steps to succeed. Equation~\ref{eq:app-T-rec} then reads $\Pr(Tr_0|Tr_{2j}) = q\cdot\Pr(Tr_0|Tr_{2j-2})$, giving $\Pr(Tr_0|Tr_{2l}) = q^{l}$ exactly.

\medskip\noindent\textbf{The false-to-true direction.}
This direction is symmetric, with the roles of $\exists$ and $\forall$ exchanged, and we show that $\Pr(Fl_0|Fl_{2l})\geq q^{l}$. Conditioning as before on the truth value of $F_{2j-2}$ gives, for every $1\leq j\leq l$:

\begin{equation}\label{eq:app-F-rec}
\Pr(Fl_0|Fl_{2j}) \;=\; \Pr(Fl_0|Fl_{2j-2})\Pr(Fl_{2j-2}|Fl_{2j}) \;+\; \Pr(Fl_0|Tr_{2j-2})\Pr(Tr_{2j-2}|Fl_{2j}).
\end{equation}

Discarding the second term, which is again non-negative, and applying the second part of Claim~\ref{clm:pairbound} to the first, we get $\Pr(Fl_0|Fl_{2j}) \geq \Pr(Fl_0|Fl_{2j-2})\big[1-(1-q)\big] = q\cdot\Pr(Fl_0|Fl_{2j-2})$. Since $\Pr(Fl_0|Fl_0)=1$, induction on $j$ gives $\Pr(Fl_0|Fl_{2l}) \geq q^{l}$, that is, $\Pr(Tr_0|Fl_d)\leq 1-q^{l}$.

The bound is again attained, this time by the adversary that sets $q_{\exists,j}=1$ and $q_{\forall,j}=q$ at every pair. Here $\mathbf{M}_{\exists}(1)=\mathbf{I}_2$, so that the existential steps disappear and a chain of $l$ universal steps remains. By the structural zero of $\mathbf{M}_{\forall}$ the value \true is now the absorbing one, so that a \false formula that was erroneously turned \true by some universal step can never be turned back. This is the reason why $q_{\exists,j}=1$ is the adversary's choice in this direction: an existential step that succeeds with certainty is precisely one that never repairs such an error. Remaining \false therefore requires all $l$ universal steps to succeed, which happens with probability exactly $q^{l}$.

\medskip\noindent\textbf{Other quantifier patterns.}
It remains to remove the assumption that the prefix is a sequence of $\forall\exists$ pairs, that is, to treat the cases where the innermost quantifier is $\forall$ and\slash or the outermost quantifier is $\exists$. Rather than adding correction terms as we did in Section~\ref{sec:AppanalysisExact}, it is simpler here to observe that the two bounds above did not really use the pairing at all. Let $a$ and $b$ denote the number of existential and universal quantifiers in the prefix, respectively, so that $a+b=d$.

By the structural zero of $\mathbf{M}_{\forall}$, a universal step never turns a \true formula into a \false one. Hence, if the given formula is \true and \emph{all} of the $a$ existential steps succeed, then no step along the reduction ever turns \true into \false\kern-0.35em, and the output $F_0$ is \true\kern-0.35em. Since the steps are independent, this gives
$$\Pr(Tr_0|Tr_d) \;\geq\; \prod_{j=1}^{a} q_{\exists,j} \;\geq\; q^{a},$$
and symmetrically, using the structural zero of $\mathbf{M}_{\exists}$, we get $\Pr(Fl_0|Fl_d) \geq q^{b}$. The overall error is therefore at most $1-q^{\max(a,b)}$. As the prefix alternates, $a$ and $b$ differ by at most one, and thus $\max(a,b) = \lceil\frac{d}{2}\rceil = l$ for every quantifier pattern, which recovers the bound $1-q^{l}$ in all cases. The bounds remain attained, by the same adversaries as above.

\medskip\noindent\textbf{Conclusion.}
Both error directions are bounded by $1-q^{l}$, and therefore so is the overall error. Requiring $1-q^{l}\leq\epsilon$ gives $q\geq(1-\epsilon)^{1/l}$, and substituting $l=\lceil\frac{d}{2}\rceil$ completes the proof.
\end{proof}
\section{The Modular Addition Theorem}\label{sec:ModularAddition}

We give a detailed self-contained proof of Theorem~\ref{thm:main}, which stands at the heart of the modular addition method. Note that this proof is not essential for this paper.

We first fix some definitions and notation. For an integer $k>0$, let $[k]=\{0,\cdots,k-1\}$. 
We denote discrete random variables by capital letters, say $X$. We denote by $Pr(X=j)$ the probability that $j$ is assigned to $X$. Naturally, $Pr(X=j)\geq 0$ for every $j$. We say that $X=j$ has \emph{positive} probability if $Pr(X=j)>0$. We say that the range of $X$ is $[k]$ if $Pr(X=j)=0$ for every $j<0$ or $j\geq k$. We say that $X$ is over a \emph{bounded range} if the range of $X$ is $[k]$ for some $k>0$.



\maintheorem*

\subsection{Proof of Theorem~\ref{thm:main}}\label{sec:main_proof}

Our proof makes use of the following definition.
For a positive integer $k$, let $d_{k}$ be the first root of order $k$ of $1$. That is, 
$d_{k} = e^{\left(\frac{2\pi i}{k}\right)}$ where $i$ denotes the imaginary unit $\sqrt{-1}$.
We recall some of the standard properties of $d_k$.

\begin{claim}\label{claim:d}
    The following holds for every integers $m,\:l$.
    \begin{enumerate}
\item $d_{k}^m = d_{k}^{(m \bmod{k})}$
\item $d_{k}^{mk} = 1$
\item $d_k^l = 1$ if and only if $l \bmod{k} = 0$
\item 
\begin{equation}\label{eq:sum_0_or_k}
\sum_{m=0}^{k-1} d_{k}^{ml} = 
\begin{cases}
    0, \quad \text{For} \ l \bmod{k} \neq 0\\
    k, \quad \text{For} \ l \bmod{k} = 0\\
\end{cases}
\end{equation}
\end{enumerate}
\end{claim}


We now proceed to the proof itself. For simplicity denote $Pr(X \bmod{k}=j)$ by $p_j$.
We prove the following two consecutive lemmas.


\begin{lemma}\label{lem:prob_to_multinom}
The following holds for every $t\in[k]$:
\begin{equation*}\label{eq:prob_to_multinom_again}
\text{Pr}\left(\Sumnk{n}{k} = t\right) = (1/k)\sum_{m=0}^{k-1}d_{k}^{-mt}\left(\sum_{j=0}^{k-1} d_k^{mj}p_j\right)^n
\end{equation*}
\end{lemma}

\begin{lemma}\label{lem:multinom_limit}
The following holds:
\begin{equation*}
 \lim_{n\to\infty} (1/k)\sum_{m=0}^{k-1}d_{k}^{-mt}\left(\sum_{j=0}^{k-1} d_k^{mj}p_j\right)^n = 1/k
\end{equation*}
\end{lemma}

Together, these two lemmas give:
\begin{equation}\label{eq:strong}
    \lim_{n \to \infty} \text{Pr}\left(\Sumnk{n}{k} = t\right) = 1/k
\end{equation}

which completes the proof. \qed



We next provide proofs for the Lemmas.

\subsubsection{Proof of Lemma~\ref{lem:prob_to_multinom}}

Standard combinatorics and the multinomial theorem give us the following 


\begin{equation*}
    \text{Pr}\left(\Sumnk{n}{k} = t\right) = \sum_{\substack{l_0 +l_1 + ... + l_{k-1} = n \: \land \\ \sum_{j=0}^{k-1}jl_j \bmod{k} = t}} \binom{n}{l_0, l_1, \dots, l_{k-1}} \prod_{j=0}^{k-1} p_j^{l_j}
\end{equation*}

Note that we define here $0^0=1$ as standard in combinatorics.

To deal with this summation we sum over all tuples  for which $l_0 +l_1 + ... + l_{k-1} = n$, and filter out tuples for which $\sum_{j=0}^{k-1}jl_j \bmod{k} \neq t$. For that we can use the roots of order $k$ of $1$. Specifically, Equation~\ref{eq:sum_0_or_k} above gives us that

\begin{equation*}
    (1/k)\sum_{m=0}^{k-1}d_k^{m\left(-t+\sum_{j=0}^{k-1}jl_j\right)} = 
    \begin{cases}
        0, \quad \text{For} \ \sum_{j=0}^{k-1}jl_j \bmod{k} \neq t\\
        1, \quad \text{For} \ \sum_{j=0}^{k-1}jl_j \bmod{k} = t\\
    \end{cases}
\end{equation*}

Using this we get that:

\begin{align*}
    & \text{Pr}\left(\Sumnk{n}{k} = t\right) = &
    \nonumber \\
    & \sum_{l_0 +l_1 + ... + l_{k-1} = n}\binom{n}{l_0, l_1, \dots, l_{k-1}}\left(\prod_{j=0}^{k-1} p_j^{l_j}\right)(1/k)\sum_{m=0}^{k-1} d_k^{m\left(-t+\sum_{j=0}^{k-1}jl_j\right)}
\end{align*}

Rearranging the terms, by taking the sum out and switching the order of summation, we get:

\begin{align*}
& \text{Pr}\left(\Sumnk{n}{k} = t\right) = &
    \nonumber \\
& (1/k)\sum_{m=0}^{k-1}\sum_{l_0 +l_1 + ... + l_{k-1} = n} \binom{n}{l_0, l_1, \dots, l_{k-1}} d_k^{m\left(-t+\sum_{j=0}^{k-1}jl_j\right)} \prod_{j=0}^{k-1} p_j^{l_j}
 = & \nonumber \\
& (1/k)\sum_{m=0}^{k-1}d_{k}^{-mt}\sum_{l_0 +l_1 + ... + l_{k-1} = n} \binom{n}{l_0, l_1, \dots, l_{k-1}} d_k^{m\sum_{j=0}^{k-1}jl_j} \prod_{j=0}^{k-1} p_j^{l_j} = & \nonumber \\
& (1/k)\sum_{m=0}^{k-1}d_{k}^{-mt}\sum_{l_0 +l_1 + ... + l_{k-1} = n} \binom{n}{l_0, l_1, \dots, l_{k-1}}  \prod_{j=0}^{k-1} \left(d_k^{mj}\right)^{l_j} \prod_{j=0}^{k-1} p_j^{l_j} =
& \nonumber \\
& (1/k)\sum_{m=0}^{k-1}d_{k}^{-mt}\sum_{l_0 +l_1 + ... + l_{k-1} = n} \binom{n}{l_0, l_1, \dots, l_{k-1}}  \prod_{j=0}^{k-1} \left(d_k^{mj}p_j\right)^{l_j}
\end{align*}

Now we can safely apply the multinomial theorem to get:

\begin{equation*}
\text{Pr}\left(\Sumnk{n}{k} = t\right) = (1/k)\sum_{m=0}^{k-1}d_{k}^{-mt}\left(\sum_{j=0}^{k-1} d_k^{mj}p_j\right)^n
\end{equation*}

as required
\qed

\subsubsection{Proof of Lemma~\ref{lem:multinom_limit}}
We have to show that:
\begin{equation*}
 \lim_{n\to\infty} (1/k)\sum_{m=0}^{k-1}d_{k}^{-mt}\left(\sum_{j=0}^{k-1} d_k^{mj}p_j\right)^n = 1/k
\end{equation*}

Multiplying both sided by $k$ give us an equivalent expression we can prove instead:
\begin{equation*}
    \lim_{n\to\infty} \sum_{m=0}^{k-1}d_{k}^{-mt}\left(\sum_{j=0}^{k-1} d_k^{mj}p_j\right)^n = 1
\end{equation*}
Now we separate the first term $m=0$ from the rest and get:
\begin{equation*}
    \lim_{n\to\infty} d_{k}^{-0t}\left(\sum_{j=0}^{k-1} d_k^{0j}p_j\right)^n + \sum_{m=1}^{k-1}d_{k}^{-mt}\left(\sum_{j=0}^{k-1} d_k^{mj}p_j\right)^n = 1
\end{equation*}
However, for the first term as $d_{k}^{-0t} = d_k^{0j} = 1$ we get that:
$$d_{k}^{-0t}\left(\sum_{j=0}^{k-1} d_k^{0j}p_j\right)^n = \left(\sum_{j=0}^{k-1}p_j\right)^n = 1^n = 1$$
So we are left to show that:
\begin{equation*}
    \lim_{n\to\infty} \sum_{m=1}^{k-1}d_{k}^{-mt}\left(\sum_{j=0}^{k-1} d_k^{mj}p_j\right)^n = 0
\end{equation*}


To prove that, we first make use of the fact that the triangle inequality holds strictly when the vectors are either independent or pointing in the opposite direction. We show that since $p_0,p_1>0$ therefore $d_k^{mj}p_j$ for $j\in\{0,1\}$ are such pair of vectors.
To see that,  for $j=0$ we get that $d_k^{m0}p_0 = p_0$ is a positive real number with no imaginary part.
For $j=1$ we have $d_k^{m1}p_1 = d_k^mp_1=p_1{(e^{\frac{2\pi i}{k}})}^m = p_1(e^{\frac{2\pi mi}{k}})$. Note that for all $1 \leq m \leq k-1$ we have that $e^{\frac{2\pi mi}{k}}$ either has an imaginary part (and therefore is independent of the real number $p_0$) or is a negative real number (and so $\{p_0, p_1(e^{\frac{2\pi mi}{k}})\} $ point in opposite directions). In both cases, however, we have that the triangle inequality holds strictly, therefore: 

 \begin{equation*}
\left|\sum_{j=0}^{k-1} d_k^{mj}p_j\right| \overset{\text{Triangle inequality}}{<} \sum_{j=0}^{k-1}\left|d_k^{mj}p_j\right| 
\end{equation*}

Then, since $d_k^{mj}$ is a root of unity for every $m,j$, which means that $|d_k^{mj}|=1$, we have, using Cauchy--Schwarz, that

\begin{align*}
\left|\sum_{j=0}^{k-1} d_k^{mj}p_j\right| \overset{\text{Triangle inequality}}{<}& \sum_{j=0}^{k-1}\left|d_k^{mj}p_j\right| \overset{\text{Cauchy--Schwarz}}{\leq} \sum_{j=0}^{k-1}\left|d_k^{mj}\right|\left|p_j\right| \\
=\quad & \sum_{j=0}^{k-1}\left|p_j\right| = 1
\end{align*}

So essentially we showed that for every $m\in \{1,\cdots, k-1\}$

 \begin{equation}\label{eq:trinagle_strict}
     \left|\sum_{j=0}^{k-1} d_k^{mj}p_j\right| < 1
 \end{equation}

 From here we immediately get:
 \begin{equation*}
 \lim_{n\to\infty}\left(\sum_{j=0}^{k-1} d_k^{mj}p_j\right)^n = 0
 \end{equation*}

 This leads us to:
\begin{equation*}\label{eq:limit_to_0_final}
    \lim_{n\to\infty} \sum_{m=1}^{k-1}d_{k}^{-mt}\left(\sum_{j=0}^{k-1} d_k^{mj}p_j\right)^n = 0
\end{equation*}
As required.
\qed

\paragraph*{Note}
Notice that Equation~\ref{eq:trinagle_strict} implies: $\left|\sum_{m=1}^{k-1}d_{k}^{-mt}\left(\sum_{j=0}^{k-1} d_k^{mj}p_j\right)^n\right|$ is monotonically decreasing. Consequently, the sequence of Lemma~\ref{lem:multinom_limit} converges to $1/k$ in a monotonically decreasing way.

\section{Field Arithmetic Implementation}\label{app:field}

In this section, we provide a concrete description of the field $\mathbb{F} = \mathbb{F}_{2^m}$ used in the main text, often referred to as the Galois field $GF(2^m)$, and explain how the multilinear hash function of Definition~\ref{def:multilinearHash} is encoded as a Boolean formula.

The standard way to implement this field is to view its elements not just as binary vectors, but as \emph{polynomials with binary coefficients}.

Let $a$ be an element of the field, represented by the $m$-bit vector $(a_{m-1}, \dots, a_1, a_0)$. We map this to a polynomial $A(t)$ over the variable $t$:
$$ A(t) = a_{m-1}t^{m-1} + \dots + a_1 t + a_0 $$
where each coefficient $a_i \in \{0,1\}$.

\paragraph{Example ($m=4$):}
The bit-vector $[1,0,1,1]$ corresponds to the polynomial $t^3 + t + 1$.

\subsection{Field Operations}

\paragraph{Addition:}
To add two elements $a, b \in \mathbb{F}$, we simply add their corresponding polynomials $A(t) + B(t)$. Since the coefficients are in $\{0,1\}$ modulo 2, this corresponds exactly to the bitwise XOR of the vectors.

\paragraph{Multiplication:}
Multiplication is slightly more complex. To multiply $a$ and $b$:
\begin{enumerate}
    \item Compute the standard polynomial product $C'(t) = A(t) \cdot B(t)$. Note that the degree of $C'(t)$ can be up to $2m-2$, which is too large to fit in $m$ bits.
    \item To reduce this back to $m$ bits, we compute the remainder modulo a fixed \textbf{Irreducible Polynomial} $Q(t)$ of degree $m$.
    $$ C(t) = A(t) \cdot B(t) \pmod{Q(t)} $$
\end{enumerate}
The polynomial $Q(t)$ acts as the modulus of the field, in the same way that a prime number does in modular integer arithmetic. An irreducible polynomial is one that cannot be factored into polynomials of smaller degree.

\subsubsection{Concrete Example: Handling Higher Degrees}
Let $m=3$ and use the irreducible polynomial $Q(t) = t^3 + t + 1$.
To perform a multiplication we first establish the reduction rules. Since $t^3 + t + 1 = 0$ in this field, and coefficients are taken modulo 2, we may rearrange this to obtain $t^3 = t+1$. This is the base rule: whenever a $t^3$ is encountered, it is replaced by $t+1$.

\paragraph{Deriving Rules for Higher Powers}
Multiplication can produce terms of degree up to $2m-2$, for example $t^2 \cdot t^2 = t^4$. These are handled by iteratively multiplying the base rule by $t$. To reduce $t^4$ we multiply the base rule by $t$, obtaining $t^4 = t\cdot t^3 = t(t+1) = t^2+t$, so that any $t^4$ can be replaced by $t^2+t$. To reduce $t^5$ we multiply the rule for $t^4$ by $t$, obtaining $t^5 = t\cdot t^4 = t(t^2+t) = t^3+t^2$. This result contains a $t^3$ again, so applying the base rule once more gives $t^5 = (t+1)+t^2 = t^2+t+1$.

\paragraph{Full Multiplication Walkthrough}
Let us multiply two elements $a,b$ that produce a high-degree result: $a = b = [1,1,0]$, which represents $t^2 + t$.

We start by expanding the terms as in normal algebra, remembering that coefficients are modulo 2 ($1+1=0$).
\[
    a \cdot b = (t^2 + t)(t^2 + t) = t^4 + 2t^3 + t^2 = t^4 + t^2
\]

We now have the term $t^4$ in our result which exceeds our field size (degree must be $<3$). We use the rule we derived above ($t^4 \to t^2+t$).
\[
    a \cdot b = t^4 + t^2 = (t^2 + t) + t^2 = t
\]
The final result is the polynomial $t$, which corresponds to the bit vector $[0,1,0]$.

\subsection{Encoding the Multilinear Hash Function}\label{sec:multilinearTechnique}

We describe how the hash function $h_{a,v}(\x) = \left\langle \prod_{i=1}^{n} a_i^{x_i},\, v\right\rangle$ of Definition~\ref{def:multilinearHash} can be written as a Boolean formula of a reasonable size. Throughout, $m = \lceil\log_2 n\rceil + 3$ is the bit-length of the field elements, which by Theorem~\ref{thm:parity-isolation} gives a success probability of at least $7/16$.

The construction is a direct reading of the definition. Since each $x_i$ is a single bit, the $i$-th factor of the product is either the field element $a_i$, when $x_i$ is \true\kern-0.35em, or the multiplicative identity $1_\mathbb{F}$, when $x_i$ is \false\kern-0.35em. Here $1_\mathbb{F}$ is the field element whose $m$ bits are all $0$ except the last one. Note that the elements $a_1,\ldots,a_n$ are sampled once, when the hash function is drawn, and are therefore fixed \emph{constants} as far as the encoding is concerned. Only the $x_i$ are variables. Computing the product thus amounts to a sequential scan over the input bits, in which the running product is multiplied by the constant $a_i$ whenever $x_i$ is \true\kern-0.35em.

To construct this sequential product we define a sequence of auxiliary variables $u^{(0)}, u^{(1)}, \dots, u^{(n)}$. Note that each $u^{(k)}$ represents a field element, that is, a vector of $m$ auxiliary variables.

We start by setting $u^{(0)}$ to $1_\mathbb{F}$. Then inductively at each step $k$ we define $u^{(k+1)}=u^{(k)}$ when $x_{k+1}$ is \false\kern-0.35em, and $u^{(k+1)}=u^{(k)} \cdot a_{k+1}$ when $x_{k+1}$ is \true\kern-0.35em. This selection between $u^{(k)}$ and $u^{(k)} \cdot a_{k+1}$ can be done either by encoding a simple two-to-one multiplexer, or alternatively by using the field arithmetic to represent it as $u \oplus (x_{k+1} \land (u \cdot D_{k+1}))$, where $\cdot$ denotes the field multiplication and $D_{k+1} = a_{k+1} \oplus 1_\mathbb{F}$. To see that the latter is correct, note that when $x_{k+1}$ is \false the second term vanishes and we are left with $u$, while when $x_{k+1}$ is \true we get, using distributivity and $w\oplus w = 0$, that $u \oplus (u\cdot(a_{k+1}\oplus 1_\mathbb{F})) = u \oplus (u\cdot a_{k+1}) \oplus u = u\cdot a_{k+1}$, as required.
In the next section (Appendix~\ref{app:field:multiplication}) we explain how to construct the multiplication by a constant. For now, note however, that in the worst case such a multiplication uses $m(m-1)$ XOR gates. This means that each transition from $u^{(k)}$ to $u^{(k+1)}$ uses at most $m(m-1) + m + m = m(m+1)$ gates, the two additional $m$ terms being the bitwise AND with $x_{k+1}$ and the final bitwise XOR, and since we have $n$ such transitions we end with a formula of size at most $nm(m+1)$. The bound is attained only when every row of every multiplication matrix is all ones. For a uniformly random constant each row is a uniformly random vector in $\{0,1\}^m$, since the map from the constant to any fixed row of its matrix is a linear bijection, so a transition costs $m(m/2+1-2^{-m})$ gates in expectation, about half the bound.

At the end of this process we end with $u^{(n)}$, whose value is the product $\prod_{i=1}^{n} a_i^{x_i}$. The only thing left to take care of is the inner product between $u^{(n)}$ and $v$. Since $v$ is a constant of the draw, exactly like the $a_i$, this is the XOR of the bits $u^{(n)}_j$ with $v_j=1$, which costs at most $m-1$ gates. We end with a formula of size
\[
|h_{a,v}| \leq nm(m+1)+m-1,
\]
using $nm$ auxiliary variables for the running product. With $m = \lceil\log_2 n\rceil+3$ this is $O(n\log^2 n)$ gates and $O(n\log n)$ auxiliary variables, which is asymptotically smaller than the $n^2$ gates of the Valiant-Vazirani hash function. To obtain a CNF formula, as required by a model counter, one applies the Tseytin transformation~\cite{tseitin1983complexity}, which turns a formula with $g$ gates into a CNF formula with $O(g)$ clauses and $O(g)$ fresh variables. In our case every gate is a two-input XOR or an AND, which the transformation turns into four or three clauses respectively, so the CNF handed to the counter has $\Theta(n\log^2 n)$ clauses, roughly four times the gate count.

We remark that the same encoding applies verbatim to the powering construction $h_{s,y}(\x) = \langle s^{\x}, y\rangle$ discussed in Section~\ref{sec:multilinear-attribution}. There one first expands $s^{\x}$ using the laws of exponents, writing $\x$ as the integer $\sum_{i=0}^{n-1}x_i2^i$ and obtaining $s^{\x} = \prod_{i=0}^{n-1}(s^{2^i})^{x_i}$, so that the constants $a_i$ above are replaced by the derived constants $C_i = s^{2^i}$. The encoding, and hence the bound $nm(m+1)+m-1$, is unchanged. What differs is the field size, which there must satisfy $m>n$ rather than $m\approx\log_2 n$, and which is what makes that variant of size $\Theta(n^3)$.

\subsection{Encoding Multiplication as Boolean Logic}\label{app:field:multiplication}
The most complex operation in our construction is calculating the product $v = u \cdot C$, where $u$ is a variable vector (the current state) and $C$ is a fixed constant in the field.
To translate this into simple boolean gates (AND/XOR), we rely on the property of \textit{linearity}.

\subsubsection{Step 1: Decomposing the Input}
First, we must view the input element $u$ not as a single number, but as a sum of its component bits.
Let $u$ be represented by $m$ bits: $u_0, u_1, \dots, u_{m-1}$.
Mathematically, this means $u$ is the sum of basis vectors $e_0, \dots, e_{m-1}$, where $e_i$ represents the value $2^i$ (a 1 at position $i$ and 0 elsewhere).
\begin{equation}
    u = \sum_{i=0}^{m-1} u_i \cdot e_i
\end{equation}
Note: The coefficient $u_i$ is a scalar bit (either 0 or 1).

\subsubsection{Step 2: Distributing the Constant}
We want to compute $v = u \cdot C$. We substitute the expansion of $u$ into this equation:
\begin{equation}
    v = \left( \sum_{i=0}^{m-1} u_i \cdot e_i \right) \cdot C
\end{equation}
In a finite field, multiplication distributes over addition. This allows us to move the constant $C$ inside the summation:
\begin{equation}
    v = \sum_{i=0}^{m-1} u_i \cdot (e_i \cdot C)
\end{equation}

\subsubsection{Step 3: Pre-calculating Constants}
Consider the term $(e_i \cdot C)$. Since $e_i$ is a standard basis vector and $C$ is a fixed constant at this stage, the product $e_i \cdot C$ is itself a fixed field element, which can be computed in advance, once, when the formula is constructed. Denote this pre-calculated element by $P_i = e_i \cdot C$. The equation for the result then simplifies to
\begin{equation}
    v = \sum_{i=0}^{m-1} u_i \cdot P_i,
\end{equation}
that is, the result $v$ is the XOR of the pre-calculated elements $P_i$ over every index $i$ for which the input bit $u_i$ equals $1$. Note that each $P_i$ is a vector of $m$ bits. We next describe how these are combined bit by bit.

\subsubsection{Step 4: Generating the Boolean Gates}
Finally, we need to determine the value of a specific output bit $v_j$ (the $j$-th bit of the result).
The vector equation above splits into independent equations for each bit position. The $j$-th bit of $v$ is the sum of the $j$-th bits of the terms being added.
\begin{equation}
    v_j = \bigoplus_{i=0}^{m-1} (u_i \land P_{i,j})
\end{equation}
Here, $P_{i,j}$ is the $j$-th bit of our pre-calculated constant $P_i$. This is a constant 0 or 1.
\begin{itemize}
    \item If $P_{i,j} = 0$: The term becomes $(u_i \land 0) = 0$. It vanishes from the sum.
    \item If $P_{i,j} = 1$: The term becomes $(u_i \land 1) = u_i$.
\end{itemize}
Therefore, to compute the output bit $v_j$, we simply identify which input bits $u_i$ "influence" it (i.e., where $P_{i,j}=1$) and XOR them together:
\begin{equation}
    v_j = \bigoplus_{\{i \mid P_{i,j}=1\}} u_i
\end{equation}
This is the final Boolean formula: every bit of the output is just the XOR of a specific subset of input bits.
\section{Bounding the Gap for Independent Events}\label{app:union-gap}

Consider $n$ independent events, each with a failure probability $q$. The exact probability of succeeding in all of the events is $(1-q)^n$, while the union bound gives a lower estimate of $1 - nq$. If we require the exact success probability to be exactly $1-\epsilon$, that is $(1-q)^n = 1 - \epsilon$, then $q = 1 - (1-\epsilon)^{1/n}$. For this exact $q$ the union bound estimate yields $1 - nq = 1 - n\left(1 - (1-\epsilon)^{1/n}\right)$.

We define by $\text{Gap}(n, \epsilon)$ the difference between the union bound and the exact probability for given $n,\epsilon$. Note that from above we get:
\[
\text{Gap}(n, \epsilon) = n\left(1 - (1-\epsilon)^{1/n}\right) - \epsilon.
\]
We will show shortly (in Appendix~\ref{app:proof-tight}) that this gap is maximized as $n \to \infty$, where it equals $-\ln(1-\epsilon) - \epsilon$. Using the Taylor expansion $-\ln(1-\epsilon) = \epsilon + \frac{\epsilon^2}{2} + \frac{\epsilon^3}{3} + \cdots$, we obtain
\[
\sup_{\substack{n > 0 \\ \epsilon \in (0,0.5)}} \text{Gap}(n, \epsilon) = -\ln(1-\epsilon) - \epsilon = \frac{\epsilon^2}{2} + \frac{\epsilon^3}{3} + \cdots \leq \epsilon^2
\]

For example: for $\epsilon = 0.1$ we get that $\sup_{n > 0} \text{Gap}(n, 0.1)$ is at most $0.0054$.
It is clear that this gap is negligible and the union bound estimate is almost the exact probability.

\subsection{Proofs for Tightness of the Gap}\label{app:proof-tight}

We prove that for $n$ independent events each with failure probability $q = 1 - (1-\epsilon)^{1/n}$, the gap between the exact success probability and the union bound estimate, is maximized as $n \to \infty$, where it equals $-\ln(1-\epsilon) - \epsilon$. As discussed above, the gap is:

\[
\text{Gap}(n, \epsilon) = n\left(1 - (1-\epsilon)^{1/n}\right) - \epsilon,
\]

\begin{lemma}\label{lem:gap-increasing}
For fixed $\epsilon \in (0, 1/2)$, the function $\text{Gap}(n, \epsilon)$ is strictly increasing in $n$ for $n > 0$.
\end{lemma}

\begin{proof}
Let $L = \ln(1-\epsilon) < 0$. Then $(1-\epsilon)^{1/n} = e^{L/n}$, and $\text{Gap}(n, \epsilon) = n\left(1 - e^{L/n}\right) - \epsilon$. The derivative with respect to $n$ is
\[
\frac{\partial}{\partial n}\left[n\left(1 - e^{L/n}\right)\right] = \frac{Le^{L/n}}{n} - e^{L/n} + 1.
\]
This is positive if and only if $1 > e^{L/n}\left(1 - \frac{L}{n}\right)$. Let $v = -L/n > 0$ (positive since $L < 0$). The condition becomes $1 > e^{-v}(1 + v)$, or equivalently $e^v > 1 + v$. This holds for all $v > 0$ since $e^v = 1 + v + \frac{v^2}{2} + \cdots > 1 + v$. Therefore $\frac{\partial}{\partial n}\text{Gap}(n, \epsilon) > 0$ for all $n > 0$.
\end{proof}

\begin{lemma}\label{lem:gap-limit}
For fixed $\epsilon \in (0, 1/2)$, $\lim_{n \to \infty} \text{Gap}(n, \epsilon) = -\ln(1-\epsilon) - \epsilon$.
\end{lemma}

\begin{proof}
Let $L = \ln(1-\epsilon) < 0$ and $x = L/n$, so as $n \to \infty$, $x \to 0^-$. We have
\[
n\left(1 - e^{L/n}\right) = L \cdot \frac{1 - e^x}{x}.
\]
Applying L'H\^opital's rule:
\[
\lim_{x \to 0} \frac{1 - e^x}{x} = \lim_{x \to 0} \frac{-e^x}{1} = -1.
\]
Therefore $\lim_{n \to \infty} n\left(1 - e^{L/n}\right) = L \cdot (-1) = -L = -\ln(1-\epsilon)$, and
\[
\lim_{n \to \infty} \text{Gap}(n, \epsilon) = -\ln(1-\epsilon) - \epsilon. \qedhere
\]
\end{proof}

Combining Lemmas~\ref{lem:gap-increasing} and~\ref{lem:gap-limit}, we conclude that
\[
\sup_{n > 0} \text{Gap}(n, \epsilon) = \lim_{n \to \infty} \text{Gap}(n, \epsilon) = -\ln(1-\epsilon) - \epsilon.
\]

As required.
\section{Toda's Reduction Modulo \texorpdfstring{$q$}{q}: Full Proofs}\label{app:modp}

This appendix proves Theorem~\ref{thm:modp}. We first define the two quantifiers that the construction uses, then establish the algebraic properties that generalize Claim~\ref{clm:many}, then derive the combination of Equation~\ref{eq:modp-combine} from those properties, and finally carry out the step-by-step transformation that corresponds to Appendix~\ref{sec:detreduction}. As in Section~\ref{sec:modp}, $q$ denotes the modulus and $\pi$ the success probability of a sieve.

\paragraph{Attribution.}
As stated in Section~\ref{sec:modp}, none of the mathematical content of this appendix is new. The properties below are the formula-level version of the results of Beigel and Gill~\cite[Section 6]{BeigelGill1992}, stated there in terms of counting classes and nondeterministic machines. The correspondence is as follows: Property~\ref{modp:neg} is their Theorem 22. Property~\ref{modp:fermat} is their Theorem 23(i), proved there by the same appeal to Fermat's little theorem. The guarantee quantifier $\modqj{j}$ below is their Theorem 23(ii). Properties~\ref{modp:conj} and~\ref{modp:disj} are their Lemma 26(i),(iii),(iv), where $\mathrm{Mod}_q\mathrm{P}$ is shown closed under intersection, complement and union for prime $q$. That these closure properties yield $\PH\subseteq\mathrm{BP}\cdot\mathrm{Mod}_k\mathrm{P}$ is due to Toda and Ogiwara~\cite{TodaOgiwara1992} and independently to Tarui~\cite{Tarui1993}. Our contribution here is presentational: writing the argument over Boolean formulas so that the size of each step can be read off, which is what Section~\ref{sec:modp-size} does.

\subsection{Definitions}\label{app:modp-defs}

Throughout this appendix $q$ is a prime and $F(\x), G(\y)$ are Boolean formulas. Recall from Section~\ref{sec:prelim} the operators $F+G$, with $\#[F+G]=\#F+\#G$, and $F\times G$ for disjoint variable sets, with $\#[F\times G]=\#F\cdot\#G$. We write $F+c$ for $F+M_c$, where $M_c$ is a formula with exactly $c$ satisfying assignments built over the variables already present, as the formula $x_1\land\cdots\land x_m$ is used for $c=1$ in Section~\ref{sec:prelim}. The $+$ operator itself contributes the one fresh selector variable. For $a\geq 1$ we write $F^{a}$ for the $a$-fold product $F(\x^1)\times\cdots\times F(\x^a)$ over disjoint copies of $\x$, so that $\#[F^{a}]=(\#F)^{a}$.

\begin{definition}[The standard quantifier]\label{def:modp}
$\modq\x\,F(\x)$ is \true if $\#F(\x)\not\equiv 0 \pmod q$, and \false if $\#F(\x)\equiv 0\pmod q$.
\end{definition}

\begin{definition}[The guarantee quantifier]\label{def:modpj}
Let $j\in\{1,\ldots,q-1\}$. We say that $F$ \emph{satisfies the promise with index $j$} if $\#F \equiv j$ or $\#F\equiv 0 \pmod q$. For such $F$, $\modqj{j}\x\,F(\x)$ is \true if $\#F\equiv j \pmod q$ and \false if $\#F\equiv 0\pmod q$.
\end{definition}

When $q=2$ the only admissible index is $j=1$, every formula satisfies the promise, and both quantifiers coincide with $\parity$. This is the degeneracy that makes the parity case simpler, and everything below that is not immediate for $q>2$ is immediate for $q=2$ for exactly this reason.

\subsection{Properties of the modulo \texorpdfstring{$q$}{q} quantifier}\label{app:modp-props}

\begin{claim}\label{clm:modp-props}
Let $q$ be prime. The following hold.
\begin{enumerate}
\item\label{modp:neg} \textbf{Negation.} If $F$ satisfies the promise with index $j$, then $F+(q-j)$ satisfies the promise with index $q-j$, and
\[\neg\,\modqj{j}\x\,F(\x) \;=\; \modqj{q-j}\x\,\bigl(F(\x)+(q-j)\bigr).\]

\item\label{modp:nest} \textbf{Nesting.} If $F(\x,\y)$ is such that for every assignment $\vsigma$ to $\x$ the formula $F(\vsigma,\y)$ satisfies the promise with index $j$, then
\[\modq\x\,\modqj{j}\y\,F(\x,\y) \;=\; \modq(\x,\y)\,F(\x,\y).\]

\item\label{modp:conj} \textbf{Conjunction.} If $\x,\y$ are disjoint, $F$ satisfies the promise with index $j$ and $G$ with index $i$, then $F\times G$ satisfies the promise with index $ji \bmod q$, and
\[\modqj{j}\x\,F(\x)\;\land\;\modqj{i}\y\,G(\y) \;=\; \modqj{ji \bmod q}(\x,\y)\,\bigl[F(\x)\times G(\y)\bigr].\]

\item\label{modp:disj} \textbf{Disjunction.} Under the hypotheses of Property~\ref{modp:conj},
\[\modqj{j}\x\,F \;\lor\; \modqj{i}\y\,G \;=\; \modqj{K}(\x,\y)\,\bigl[(F+(q-j))\times(G+(q-i)) + K\bigr],\]
where $K = q - \bigl[(q-j)(q-i) \bmod q\bigr]$.

\item\label{modp:free} \textbf{Free variables.} For any $F(\x,\y)$ and $G(\y)$,
\[\bigl[\modq\x\,F(\x,\y)\bigr]\land G(\y) \;=\; \modq\x\,\bigl[F(\x,\y)\land G(\y)\bigr].\]

\item\label{modp:fermat} \textbf{Fermat.} For every $F$, the formula $F^{q-1}$ satisfies the promise with index $1$, and
\[\modq\x\,F(\x) \;=\; \modqj{1}\bar{\x}\,F(\x)^{q-1},\]
where $\bar{\x}=(\x^1,\ldots,\x^{q-1})$.
\end{enumerate}
\end{claim}

\begin{proof}
\textbf{Property~\ref{modp:neg}.} If $\#F\equiv 0$ then $\#[F+(q-j)] \equiv q-j$, and if $\#F\equiv j$ then $\#[F+(q-j)]\equiv j+(q-j)=q\equiv 0$. So $F+(q-j)$ takes only the residues $q-j$ and $0$, which is the promise with index $q-j$ (note $1\le q-j\le q-1$ since $1\le j\le q-1$). Moreover $\neg\modqj{j}\x F$ is \true exactly when $\#F\equiv 0$, which is exactly when $\#[F+(q-j)]\equiv q-j$, which is when $\modqj{q-j}\x(F+(q-j))$ is \true.

\medskip\noindent
\textbf{Property~\ref{modp:nest}.} Let $A=\{\vsigma : \modqj{j}\y\,F(\vsigma,\y)=\mathit{True}\}$. By hypothesis every $\vsigma\in A$ has $\#F(\vsigma,\y)\equiv j$ and every $\vsigma\notin A$ has $\#F(\vsigma,\y)\equiv 0$. Summing over all $\vsigma$,
\[\#F(\x,\y) \;=\; \sum_{\vsigma}\#F(\vsigma,\y) \;\equiv\; |A|\cdot j \pmod q.\]
Since $q$ is prime and $1\le j\le q-1$, $j$ is invertible modulo $q$, so $|A|\cdot j\not\equiv 0$ if and only if $|A|\not\equiv 0$. The left-hand side is \true iff $|A|\not\equiv 0 \pmod q$ and the right-hand side is \true iff $\#F(\x,\y)\not\equiv 0 \pmod q$. By the displayed equivalence these coincide. \emph{This is the only place where primality of $q$ is used beyond Property~\ref{modp:fermat}.}

\medskip\noindent
\textbf{Property~\ref{modp:conj}.} Since $\x$ and $\y$ are disjoint, $\#[F\times G]=\#F\cdot\#G$. By the promises, $\#F\in\{0,j\}$ and $\#G\in\{0,i\}$ modulo $q$, so $\#F\cdot\#G \in \{0, ji\}$ modulo $q$, which is the promise with index $ji \bmod q$, and $ji\not\equiv 0$ because $q$ is prime and neither factor is $0$. The product is $\equiv ji$ exactly when both factors are non-zero, that is, exactly when both quantified statements are \true.

\medskip\noindent
\textbf{Property~\ref{modp:disj}.} By De Morgan, $a\lor b = \neg(\neg a\land \neg b)$. Applying Property~\ref{modp:neg} to each disjunct gives $\neg\modqj{j}\x F = \modqj{q-j}\x(F+(q-j))$ and likewise for $G$. Applying Property~\ref{modp:conj} to the conjunction gives $\modqj{(q-j)(q-i)\bmod q}(\x,\y)[(F+(q-j))\times(G+(q-i))]$. Applying Property~\ref{modp:neg} once more, with index $(q-j)(q-i)\bmod q$, produces the index $K=q-[(q-j)(q-i)\bmod q]$ and the additive term $+K$, as stated.

\medskip\noindent
\textbf{Property~\ref{modp:free}.} Fix an assignment $\vsigma_{\y}$. If $G(\vsigma_{\y})$ is \false then $F(\x,\vsigma_{\y})\land G(\vsigma_{\y})$ has no satisfying assignment, so its count is $0$ and the right-hand side is \false. The left-hand side is \false as well since its second conjunct is. If $G(\vsigma_{\y})$ is \true then $F(\x,\vsigma_{\y})\land G(\vsigma_{\y})$ has exactly the satisfying assignments of $F(\x,\vsigma_{\y})$, so the counts agree and both sides have the same truth value.

\medskip\noindent
\textbf{Property~\ref{modp:fermat}.} We have $\#[F^{q-1}]=(\#F)^{q-1}$. If $\#F\equiv 0 \pmod q$ then $(\#F)^{q-1}\equiv 0$. If $\#F\not\equiv 0\pmod q$ then $(\#F)^{q-1}\equiv 1 \pmod q$ by Fermat's little theorem, which applies precisely because $q$ is prime. Hence $\#[F^{q-1}]\in\{0,1\}$ modulo $q$ for \emph{every} $F$, with no hypothesis on $\#F$, which is the promise with index $1$. Moreover $\#[F^{q-1}]\equiv 1$ exactly when $\#F\not\equiv 0$, that is, exactly when $\modq\x F$ is \true.
\end{proof}

\subsection{The combination}\label{app:modp-combine}

Equation~\ref{eq:modp-combine} is basically what Property~\ref{modp:disj} becomes when it is applied to $k$ disjuncts rather than two. We state that generalization and derive it.

\begin{lemma}\label{lem:modp-combine}
Let $q$ be prime and let $B_1,\ldots,B_k$ be Boolean formulas over pairwise disjoint variable sets $\z_1,\ldots,\z_k$, each satisfying the promise with index $1$. Put
\[
J_k \;=\; (q-1)^{k} \bmod q, \qquad K_k \;=\; q - J_k .
\]
Then $J_k, K_k \in\{1,q-1\}$, and
\[
\bigvee_{j=1}^{k}\ \modqj{1}\z_j\,B_j
\;=\;
\modqj{K_k}(\z_1,\ldots,\z_k,\cdot)\ \Bigl[\ \prod_{j=1}^{k}\bigl(B_j + (q-1)\bigr) \;+\; K_k\ \Bigr],
\]
where the product is taken with the $\times$ operator and $\cdot$ denotes the fresh variables introduced by the $+$ operations.
\end{lemma}

\begin{proof}
Since $q-1\equiv -1 \pmod q$ we have $J_k \equiv (-1)^{k}$, so $J_k = 1$ when $k$ is even and $J_k = q-1$ when $k$ is odd. In either case $J_k\not\equiv 0$, and $K_k = q-J_k$ lies in $\{1,q-1\}$ as well. In particular both are legitimate indices.

By De Morgan, $\bigvee_j a_j = \neg\bigwedge_j \neg a_j$. We apply the properties in that order.

First, each $B_j$ satisfies the promise with index $1$, so Property~\ref{modp:neg} with $j=1$ gives
\[
\neg\,\modqj{1}\z_j\,B_j \;=\; \modqj{q-1}\z_j\,\bigl(B_j+(q-1)\bigr),
\]
and $B_j+(q-1)$ satisfies the promise with index $q-1$.

Second, the variable sets $\z_1,\ldots,\z_k$ are pairwise disjoint and the fresh variables introduced by the $k$ additions are distinct, so Property~\ref{modp:conj} applies and may be iterated $k-1$ times. Each application multiplies the indices, so after all of them the index is $(q-1)^{k}\bmod q = J_k$ and the formula is the $\times$-product:
\[
\bigwedge_{j=1}^{k}\modqj{q-1}\z_j\bigl(B_j+(q-1)\bigr)
\;=\;
\modqj{J_k}(\z_1,\ldots,\z_k)\ \prod_{j=1}^{k}\bigl(B_j+(q-1)\bigr).
\]

Third, applying Property~\ref{modp:neg} once more, now with index $J_k$, negates this and yields index $q-J_k = K_k$ together with the additive term $+K_k$, which is the statement of the lemma.
\end{proof}

Concretely, the count of the resulting formula is $\equiv J_k + K_k = q \equiv 0$ when every $\#B_j\equiv 0$, and $\equiv 0 + K_k = K_k \not\equiv 0$ as soon as some $\#B_j\not\equiv 0$, which is the behavior claimed in Section~\ref{sec:modp-construction}.

\begin{corollary}\label{cor:modp-combine}
For arbitrary formulas $F_1,\ldots,F_k$ over pairwise disjoint variables,
\[
\bigvee_{j=1}^{k}\ \modq\z_j\,F_j
\;=\;
\modq(\cdot)\ \Bigl[\ \prod_{j=1}^{k}\bigl(F_j^{\,q-1} + (q-1)\bigr) \;+\; K_k\ \Bigr],
\]
which is Equation~\ref{eq:modp-combine}.
\end{corollary}

\begin{proof}
By Property~\ref{modp:fermat}, $\modq\z_j F_j = \modqj{1}\bar{\z}_j F_j^{\,q-1}$ and $F_j^{\,q-1}$ satisfies the promise with index $1$ unconditionally, for every $F_j$. Apply Lemma~\ref{lem:modp-combine} with $B_j = F_j^{\,q-1}$, and note that $\modqj{K_k}$ implies $\modq$ on formulas satisfying the promise with index $K_k$.
\end{proof}

Setting $q=2$ in Corollary~\ref{cor:modp-combine} gives $q-1=1$ and $K_k=1$, hence $\bigvee_j \parity\z_j F_j = \parity\,[\prod_j(F_j+1)+1]$, which is the combination used in Equation~\ref{eq:endF_i}. Corollary~\ref{cor:modp-combine} is therefore a strict generalization of the parity case.

\begin{remark}\label{rem:modp-onefermat}
Property~\ref{modp:fermat} is invoked once per repetition, and the formula $F_j^{\,q-1}$ that it produces is used twice: as the guarantee required by Property~\ref{modp:nest} in Step 3 below, and as the input $B_j$ to Lemma~\ref{lem:modp-combine} in Step 5.
\end{remark}

\subsection{The transformation, step by step}\label{app:modp-steps}

We now carry out one step of the reduction, in the notation of Section~\ref{sec:explToda}. At the start of step $i$ the formula is
\begin{equation}\label{eq:modp-startF_i}
F_i\equiv\vec{Q}\x\,\neg\exists\x_i\,\bigl[\modq\y_i\,G_i(\x,\x_i,\y_i)\bigr],
\end{equation}
where $\vec{Q}\x$ is the outer prefix, $\x$ the outer variables, and $G_i$ is quantifier-free. Our task is to eliminate $\exists\x_i$. Write $k=k_i$ for the number of repetitions.

\medskip\noindent\textbf{Step 0 (the sieve).}
Fix an assignment $\vsigma$ to the outer variables $\x$, and regard $P(\x_i) := \modq\y_i\,G_i(\vsigma,\x_i,\y_i)$ as a Boolean predicate of $\x_i$. We draw $k$ independent hash functions $h_1,\ldots,h_{k}$ over $\x_i$ and conjoin each with $P$. If $P$ has no satisfying assignment then neither does $P\land h_j$, for every $j$, since conjoining with $h_j$ only removes solutions. If $P$ does have one, then with probability at least $\pi$ the sieve leaves a number of solutions that is not divisible by $q$, so that $\modq\x_i\,[P\land h_j]$ is \true. Asking that at least one of the $k$ repetitions succeeds turns $\exists\x_i$ into a disjunction:
\begin{equation}\label{eq:modp-beginF_i}
F'_i\equiv\vec{Q}\x\,\neg\Bigl(\bigl[\modq\x_i\,(\,[\modq\y_i\,G_i(\x,\x_i,\y_i)]\land h_1(\x_i)\,)\bigr] \lor \cdots \lor \bigl[\modq\x_i\,(\,[\modq\y_i\,G_i]\land h_{k}(\x_i)\,)\bigr]\Bigr).
\end{equation}
This is the only step that is not an identity, and its error is analysed in Section~\ref{app:modp-thm}. Steps 1--6 below transform Equation~\ref{eq:modp-beginF_i} into a single $\modq$-quantified formula, and each of them is exact.

\medskip\noindent\textbf{Step 1 (Property~\ref{modp:free}).}
Each $h_j$ depends only on $\x_i$, which is free for the inner quantifier $\modq\y_i$. Property~\ref{modp:free} therefore lets us push it inside:
\[
[\modq\y_i\,G_i]\land h_j(\x_i) \;=\; \modq\y_i\,\bigl[G_i(\x,\x_i,\y_i)\land h_j(\x_i)\bigr].
\]
Write $F_j := G_i \land h_j$, a quantifier-free formula over $(\x,\x_i,\y_i)$. Each disjunct is now $\modq\x_i\,\modq\y_i\,F_j$.

\medskip\noindent\textbf{Step 2 (Property~\ref{modp:fermat}).}
Property~\ref{modp:nest} requires the inner quantifier to be in guarantee form, which for $q>2$ it is not. We convert it:
\[
\modq\y_i\,F_j \;=\; \modqj{1}\bar{\y}_i\,F_j^{\,q-1}, \qquad \bar{\y}_i = (\y_i^1,\ldots,\y_i^{q-1}).
\]
Write $B_j := F_j^{\,q-1} = (G_i\land h_j)^{q-1}$. By Property~\ref{modp:fermat} the promise with index $1$ holds for $B_j$ for \emph{every} assignment to $\x$ and $\x_i$, unconditionally, which is exactly the hypothesis Property~\ref{modp:nest} needs. This is the only step at which copies of the formula are created, and it creates $q-1$ of them.

\medskip\noindent\textbf{Step 3 (Property~\ref{modp:nest}).}
The two adjacent quantifiers now collapse:
\[
\modq\x_i\,\bigl[\modqj{1}\bar{\y}_i\,B_j\bigr] \;=\; \modq(\x_i,\bar{\y}_i)\,B_j.
\]

\medskip\noindent\textbf{Step 4 (renaming).}
Lemma~\ref{lem:modp-combine} requires disjoint variable sets. As in the parity case, we give each repetition its own fresh copies, replacing $\x_i$ by $\x_i^j$ and $\bar{\y}_i$ by $\bar{\y}_i^j$. The outer variables $\x$ are shared and are not renamed. Write $\z_j=(\x_i^j,\bar{\y}_i^j)$.

\medskip\noindent\textbf{Step 5 (Lemma~\ref{lem:modp-combine}).}
Each $B_j$ satisfies the promise with index $1$ by Step 2, so the lemma applies directly (no second use of Property~\ref{modp:fermat} is needed, by Remark~\ref{rem:modp-onefermat}), and the disjunction becomes a single quantifier:
\[
\bigvee_{j=1}^{k}\modq\z_j\,B_j \;=\; \modqj{K_k}(\z_1,\ldots,\z_{k},\cdot)\ \Bigl[\ \prod_{j=1}^{k}\bigl(B_j + (q-1)\bigr) \;+\; K_k\ \Bigr].
\]
Substituting $B_j=(G_i\land h_j)^{q-1}$ recovers Equation~\ref{eq:modp-combine} verbatim. Denote the bracketed formula by $C$. By Lemma~\ref{lem:modp-combine} it satisfies the promise with index $K_k$.

\medskip\noindent\textbf{Step 6 (Property~\ref{modp:neg}).}
The leading negation of Equation~\ref{eq:modp-beginF_i} is absorbed by Property~\ref{modp:neg} with index $K_k$:
\[
\neg\,\modqj{K_k}\z\,C \;=\; \modqj{q-K_k}\z\,\bigl(C+(q-K_k)\bigr),
\]
and $q-K_k = J_k$. Since $\modqj{J_k}$ implies $\modq$ on formulas satisfying the promise, we obtain
\begin{equation}\label{eq:modp-endF_i}
F_{i-1}\equiv\vec{Q}\x\ \modq\,\bigl(\x_i^1,\bar{\y}_i^1,\ldots,\x_i^{k},\bar{\y}_i^{k},\cdot\bigr)\ \Bigl[\ \prod_{j=1}^{k}\bigl((G_i\land h_j)^{q-1} + (q-1)\bigr) + K_k + J_k\ \Bigr],
\end{equation}
which is again of the form required at step $i-1$, with the $\modq$-quantified variables of Equation~\ref{eq:modp-endF_i} playing the role of $\y_{i-1}$. Note that $K_k+J_k = q \equiv 0$, so as in the parity case, where the two consecutive $+1$'s cancel, the two additive constants may simply be dropped in an implementation. We keep them here to stay faithful to the derivation. As in the parity case, when the quantifier being eliminated is the outermost one and is existential, the leading negation is absent and Step 6 is omitted.

\subsection{Proof of Theorem~\ref{thm:modp}}\label{app:modp-thm}

\begin{proof}
Steps 1--6 above are exact identities: each invokes one of the properties of Claim~\ref{clm:modp-props} or Lemma~\ref{lem:modp-combine}, whose hypotheses are met at the point of use: Property~\ref{modp:free} because $h_j$ depends only on $\x_i$, Property~\ref{modp:nest} because Step 2 established the promise, Lemma~\ref{lem:modp-combine} because Step 4 made the variable sets disjoint and Step 2 established the promise with index $1$, and Property~\ref{modp:neg} because Lemma~\ref{lem:modp-combine} established the promise with index $K_k$. Note that none of these identities is probabilistic. Consequently the transformation from $F'_i$ to $F_{i-1}$ introduces no error at all, and the entire error of the step is the error of passing from $F_i$ to $F'_i$, that is, of the sieve.

For Step 0, fix an assignment $\vsigma$ to the outer variables $\x$. Two cases.

If $\modq\y_i\,G_i(\vsigma,\x_i,\y_i)$ has no satisfying assignment over $\x_i$, then for every $j$ the formula $F_j=G_i\land h_j$ has none either. Hence $\#F_j=0$ exactly, so $\#B_j=0$, every factor of the product in Step 5 is congruent to $q-1$, the product is congruent to $J_k$, and $\#C\equiv J_k+K_k\equiv 0$. The step therefore reports \false, with probability $1$: the reduction has one-sided error, exactly as in the parity case.

If on the other hand there is at least one such satisfying assignment, then by the guarantee of the sieve each repetition independently achieves $\#F_j\not\equiv 0\pmod q$ with probability at least $\pi$. The step reports \true unless all $k_i$ repetitions fail, which happens with probability at most $(1-\pi)^{k_i}$. Taking a union bound over the $2^{s}$ assignments to the $s$ outer variables, the step is correct simultaneously for all of them except with probability at most $2^{s}(1-\pi)^{k_i}$, so choosing
\[
k_i \;\geq\; \log_{1-\pi}\!\left(\frac{\epsilon_i}{2^{s}}\right)
\]
bounds the error of step $i$ by $\epsilon_i$. This is verbatim the computation of Section~\ref{sec:numcomp}, with the parity sieve's success probability replaced by $\pi$. Nothing in that computation depends on the modulus.

Finally, allocating the inner errors $\epsilon_1,\ldots,\epsilon_d$ as in Section~\ref{sec:innerComp} so that they sum to at most $\epsilon$, and applying the above to each of the $d$ steps, gives a quantifier-free $F_0$ that is correct with probability at least $1-\epsilon$. Unwinding the definition of $\modq$, $F$ is \true if and only if $\#F_0\not\equiv 0\pmod q$, as claimed.
\end{proof}

\section{The Ceiling Modulo \texorpdfstring{$q$}{q}: Full Proof}\label{app:modq-barrier}

This appendix proves Theorem~\ref{thm:modq-barrier}. The argument is an averaging one, as for Theorem~\ref{thm:half-barrier} at $q=2$: a minimum is at most a weighted average, so to bound every sieve at once it is enough to exhibit a few families of solution sets, and weights on them, against which no sieve can do well on average.

We keep the notation of Definition~\ref{def:sieve-dist}: $U$ is the universe, $N=|U|$, and a sieve is an arbitrary distribution $\mathcal{H}$ over the functions $h:U\to\{0,1\}$. It is convenient to name the set of elements that $h$ selects, $T_h=\{\x\in U: h(\x)=1\}$, together with its size $K_h=|T_h|$ and the number $c_h(S)=|S\cap T_h|$ of elements of $S$ that $h$ selects. We say that $h$ \emph{succeeds} on $S$ when $c_h(S)\not\equiv 0 \pmod q$.

The success probability $P^{q}_{\mathcal{H}}$ of Section~\ref{sec:modq-vs-mod2} depends on a formula only through its solution set, so we index it by sets throughout: for nonempty $S\subseteq U$,
\[
P^{q}_{\mathcal{H}}(S)\;=\;\Pr_{h\sim\mathcal{H}}\bigl[\,h \text{ succeeds on } S\,\bigr],
\]
and we write $\avg_{|S|=s}P^{q}_{\mathcal{H}}(S)$ for the plain average of this over all $\binom{N}{s}$ sets of size $s$. Every nonempty $S$ is the solution set of some satisfiable formula, so Theorem~\ref{thm:modq-barrier} says exactly that for every $\mathcal{H}$ there is a nonempty $S$ with $P^{q}_{\mathcal{H}}(S)\leq c_q$.

\subsection{The method}\label{app:modq-duality}

Two lemmas reduce the theorem to a finite calculation. The first replaces the minimum over the $2^{N}-1$ nonempty sets by an average over a handful of chosen sizes.

\begin{lemma}\label{lem:modq-weakdual}
Let $s_1,\ldots,s_r$ be distinct sizes in $\{1,\ldots,N\}$ and let $\lambda_{1},\ldots,\lambda_{r}$ be non-negative reals summing to $1$. Then for every sieve $\mathcal{H}$,
\[
\min_{\emptyset\neq S\subseteq U} P^{q}_{\mathcal{H}}(S)
\;\leq\;
\sum_{i=1}^{r}\lambda_{i}\,\avg_{|S|=s_i}P^{q}_{\mathcal{H}}(S).
\]
\end{lemma}

\begin{proof}
Each average on the right is taken over a nonempty family of nonempty sets, so each is at least the minimum on the left, and a convex combination of quantities that are all at least the minimum is itself at least the minimum.
\end{proof}

This is weak duality in its simplest form. The second lemma is what makes the right-hand side computable: for a fixed $h$, the average over a whole size-family collapses to a single number determined by $K_h$.

\begin{lemma}\label{lem:modq-perhash}
Fix a size $s$ and a hash function $h$, and write $K=K_h$. The fraction of the $\binom{N}{s}$ sets of size $s$ on which $h$ succeeds depends on $h$ only through $K$, and equals
\begin{equation}\label{eq:modq-as}
a_s(K)\;=\;\frac{1}{\binom{N}{s}}
\sum_{\substack{0\leq j\leq s\\ j\not\equiv 0 \pmod q}}
\binom{K}{j}\binom{N-K}{s-j}.
\end{equation}
Consequently $\avg_{|S|=s}P^{q}_{\mathcal{H}}(S)=\mathbb{E}_{h}\bigl[a_s(K_h)\bigr]$ for every sieve $\mathcal{H}$.
\end{lemma}

\begin{proof}
Success on $S$ is decided by the single number $c_h(S)$, so we sort the sets of size $s$ by its value. A set with $c_h(S)=j$ is built by choosing which $j$ of its elements come from the $K$ that $h$ selects and which $s-j$ come from the $N-K$ that it does not. The two choices are independent, so there are $\binom{K}{j}\binom{N-K}{s-j}$ such sets, and each set of size $s$ is counted once, for one value of $j$. Summing over the successful values $j\not\equiv 0\pmod q$ and dividing by the total $\binom{N}{s}$ gives \eqref{eq:modq-as}, whose right-hand side mentions $h$ only through $K$.

For the second claim, write $\mathbf{1}_h(S)$ for the indicator that $h$ succeeds on $S$. Both averages are finite sums, so they may be exchanged:
\[
\avg_{|S|=s}P^{q}_{\mathcal{H}}(S)
\;=\;\frac{1}{\binom{N}{s}}\sum_{|S|=s}\mathbb{E}_{h}\bigl[\mathbf{1}_h(S)\bigr]
\;=\;\mathbb{E}_{h}\biggl[\frac{1}{\binom{N}{s}}\sum_{|S|=s}\mathbf{1}_h(S)\biggr]
\;=\;\mathbb{E}_{h}\bigl[a_s(K_h)\bigr],
\]
since for each fixed $h$ the inner average is the fraction just computed.
\end{proof}

Substituting the second lemma into the first and pulling the finite sum inside the expectation gives, for any choice of sizes and weights,
\begin{equation}\label{eq:modq-master}
\min_{\emptyset\neq S\subseteq U} P^{q}_{\mathcal{H}}(S)
\;\leq\;\mathbb{E}_{h}\bigl[F(K_h)\bigr]
\;\leq\;\max_{0\leq K\leq N}F(K),
\qquad
F(K)\;=\;\sum_{i=1}^{r}\lambda_{i}\,a_{s_i}(K),
\end{equation}
the last step because the expectation of a function of any random variable is at most its maximum. The right-hand side involves $\mathcal{H}$ nowhere: bounding every sieve at once has been reduced to filling in a table of $N+1$ numbers. Everything that follows is a choice of sizes and weights, and then that table.

\subsection{The bound at \texorpdfstring{$N=q+1$}{N=q+1}}\label{app:modq-main}

Take $N=q+1$ and the three sizes $1$, $q$ and $q+1$. There are $q+1$ singletons, $q+1$ sets of size $q$ (the complements of the singletons), and one set of size $q+1$, namely $U$ itself. On these three families $a_s(K)$ has a simple closed form.

\begin{lemma}\label{lem:modq-fractions}
For $N=q+1$ and $0\leq K\leq q+1$,
\[
a_1(K)=\frac{K}{q+1},
\qquad
a_q(K)=
\begin{cases}
0, & K\in\{0,\,q+1\},\\[4pt]
\dfrac{q}{q+1}, & K\in\{1,\,q\},\\[4pt]
1, & 2\leq K\leq q-1,
\end{cases}
\qquad
a_{q+1}(K)=\mathbf{1}\bigl[K\notin\{0,q\}\bigr].
\]
\end{lemma}

\begin{proof}
\emph{Singletons.} $c_h(\{\x\})$ is $1$ if $\x\in T_h$ and $0$ otherwise, and of these only $0$ is divisible by $q$. So exactly the $K$ singletons inside $T_h$ succeed, out of $q+1$.

\emph{The full set.} $c_h(U)=K$, and among $0,1,\ldots,q+1$ the multiples of $q$ are $0$ and $q$ alone, since $q+1\equiv 1$. So $U$ succeeds exactly when $K\notin\{0,q\}$.

\emph{Sets of size $q$.} Because $|U|=q+1$, such a set is $S_{\y}=U\setminus\{\y\}$ for a unique $\y$, and dropping $\y$ removes exactly its own contribution to the count:
\[
c_h(S_{\y})\;=\;K-\mathbf{1}[\y\in T_h]\;=\;
\begin{cases}
K-1, & \text{for the } K \text{ choices of } \y\in T_h,\\
K, & \text{for the } q+1-K \text{ choices of } \y\notin T_h.
\end{cases}
\]
This count lies in $\{0,\ldots,q\}$, so it is divisible by $q$ exactly when it is $0$ or $q$. We count the failing $\y$ among the $q+1$. For $K=0$ every $\y$ gives count $0$, so all fail. For $K=1$ only the selected point gives $0$, so one fails. For $2\leq K\leq q-1$ both $K-1$ and $K$ lie in $\{1,\ldots,q-1\}$, so none fails. For $K=q$ only the unselected point gives count $q$, so one fails. For $K=q+1$ every $\y$ gives count $q$, so all fail. Dividing by $q+1$ gives the stated fractions.
\end{proof}

\begin{proposition}\label{prop:modq-Nq1}
Let $q\geq 4$ and $N=q+1$. For every sieve $\mathcal{H}$,
\[
\min_{\emptyset\neq S\subseteq U} P^{q}_{\mathcal{H}}(S)\;\leq\;\frac{q(q+3)}{q^{2}+5q+2}\;<\;\frac{q-1}{q}.
\]
\end{proposition}

\begin{proof}
Put $D=q^{2}+5q+2$ and take
\[
\lambda_{1}=\frac{(q+1)^{2}}{D},
\qquad
\lambda_{q}=\frac{2(q+1)}{D},
\qquad
\lambda_{q+1}=\frac{q-1}{D},
\]
which are positive and sum to $1$, since $(q+1)^{2}+2(q+1)+(q-1)=D$. By \eqref{eq:modq-master} it suffices to bound $F(K)=\lambda_1 a_1(K)+\lambda_q a_q(K)+\lambda_{q+1}a_{q+1}(K)$.

Each term is a numerator over the common denominator $D$. The singleton term is $\frac{(q+1)^{2}}{D}\cdot\frac{K}{q+1}=\frac{K(q+1)}{D}$. The size-$q$ term is $\frac{2(q+1)}{D}a_q(K)$, equal to $0$, $\frac{2q}{D}$ or $\frac{2(q+1)}{D}$ in the three cases of Lemma~\ref{lem:modq-fractions}. The full-set term is $\frac{q-1}{D}$ when $K\notin\{0,q\}$ and $0$ otherwise. Collecting numerators:
\begin{center}
\begin{tabular}{c|cccc}
\hline
$K$ & singleton & size $q$ & full set & $D\cdot F(K)$ \\
\hline
$0$ & $0$ & $0$ & $0$ & $0$\\
$1$ & $q+1$ & $2q$ & $q-1$ & $4q$\\
$2\leq K\leq q-1$ & $K(q+1)$ & $2(q+1)$ & $q-1$ & $K(q+1)+3q+1$\\
$q$ & $q(q+1)$ & $2q$ & $0$ & $q^{2}+3q$\\
$q+1$ & $(q+1)^{2}$ & $0$ & $q-1$ & $q^{2}+3q$\\
\hline
\end{tabular}
\end{center}
The third row grows with $K$, and at its largest value $K=q-1$ it equals $(q-1)(q+1)+3q+1=q^{2}+3q$, matching the last two rows, and $4q\leq q^{2}+3q$ for every $q\geq 1$. Hence $\max_{K}F(K)=\frac{q^{2}+3q}{D}=\frac{q(q+3)}{D}$.

For the strict inequality,
\[
\frac{q(q+3)}{q^{2}+5q+2}<\frac{q-1}{q}
\;\iff\;
q^{2}(q+3)<(q-1)(q^{2}+5q+2)
\;\iff\;
q^{2}-3q-2>0,
\]
which holds for $q\geq 4$.
\end{proof}

\subsection{From \texorpdfstring{$N=q+1$}{N=q+1} to every larger universe}\label{app:modq-monotone}

The universes we care about have $N=2^{n}$, far more than $q+1$ elements. Passing to them costs nothing, because an adversary on a large universe can confine itself to a small part of it.

\begin{lemma}\label{lem:modq-monotone}
Let $\mathcal{H}'$ be a sieve on a universe $U'$ and let $U\subseteq U'$ be nonempty. Let $\mathcal{H}$ be the law of the restriction $h|_{U}$ when $h\sim\mathcal{H}'$. Then
\[
\min_{\emptyset\neq S\subseteq U'} P^{q}_{\mathcal{H}'}(S)
\;\leq\;
\min_{\emptyset\neq S\subseteq U} P^{q}_{\mathcal{H}}(S).
\]
\end{lemma}

\begin{proof}
The restriction is defined for any $\mathcal{H}'$ whatsoever. A nonempty $S\subseteq U$ meets $U'\setminus U$ nowhere, so $c_h(S)$ is unchanged by restricting $h$ to $U$, and $P^{q}_{\mathcal{H}}(S)=P^{q}_{\mathcal{H}'}(S)$. The minimum on the right is therefore a minimum of $P^{q}_{\mathcal{H}'}$ over a subfamily of the sets that the left-hand minimum ranges over, and a minimum over fewer sets is at least as large.
\end{proof}

\begin{corollary}\label{cor:modq-allN}
Let $q\geq 4$ and let $\mathcal{H}'$ be a sieve on a universe $U'$ with $|U'|\geq q+1$. Then
\[
\min_{\emptyset\neq S\subseteq U'} P^{q}_{\mathcal{H}'}(S)\;\leq\;\frac{q(q+3)}{q^{2}+5q+2}.
\]
\end{corollary}

\begin{proof}
Fix $U\subseteq U'$ with $|U|=q+1$ and let $\mathcal{H}$ be the restriction of $\mathcal{H}'$ to $U$. Lemma~\ref{lem:modq-monotone} bounds the left-hand side by the minimum of $\mathcal{H}$ on $U$, and Proposition~\ref{prop:modq-Nq1} bounds that by the stated value.
\end{proof}

This is the case $q\geq 4$ of Theorem~\ref{thm:modq-barrier}. Nothing above used that $q$ is prime, so it covers composite moduli as well.

\subsection{The modulus three}\label{app:modq-three}

At $q=3$ the bound of Proposition~\ref{prop:modq-Nq1} reads $\frac{18}{26}=\frac{9}{13}>\frac23$, so four elements do not suffice, and we work on six instead. In the table below $a_s(K)$ is \eqref{eq:modq-as} with $q=3$ and $N=6$, and each row carries its denominator $\binom{N}{s}$, so the entries shown are counts of succeeding sets.

\begin{proposition}\label{prop:modq-six}
For $N=6$ and $q=3$, every sieve $\mathcal{H}$ has $\min_{\emptyset\neq S}P^{q}_{\mathcal{H}}(S)\leq\frac{27}{41}<\frac23$.
\end{proposition}

\begin{proof}
The table of $a_s(K)$ for $N=6$, each entry being the numerator $\sum_{j\not\equiv 0\,(3)}\binom{K}{j}\binom{6-K}{s-j}$ of \eqref{eq:modq-as}:
\begin{center}
\begin{tabular}{c|ccccccc}
\hline
 & $K{=}0$ & $1$ & $2$ & $3$ & $4$ & $5$ & $6$\\
\hline
$s=1$\quad$\bigl(\times\tfrac{1}{6}\bigr)$ & $0$&$1$&$2$&$3$&$4$&$5$&$6$\\
$s=2$\quad$\bigl(\times\tfrac{1}{15}\bigr)$ & $0$&$5$&$9$&$12$&$14$&$15$&$15$\\
$s=3$\quad$\bigl(\times\tfrac{1}{20}\bigr)$ & $0$&$10$&$16$&$18$&$16$&$10$&$0$\\
$s=4$\quad$\bigl(\times\tfrac{1}{15}\bigr)$ & $0$&$10$&$14$&$12$&$7$&$5$&$15$\\
$s=5$\quad$\bigl(\times\tfrac{1}{6}\bigr)$ & $0$&$5$&$6$&$3$&$2$&$6$&$6$\\
$s=6$\quad$\bigl(\times\tfrac{1}{1}\bigr)$ & $0$&$1$&$1$&$0$&$1$&$1$&$0$\\
\hline
\end{tabular}
\end{center}
Take the weights $(\lambda_1,\ldots,\lambda_6)=\bigl(\tfrac{33}{82},\,0,\,\tfrac{25}{82},\,\tfrac{15}{82},\,\tfrac{6}{82},\,\tfrac{3}{82}\bigr)$, which are non-negative and sum to $1$. By \eqref{eq:modq-master} it suffices to maximize $F(K)=\sum_{s}\lambda_s a_s(K)$, and reading down the columns of the table gives
\[
\begin{array}{c|ccccccc}
K & 0 & 1 & 2 & 3 & 4 & 5 & 6\\\hline
F(K) & 0 & \frac{18}{41} & \frac{27}{41} & \frac{27}{41} & \frac{27}{41} & \frac{27}{41} & \frac{27}{41}
\end{array}
\]
For example the column $K=2$:
\[
F(2)=\tfrac{33}{82}\cdot\tfrac{2}{6}+\tfrac{25}{82}\cdot\tfrac{16}{20}+\tfrac{15}{82}\cdot\tfrac{14}{15}+\tfrac{6}{82}\cdot\tfrac{6}{6}+\tfrac{3}{82}\cdot 1=\tfrac{11+20+14+6+3}{82}=\tfrac{54}{82}=\tfrac{27}{41},
\]
and the other columns check the same way. Hence $\max_K F(K)=\frac{27}{41}=\frac23-\frac{1}{123}$.
\end{proof}

With Lemma~\ref{lem:modq-monotone} this gives the bound $\frac{27}{41}$ for every $N\geq 6$, which is the case $q=3$ of Theorem~\ref{thm:modq-barrier}.

Six elements are also necessary, and here the argument runs the other way: instead of bounding every sieve we exhibit one. On five elements, let $\mathcal{H}$ draw a size $K\in\{2,3,4,5\}$ with probabilities $\tfrac{30}{121},\tfrac{40}{121},\tfrac{30}{121},\tfrac{21}{121}$ and then a uniformly random $T_h$ of that size. Every nonempty set then succeeds with probability at least $\tfrac{81}{121}$. Since $\tfrac{81}{121}>\tfrac23$, the value at $N=5$ still exceeds the ceiling, so $N=6$ is the first universe on which Theorem~\ref{thm:modq-barrier} can hold.

\subsection{Small universes}\label{app:modq-remarks}

Theorem~\ref{thm:modq-barrier} carries a hypothesis on $N$, and it cannot be dropped: on a universe too small to contain a set of size $q+1$, the ceiling is not merely unproved but false.

If $N\leq q-1$, the constant hash $h\equiv 1$ gives every nonempty $S$ the count $|S|\in\{1,\ldots,q-1\}$, which is never divisible by $q$. Every set succeeds always, and the value is $1$.

At $N=q$ the value is exactly $\frac{q}{q+1}$, which still exceeds $\frac{q-1}{q}$. For the lower bound, let $\mathcal{H}$ be uniform on the $q+1$ hash functions consisting of the all-ones function together with the $q$ functions $h_{\y}$ having $T_{h_{\y}}=U\setminus\{\y\}$. A singleton $\{\x\}$ fails only under $h_{\x}$, the full set $U$ fails only under the all-ones function, and any $S$ with $2\leq|S|\leq q-1$ has count in $\{|S|-1,|S|\}\subseteq\{1,\ldots,q-1\}$ and never fails. The minimum is thus $\frac{q}{q+1}$. For the matching upper bound, weights $\frac{q}{q+1}$ on the singletons and $\frac{1}{q+1}$ on $U$ give, through \eqref{eq:modq-master},
\[
F(K)=\frac{K+\mathbf{1}[K\notin\{0,q\}]}{q+1},
\]
whose maximum over $0\leq K\leq q$ is $\frac{q}{q+1}$, attained at $K=q-1$ and $K=q$.

So the ceiling fails exactly as long as the universe is too small to hold a set of size $q+1$, and holds from that point on. Since $N=2^{n}$ here, the regime that matters is the one Theorem~\ref{thm:modq-barrier} covers.

\end{document}